\documentclass[12pt, onecolumn, twoside]{IEEEtran}

\usepackage{graphicx}
\usepackage{amssymb}
\usepackage{amsmath}
\usepackage{dsfont}
\usepackage{url}
\usepackage{psfrag}
\usepackage{subfigure}

\newtheorem{theorem}{Theorem}
\newtheorem{cor}{Corollary}
\newtheorem{lem}{Lemma}

\newtheorem{defn}{Definition}

\newcommand{\Reals}{\mathds R}
\newcommand{\Integers}{\mathds Z}
\newcommand{\Naturals}{\mathds N}

\newcommand{\bx}{\mathbf{x}}

\newcommand{\by}{\mathbf{y}}
\newcommand{\bz}{\mathbf{z}}

\newcommand{\bZ}{\mathbf{Z}}

\newcommand{\Vor}{\mathcal{V}}
\newcommand{\SPB}{\textrm{SPB}}
\newcommand{\ball}[1]{\textrm{Ball}\!\left(#1\right)}
\newcommand{\ballx}[2]{\textrm{Ball}\!\left(#1,#2\right)}
\newcommand{\cube}[1]{\mathrm{Cb}\!\left(#1\right)}

\newcommand{\NLD}{{\boldsymbol\delta}}
\newcommand{\reff}{r_{\mathrm{eff}}}

\newcommand{\eps}{\varepsilon}

\newcommand{\E}{\mathbf{E}}

\newcommand{\Err}{\mathcal{E}}

\newcommand{\EE}{\mathbb{E}}

\newcommand{\VAR}{V\!AR}
\renewcommand{\S}{\mathcal{S}}
\newcommand{\ra}{\rightarrow}

\def\Intro{1}
\def\Defs{1}
\def\NewBounds{1}

\def\Properties{1}
\def\NormalApprox{1}
\def\Comparison{1}
\def\VNR{1}
\def\Summary{1}
\def\app{1}

\begin{document}

\title{Finite Dimensional Infinite Constellations}
\author{
\IEEEauthorblockN{Amir Ingber, Ram Zamir and Meir Feder\\}
\IEEEauthorblockA{School of Electrical Engineering,\\
Tel Aviv University \\
Tel Aviv 69978, Israel\\
Email: \{ingber, zamir, meir\}@eng.tau.ac.il}
\thanks{The material in this paper will be presented in part at the IEEE International Symposium on Information Theory (ISIT) 2011}%
\thanks{A. Ingber is supported by the Adams Fellowship Program of the Israel Academy of Sciences and Humanities.}%
\thanks{This research was supported in part by the Israel Science Foundation, grant no. 634/09.}

}

\maketitle
\markboth{Submitted to IEEE Transactions on Information Theory}{Ingber et al.: Finite Dimensional Infinite Constellations}

\begin{abstract}
In the setting of a Gaussian channel without power constraints, proposed by Poltyrev, the codewords are points in an $n$-dimensional Euclidean space (an infinite constellation) and the tradeoff between their {\em density} and the error probability is considered.
The capacity in this setting is the highest achievable normalized log density (NLD) with vanishing error probability. This capacity as well as error exponent bounds for this setting are known.
In this work we consider the optimal performance achievable in the fixed blocklength (dimension) regime.
We provide two new achievability bounds, and extend the validity of the sphere bound to finite dimensional infinite constellations. We also provide asymptotic analysis of the bounds: When the NLD is fixed, we provide asymptotic expansions for the bounds that are significantly tighter than the previously known error exponent results.
When the error probability is fixed, we show that as $n$ grows, the gap to capacity is inversely proportional (up to the first order) to the square-root of $n$ where the proportion constant is given by the inverse Q-function of the allowed error probability, times the square root of $\frac{1}{2}$.  In an analogy to similar result in channel coding, the dispersion of infinite constellations is $\frac{1}{2}\mathsf{nat}^2$ per channel use.
All our achievability results use lattices and therefore hold for the maximal error probability as well. Connections to the error exponent of the power constrained Gaussian channel and to the volume-to-noise ratio as a figure of merit are discussed. In addition, we demonstrate the tightness of the results numerically and compare to state-of-the-art coding schemes.
\end{abstract}


\section{Introduction}
\if \Intro 1

Coding schemes over the Gaussian channel are traditionally limited by the average/peak power of the transmitted signal \cite{ForneyUngerboeck98}. Without the power restriction (or a similar restriction) the channel capacity becomes infinite, since one can space the codewords arbitrarily far apart from each other and achieve a vanishing error probability. However, many coded modulation schemes take an infinite constellation (IC) and restrict the usage to points of the IC that lie within some $n$-dimensional form in Euclidean space (a `shaping' region). Probably the most important example for an IC is a lattice (see Fig.~\ref{fig:LatticeAndICs}), and examples for the shaping regions include a hypersphere in $n$ dimensions, and a Voronoi region of another lattice \cite{ErezZamirAWGN}.

\begin{figure}[t]
\begin{center}
\subfigure[A lattice]{\label{fig:latticePlot}%
\includegraphics[width=.49\textwidth]{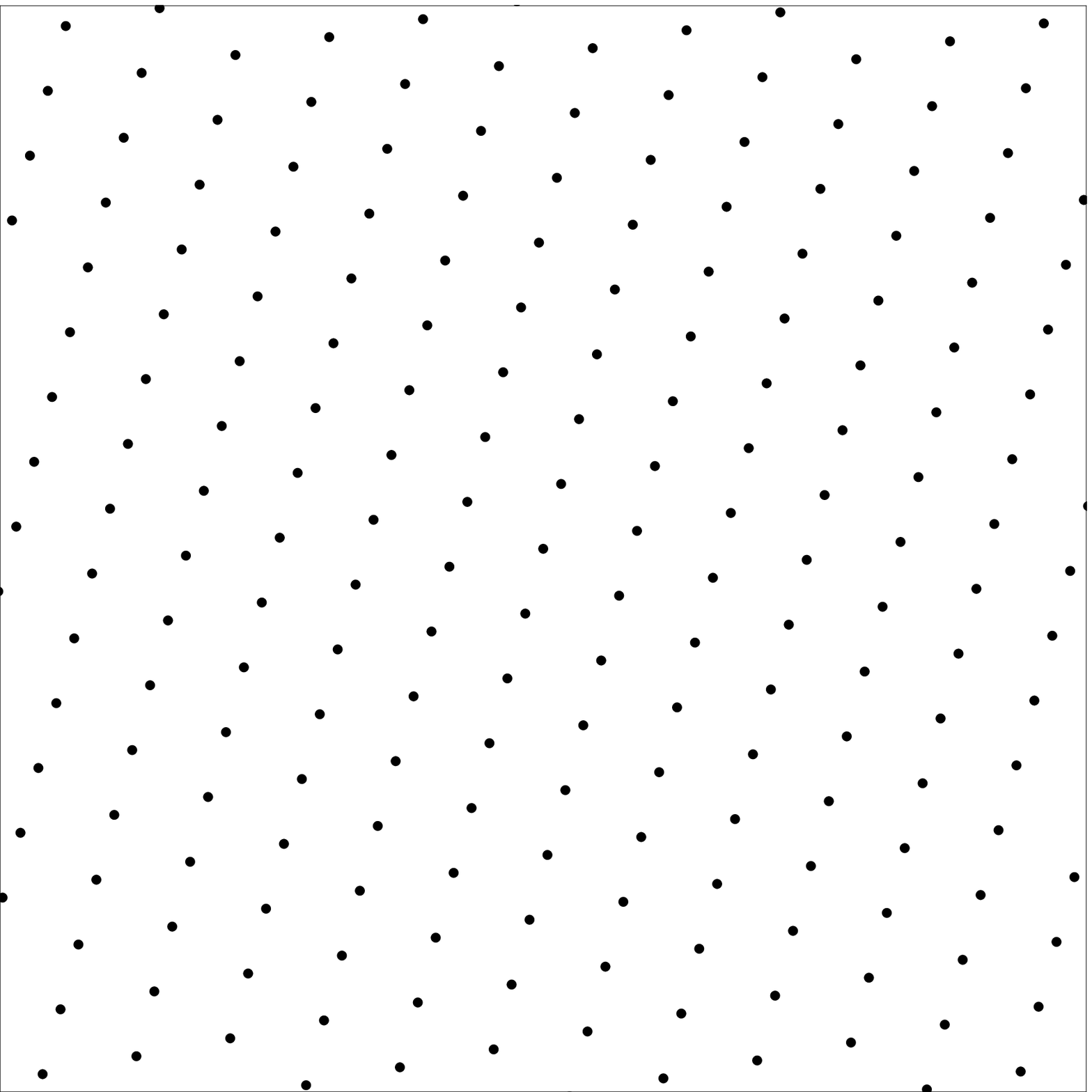}}
\subfigure[A non-lattice infinite constellation]{\label{fig:randomICPlot}%
\includegraphics[width=.49\textwidth]{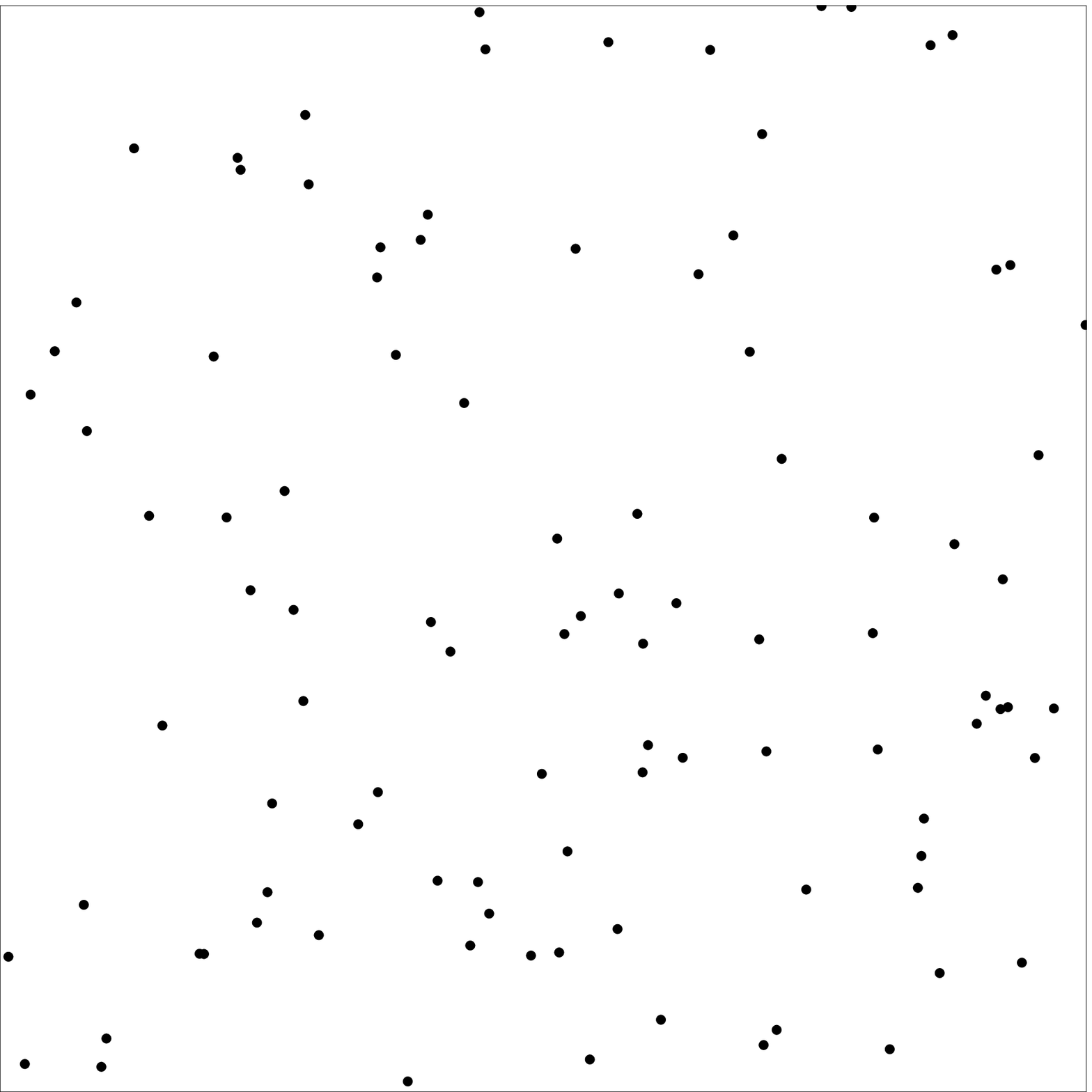}}
\caption{Examples for 2-dimensional infinite constellations. Only a finite section of the IC is shown.}
\label{fig:LatticeAndICs}
\end{center}
\end{figure}

In 1994, Poltyrev \cite{Poltyrev94_CodingWithoutRestrictions} studied the model of a channel with Gaussian noise without power constraints. In this setting the codewords are simply points of the infinite constellation in the $n$-dimensional Euclidean space. The analog to the number of codewords is the density $\gamma$ of the constellation points (the average number of points per unit volume). The analog of the communication rate is the normalized log density (NLD) $\NLD \triangleq \frac{1}{n}\log \gamma$.
The error probability in this setting can be thought of as the average error probability, where all the points of the IC have equal transmission probability (precise definitions follow later on in the paper).

Poltyrev showed that the NLD $\NLD$ is the analog of the rate in classical channel coding, and established the corresponding ``capacity'', the ultimate limit for the NLD denoted $\NLD^*$ (also known as Poltyrev's capacity), given by $\frac{1}{2}\log \frac{1}{2\pi e \sigma^2}$, where $\sigma^2$ denotes the noise variance per dimension\footnote{logarithms are taken w.r.t. to the natural base $e$ and rates are given in nats.}. Random coding, expurgation and sphere packing error exponent bounds were derived, which are analogous to Gallager's error exponents in the classical channel coding setting \cite{GallagerInfoTheoryBook}, and to the error exponents of the power-constrained additive white Gaussian noise (AWGN) channel \cite{Shannon59Gaussian,GallagerInfoTheoryBook}.

In classical channel coding, the channel capacity gives the ultimate limit for the rate when arbitrarily small error probability is required, and the error exponent quantifies the (exponential) speed at which the error probability goes to zero as the dimension grows, where the rate is fixed (and below the channel capacity). This type of analysis is asymptotic in nature - neither the capacity nor the error exponent theory can tell what is the best achievable error probability with a given rate $R$ and block length $n$. A big step in the non-asymptotic direction was recently made in a paper by Polyanskiy et al. \cite{PolyanskiyPVFiniteLength10}, where explicit bounds for finite $n$ were derived.
In addition to the error exponent formulation, another asymptotic question can be asked: Suppose that the (codeword) error probability is fixed to some value $\eps$. Let $R_\eps(n)$ denote the maximal rate for which there exist communication schemes with codelength $n$ and error probability at most $\eps$. As $n$ grows, $R_\eps(n)$ approaches the channel capacity $C$, and the speed of convergence is quantified by \cite{Strassen62_Asymptotische}\cite{PolyanskiyPVFiniteLength10}
\begin{equation}\label{eqn:DMCdispersion}
    R_\eps(n) = C - \sqrt \frac{V}{n}Q^{-1}(\eps) + O\left(\frac{\log n}{n}\right),
\end{equation}
where $Q^{-1}(\cdot)$ is the inverse complementary standard Gaussian cumulative distribution function. The constant $V$, termed the channel dispersion, is the variance of the information spectrum $i(x;y) \triangleq  \log \frac{P_{XY}(x,y)}{P_X(x) P_Y(y)} $ for a capacity-achieving distribution. This result holds for discrete memoryless channels (DMC's), and was recently extended to the (power constrained) AWGN channel \cite{PolyanskiyPV09_GaussianDispersion}\cite{PolyanskiyPVFiniteLength10}. More refinements of \eqref{eqn:DMCdispersion} and further details can be found in \cite{PolyanskiyPVFiniteLength10}.

\bigskip

In this paper we take an in-depth look at the unconstrained Gaussian channel where the block length (dimension) is finite. We give new achievability bounds which enable easy evaluation of the achievable error probability. We then analyze the new achievability bounds and the so-called sphere bound (converse bound), and obtain asymptotic analysis of the lowest achievable error probability for fixed NLD $\NLD$ which greatly refines Poltyrev's error exponent results. In addition, we analyze the behavior of the highest NLD when the error probability is fixed. We show that the behavior demonstrated in \eqref{eqn:DMCdispersion} for DMC's and the power constrained AWGN channel carries on to the unconstrained AWGN channel as well. We demonstrate the tightness of the results both analytically and numerically, and compare to state-of-the-art coding schemes.

The main results in the paper are summarized below.

\subsection{New Finite-Dimensional Performance Bounds}
Poltyrev's achievability results \cite{Poltyrev94_CodingWithoutRestrictions} for the capacity and for the error exponent are based on a bound that holds for finite dimensions, but is hard to calculate, as it involves optimizing w.r.t. a parameter and 3-dimensional integration. We derive two new bounds that hold for finite dimensions, and are easier to calculate than Poltyrev's. Like Poltyrev's bound, we bound the error probability by the sum of the probability that the noise leaves a certain region (a sphere), and the probability of error for noise realization within that sphere. This classic technique is due to Gallager \cite{GallagerLDPC}, sometimes called ``Gallager's first bounding technique'' \cite{DivsalarBounds}. Our first bound, called the \emph{typicality bound}, is based on a simple `typicality' decoder (close in spirit to that used in the standard achievability proofs \cite{CoverThomas_InfoTheoryBook}). It shows that there exist IC's with NLD $\NLD$ and error probability bounded by
\begin{equation}
   P_e \leq  P_e^{TB}(n,\NLD) \triangleq e^{n\NLD} V_n r^n + \Pr\left\{ \|\bZ\| > r \right\},
\end{equation}
where $V_n$ denotes the volume of an $n$-dimensional sphere with unit radius \cite{ConwaySloane1993} and $\bZ$ denotes the noise vector. The bound holds for any $r>0$, and the value minimizing the bound is given by $r = \sigma\sqrt{n (1 + 2\NLD^* -2\NLD)}$. Evaluating this bound only involves 1D integration, and the simple expression is amenable to precise asymptotic analysis.
A stronger bound, called the \emph{maximum likelihood (ML) bound}, which is based on the ML decoder, shows that there exist IC's with error probability bounded by
\begin{equation}\label{eqn:AchievabilityMaxLikeFirst}
  P_e \leq P_e^{MLB}(n,\NLD) \triangleq e^{n\NLD} V_n \int_0^r f_R(\tilde r) \tilde r^n d\tilde r + \Pr\left\{ \|\bZ\| > r \right\},
\end{equation}
$f_R(\cdot)$ is the pdf of the norm $\|\bZ\|$ of the noise vector. The bound holds for any $r>0$, and the value minimizing the bound is given by $r = \reff \triangleq e^{-\NLD}V_n^{-1/n}$. Note that $\reff$, called the \emph{effective radius} of the lattice (or IC), is the radius of a sphere with the same volume as the Voronoi cell of the lattice (or the average volume of the Voronoi cells of the IC\footnote{Note that the average volume of the Voronoi cells is not always well-defined, as in general there may exist cells with infinite volume. See \ref{ssec:SphereBoundIC} for more details.}
). Evaluating the ML bound also involves 1D integration only. We further show that the ML bound gives the exact value of Poltyrev's bound, therefore the simplicity does not come at the price of a weaker bound.

In the achievability part of the results we use lattices (and the Minkowski-Hlawka theorem \cite{HlawkaGeometric1991}\cite{Lekkerkerker87}). Because of the regular structure of lattices, all our achievability results hold in the stronger sense of maximal error probability. In the converse part we base our results on the sphere bound \cite{TarokhVardyZeger99_universal}\cite{Poltyrev94_CodingWithoutRestrictions}\cite{ForneySphereBound00}, i.e. on the fact that the error probability is lower bounded by the probability that the noise leaves a sphere with the same volume as a Voronoi cell.
For lattices (and more generally, for IC's with equal-volume Voronoi cells), it is given by
\begin{equation}\label{eqn:SphereBoundFirst}
    P_e \geq P_e^{SB}(n,\NLD) \triangleq \Pr\{\|\bZ\|>\reff\}.
\end{equation}
We extend the validity of the sphere bound to \emph{any} IC, and to the stronger sense of \emph{average} error probability. Therefore our results hold for both average and maximal error probability, and for any IC (lattice or not).

Note that since the optimal value for $r$ in the ML bound \eqref{eqn:AchievabilityMaxLikeFirst} is exactly $\reff$, the difference between the ML upper bound and the sphere packing lower bound is the left term in \eqref{eqn:AchievabilityMaxLikeFirst}. This fact enables a precise evaluation of the best achievable $P_e$, see Section~\ref{sec:Properties}.

\subsection{Asymptotic Analysis: Fixed NLD}

The asymptotics of the bounds on the error probability were studied by Poltyrev \cite{Poltyrev94_CodingWithoutRestrictions} using large deviation techniques and error exponents. The error exponent for the unconstrained AWGN is defined in the usual manner:
\begin{equation}\label{eqn:errExpDef}
    \E(\NLD) \triangleq \lim_{n\ra\infty} \frac{1}{n}\log P_e(n,\NLD),
\end{equation}
(assuming the limit exists), where $P_e(n,\NLD)$ is the best error probability for any IC with NLD $\NLD$. Poltyrev showed that the error exponent is bounded by the random coding and sphere packing exponents $\E_r(\NLD)$ and $\E_{sp}(\NLD)$ which are the infinite constellation counterparts of the similar exponents in the power constrained AWGN. The random coding and sphere packing exponents coincide when the NLD is above the critical NLD $\NLD_{cr}$, defined later on. However, even when the error exponent bounds coincide, the optimal error probability $P_e(n,\NLD)$ is known only up to an unknown sub-exponential term (which can be, for example $n^{100}$, or worse, e.g. $e^{\sqrt n}$). We present a significantly tighter asymptotic analysis using a more delicate (and direct) approach. Specifically, we show that the sphere bound is given asymptotically by
\begin{equation}
    P_e^{SB}(n,\NLD) \cong e^{-n E_{sp}(\NLD)} \frac{(n\pi)^{-\frac{1}{2}e^{2(\NLD^* - \NLD)}}}{e^{2(\NLD^* - \NLD)}-1},
\end{equation}
where $a \cong b$ means that $\frac{a}{b} \ra 1$.
We further show that the ML bound is given by
\begin{equation}
     P_e^{MLB}(n,\NLD)\cong \left\{
                     \begin{array}{ll}
                       e^{-n\E_r(\NLD)}\frac{1}{\sqrt{2 \pi n}}
                       , & \hbox{$\NLD < \NLD_{cr}$;} \\
                       e^{-n\E_r(\NLD)} \frac{1}{\sqrt{8\pi n}}
                       , & \hbox{$\NLD = \NLD_{cr}$;} \\
                       e^{-n\E_{r}(\NLD)}\frac{(n\pi)^{-\frac{1}{2}e^{2(\NLD^*-\NLD)}}}
    {\left(2-e^{2(\NLD^*-\NLD)}\right)\left(e^{2(\NLD^*-\NLD)}-1\right)},
                       & \hbox{$\NLD_{cr} < \NLD < \NLD^*  $;}
                     \end{array}
                   \right.
\end{equation}
As a consequence, for NLD above $\NLD_{cr}$, $P_e(n,\NLD)$ is known asymptotically up to \emph{a constant}, compared to a sub-exponential term in Poltyrev's error exponent analysis.
The weaker typicality bound is given by
\begin{equation}
P_e^{TB}(n,\NLD) \cong e^{-n \E_t(\NLD)}\frac{1}{\sqrt{n\pi}} \cdot \frac{1 + 2(\NLD^* - \NLD)}{2(\NLD^* - \NLD)}
\end{equation}
where $\E_t(\NLD)$ is the \emph{typicality exponent}, defined later on, which is lower than $\E_r(\NLD)$.

\subsection{Asymptotic Analysis: Fixed Error Probability}
For a fixed error probability value $\eps$, let $\NLD_\eps(n)$ denote the maximal NLD for which there exists an IC with dimension $n$ and error probability at most $\eps$. We shall be interested in the asymptotic behavior  $\NLD_\eps(n)$. This type of analysis for infinite constellations has never appeared in literature (to the best of the authors' knowledge). In the current paper we utilize central limit theorem (CLT) type tools (specifically, the Berry-Esseen theorem) to give a precise asymptotic analysis of $\NLD_\eps(n)$, a result analogous to the channel dispersion \cite{Strassen62_Asymptotische}\cite{PolyanskiyPV09_GaussianDispersion}\cite{PolyanskiyPVFiniteLength10} in channel coding. Specifically, we show that
\begin{equation}\label{eqn:mainResult}
    \NLD_\eps(n) = \NLD^* - \sqrt\frac{1}{2n}Q^{-1}(\eps) +\frac{1}{2n}\log n + O\left(\frac{1}{n}\right).
\end{equation}
By the similarity to Eq. \eqref{eqn:DMCdispersion}, we identify the constant $\frac{1}{2}$ as the dispersion of infinite constellations. This fact can be intuitively explained in several ways:
\begin{itemize}
  \item \emph{The dispersion as the (inverse of the) second derivative of the error exponent:} for DMC's and for the power constrained AWGN channel, the channel dispersion is given by the inverse of the second derivative of the error exponent evaluated at the capacity \cite{PolyanskiyPVFiniteLength10}. Straightforward differentiation of the error exponent $\E(\NLD)$ (which near the capacity is given by $\E_r(\NLD) = \E_{sp}(\NLD)$) verifies the value of $\frac{1}{2}$.
  \item \emph{The unconstrained AWGN channel as the high-SNR AWGN channel:} While the capacity of the power constrained AWGN channel grows without bound with the SNR, the error exponent attains a nontrivial limit. This limit is the error exponent of the unconstrained AWGN channel (as noticed in \cite{ErezZamirAWGN}), where the distance to capacity is replaced by the NLD distance to $\NLD^*$. By this analogy, we examine the high-SNR limit of the dispersion of the AWGN channel (given in \cite{PolyanskiyPV09_GaussianDispersion}\cite{PolyanskiyPVFiniteLength10} by $\tfrac{1}{2}\left(1 - (1+\mathrm{SNR})^{-2}\right)$) and arrive at the expected value of $\frac{1}{2}$.
\end{itemize}

\subsection{Volume-to-Noise Ratio (VNR)}
Another figure of merit for lattices (that can be defined for general IC's as well) is the volume-to-noise ratio (VNR), which generalizes the SNR notion \cite{ForneySphereBound00} (see also \cite{ZamirLatticesEverywhere09}). The VNR quantifies how good a lattice is for channel coding over the unconstrained AWGN at some given error probability $\eps$. It is known that for any $\eps>0$, the optimal (minimal) VNR of any lattice approaches $1$ when the dimension $n$ grows (see e.g. \cite{ZamirLatticesEverywhere09}). We note that the VNR and the NLD are tightly connected, and deduce equivalent finite-dimensional and asymptotic results for the optimal VNR.

\bigskip

The paper is organized as follows. In Section~\ref{sec:defs} we define the notations and in Section~\ref{sec:PreviousResults} we review previous results. In Section~\ref{sec:NewBounds} we derive the new typicality and ML bounds for the optimal error probability of finite dimensional IC's, and we refine the sphere bound as a lower bound on the average error probability for any finite dimensional IC. In Section~\ref{sec:Properties} the bounds are analyzed asymptotically with the dimension where the NLD is fixed, to derive asymptotic bounds that refine the error exponent bounds. In Section~\ref{sec:NormalApprox} we fix the error probability and study the asymptotic behavior of the optimal achievable NLD with $n$. We use normal approximation tools to derive the dispersion theorem for the setting.
In Section~\ref{sec:Comparison} we compare the bounds from previous sections with the performance of some good known infinite constellations.
In Section~\ref{sec:VNR} we discuss the VNR and its connection to the NLD $\NLD$. We conclude the paper in Section~\ref{sec:Summary}.

\fi

\section{Definitions}\label{sec:defs}
\if \Defs 1

\subsection{Notation} We adopt most of the notations of Poltyrev's paper \cite{Poltyrev94_CodingWithoutRestrictions}:
Let $\cube{a}$ denote a hypercube in $\Reals^n$
\begin{equation}
   \cube{a}\triangleq \left\{\bx\in\Reals^n\ s.t.\ \forall_i |x_i|<\frac{a}{2}\right\}.
\end{equation}
Let $\ball{r}$ denote a hypersphere in $\Reals^n$ and radius $r>0$, centered at the origin
\begin{equation}
   \ball{r}\triangleq \{\bx\in\Reals^n\ s.t.\  \|\bx\|<r\},
\end{equation}
and let $\ballx{\by}{r}$ denote a hypersphere in $\Reals^n$ and radius $r>0$, centered at $\by\in\Reals^n$
\begin{equation}
   \ballx{\by}{r}\triangleq \{\bx\in\Reals^n\ s.t.\  \|\bx-\by\|<r\}.
\end{equation}

Let $\S$ be an IC. We denote by $M(\S,a)$ the number of points in the intersection of  $\cube{a}$ and the IC $\S$, i.e.  $M(\S,a) \triangleq |\S \bigcap\cube{a}|$. The density of $\S$, denoted by $\gamma(\S)$, or simply $\gamma$, measured in points per volume unit, is defined by
\begin{equation}
    \gamma(\S) \triangleq \limsup_{a\ra\infty} \frac{M(\S,a)}{a^n}.
\end{equation}
The normalized log density (NLD) $\NLD$ is defined by
\begin{equation}
    \NLD = \NLD(\S) \triangleq \frac{1}{n} \log \gamma.
\end{equation}

It will prove useful to define the following:
\begin{defn}[Expectation over points in a hypercube]
Let $\EE_a[f(s)]$ denote the expectation of an arbitrary function $f(s)$, $f:\S\ra \Reals$, where $s$ is drawn uniformly from the code points that reside in the hypercube $\cube{a}$:
\begin{equation}
    \EE_a[f(s)] \triangleq \frac{1}{M(\S,a)} \sum_{s\in\S\cap \cube{a}} f(s).
\end{equation}
\end{defn}

\bigskip

Throughout the paper, an IC will be used for transmission of information through the unconstrained AWGN channel with noise variance $\sigma^2$ (per dimension). The additive noise shall be denoted by $\bZ = [Z_1,...,Z_n]^T$. An instantiation of the noise vector shall be denoted by $\bz = [z_1,...,z_n]^T$.

For $s\in\S$, let $P_e(s)$ denote the error probability when $s$ was transmitted. When the maximum likelihood (ML) decoder is used, the error probability is given by
\begin{equation}
    P_e(s) = \Pr\{ s + \bZ \notin W(s)\},
\end{equation}
where $W(s)$ is the \emph{Voronoi region} of $s$, i.e. the convex polytope of the points that are closer to $s$ than to any other point $s'\in\S$. The maximal error probability is defined by
\begin{equation}
    P_e^{\max}(\S) \triangleq \sup_{s \in \S} P_e(s),
\end{equation}
and the average error probability is defined by
\begin{equation}
    P_e(\S) \triangleq \limsup_{a \ra \infty} \EE_a[P_e(s)].
\end{equation}

The following related quantities, define the optimal performance limits for IC's.
\begin{defn}[Optimal Error Probability and Optimal NLD]
$ $

\begin{itemize}
  \item Given NLD value $\NLD$ and dimension $n$, $P_e(n,\NLD)$ denotes the optimal error probability that can be obtained by any IC with NLD $\NLD$ and a finite dimension $n$.
  \item Given an error probability value $\eps$ and dimension $n$, $\NLD_\eps(n)$ denotes the maximal NLD for which there exists an IC with dimension $n$ and error probability at most $\eps$.
  \end{itemize}
\end{defn}

Clearly, these two quantities are tightly connected, and any nonasymptotic bound for either quantity gives a bound for the other. However, their asymptotic analysis (with $n \ra \infty$) is different: for fixed $\NLD<\NLD^*$, it is known that $P_e(n,\NLD)$ vanishes exponentially with $n$. In this paper we will refine these results. For a fixed error probability $\eps$, it is known that $\NLD_\eps(n)$ goes to $\NLD^*$ when $n \ra \infty$. In this paper we will show that the gap to $\NLD^*$ vanishes like $O\left(1/\sqrt n\right)$, see Section~\ref{sec:NormalApprox}.

\bigskip

$f_n = O(g_n)$ shall mean that there exist a constant $c$ s.t. for all $n>n_0$ for some $n_0$, $|f_n|\leq c\cdot g_n$. Similarly, $f_n \leq O(g_n)$ shall mean that there exist $c,n_0$ s.t. for all $n>n_0$, $f_n\leq c\cdot g_n$. $f_n \geq O(g_n)$ means $-f_n \leq O(-g_n)$. $f_n = \Theta(g_n)$ shall mean that both $f_n = O(g_n)$ and $g_n = O(f_n)$ hold.
\subsection{Measuring the Gap from Capacity}\label{ssec:MeasuringGapFromCapacity}

Suppose we are given an IC $\S$ with a given density $\gamma$ (and NLD $\NLD = \frac{1}{n}\log\gamma$), used for information transmission over the unconstrained AWGN with noise variance $\sigma^2$. The gap from optimality can be quantified in several ways.

Knowing that the optimal NLD (for $n\ra\infty$) is $\NLD^*$, we may consider the difference
\begin{equation}\label{eqn:gapToCapacityDeltaNLD}
    \Delta\NLD = \NLD^*-\NLD,
\end{equation}
which gives the gap to capacity in $\mathsf{nats}$, where a zero gap means that we are working at capacity. An equivalent alternative would be to measure the ratio between the noise variance that is tolerable (in the capacity sense) with the given NLD $\NLD$, given by $\frac{e^{-2\NLD}}{2\pi e}$, and the actual noise variance $\sigma^2$ (equal to $\frac{e^{-2\NLD^*}}{2\pi e}$).
This ratio is given by
\begin{equation}\label{eqn:gapToCapacityRatioSigma}
    \mu \triangleq \frac{e^{-2\NLD}/(2\pi e)}{\sigma^2} = e^{2(\NLD^*-\NLD)}.
\end{equation}
For lattices, the term $e^{-2\NLD}$ is equal to $v^{2/n}$, where $v$ is the volume of a Voronoi cell of the lattice. Therefore $\mu$ was termed the \emph{Volume-to-Noise Ratio} (VNR) by Forney et al. \cite{ForneySphereBound00} (where it is denoted by $\alpha^2(\Lambda,\sigma^2)$). The VNR can be defined for general IC's as well. It is generally above $1$ (below capacity) and approaches $1$ at capacity. It is often expressed in dB
\footnote{For $\Delta\NLD$ measured in bits we would get the familiar 6.02 dB/bit instead of 8.6859 dB/nat in \eqref{eqn:gapToCapacityRatioSigmaDB}.}, i.e.
\begin{equation}\label{eqn:gapToCapacityRatioSigmaDB}
    10 \log_{10} \frac{e^{-2\NLD}/(2\pi e)}{\sigma^2} = 10 \log_{10} e^{2(\NLD^*-\NLD)} \cong 8.6859 \Delta\NLD.
\end{equation}
Note that the VNR appears under different names and scalings in the literature. Poltyrev \cite{Poltyrev94_CodingWithoutRestrictions} defined the quantity $\frac{e^{-2\NLD}}{\sigma^2}$ and called it the Generalized SNR (and also denoted it by $\mu$). In certain cases the latter definition is beneficial, as it can be viewed as the dual of the normalized second moment (NSM), which at $n\ra\infty$ approaches $\frac{1}{2\pi e}$ \cite{ZamirLatticesEverywhere09}.

\bigskip

An alternative way to quantify the gap from optimal performance is based on the fact that the Voronoi regions of an optimal IC (at $n \ra \infty$) becomes sphere-like. For example, the sphere bound (the converse bound) is based on a sphere with the same volume as the Voronoi cells of the IC (i.e. a sphere with radius $\reff$). As $n$ grows, the Voronoi regions of the optimal IC (that achieves capacity) becomes closer to a sphere with squared radius that is equal to the mean squared radius of the noise, $n\sigma^2$. Therefore a plausible way to measure the gap from optimality would be to measure the ratio between the squared effective radius of the IC and the expected squared noise amplitude, i.e.
\begin{equation}\label{eqn:gapToCapacityRatioRadii}
    \rho \triangleq \frac{\reff^2}{n\sigma^2} = \frac{e^{-2\NLD}V_n^{-2/n}}{ n\sigma^2}.
\end{equation}
This quantity was called ``Lattice SNR'' in \cite{TarokhVardyZeger99_universal}, and ``Voronoi-to-Noise Effective Radius Ratio'' (squared) in \cite{ErezLitsynZamirLatticesGoodForEverything}.
Similarly to the VNR $\mu$, this ratio also approaches $1$ at capacity, and is also often expressed in dB. However, the two measures \eqref{eqn:gapToCapacityRatioSigma} and \eqref{eqn:gapToCapacityRatioRadii} are not equivalent. For a given gap in dB, different IC densities (and NLD's) are derived, and only as $n \ra \infty$ the measures coincide (this can be seen by approximating $V_n$, see Appendix~\ref{app:Vn}). In the current paper, whenever we state a gap from capacity in dB, we refer to the gap \eqref{eqn:gapToCapacityRatioSigmaDB}.

\bigskip

In the current paper we shall be interested in the gap to capacity in the forms of \eqref{eqn:gapToCapacityDeltaNLD} and \eqref{eqn:gapToCapacityRatioSigma}. The finite-dimensional results in Section~\ref{sec:NewBounds} are specific for each $n$ and can be written as a function of either the NLD $\NLD$ or the ratio \eqref{eqn:gapToCapacityRatioRadii}. However, the asymptotic analysis in Sections \ref{sec:Properties} and \ref{sec:NormalApprox} depends on the selected measure. Specifically, in Section~ \ref{sec:Properties} we study the behavior of the error probability with $n\ra\infty$ where $\NLD$ is fixed. This is equivalent to fixing the ratio \eqref{eqn:gapToCapacityRatioSigma} (but not \eqref{eqn:gapToCapacityRatioRadii}). While the exponential behavior of the bounds on the error probability is the same whether we fix \eqref{eqn:gapToCapacityRatioSigma} or \eqref{eqn:gapToCapacityRatioRadii}, the sub-exponential behavior differs. In Section~\ref{sec:NormalApprox} we are interested in the behavior of the gap \eqref{eqn:gapToCapacityDeltaNLD} with $n\ra\infty$ for fixed error probability. Equivalent results in terms of the ratio  \eqref{eqn:gapToCapacityRatioRadii} can be derived using the same tools\footnote{It is interesting to note that although we choose to stick with the gap in nats and to the ratio \eqref{eqn:gapToCapacityRatioSigma}, the term \eqref{eqn:gapToCapacityRatioRadii} will pop out in the asymptotic analysis in Section~\ref{sec:Properties}.}.

\section{Previous Results}\label{sec:PreviousResults}
\subsection{Known Bounds on $P_e(n,\NLD)$}

Here we review existing non-asymptotic bounds on $P_e(n,\NLD)$, and discuss how easy are they for evaluation and asymptotic analysis.

The following non-asymptotic achievability bound can be distilled from Poltyrev's paper \cite{Poltyrev94_CodingWithoutRestrictions}:
\begin{theorem}[Poltyrev's achievability]\label{thm:AchievabilityPoltyrev}
For any $r>0$,
\begin{align}\label{eqn:AchievabilityPoltyrev}
    P_e(n,\NLD) \leq e^{n\NLD} n V_n \int_0^{2r} w^{n-1}\Pr\{\bZ \in D(r,w)\} dw + \Pr\{\|\bZ\|>r\},
\end{align}
where $D(r,w)$ denotes the section of the sphere with radius $r$ that is cut off by a hyperplane at a distance $\frac{w}{2}$ from the origin.
\end{theorem}

In \cite{Poltyrev94_CodingWithoutRestrictions} it is stated that the optimal value for $r$ (the one that minimizes the upper bound) is given by the solution to an integral equation, and it is shown that as $n \ra \infty$, the optimal $r$ satisfies $\frac{r^2}{n} \ra \sigma^2 e^{2(\NLD^* - \NLD)}$. However, no explicit expression for the optimal $r$ is given, so in order to compute the bound for finite values of $n$ one has to numerically optimize w.r.t. $r$ (in addition to the numerical integration). In order to derive the error exponent result, Poltyrev \cite{Poltyrev94_CodingWithoutRestrictions} used the asymptotic (but suboptimal)
$r = \sqrt n \sigma e^{\NLD^* - \NLD}$.

\bigskip

The converse bound used in \cite{Poltyrev94_CodingWithoutRestrictions}, which will be used in the current paper as well, is based on the following simple fact:
\begin{theorem}[Sphere bound]\label{thm:SphereBound}
Let $W(s)$ be the Voronoi region of an IC point $s$, and let $S_{W(s)}$ denote a sphere with the same volume as $W(s)$. Then the error probability $P_e(s)$ is lower bounded by
\begin{equation}\label{eqn:SphereBound1}
    P_e(s) \geq \Pr\{\bZ\notin S_{W(s)}\},
\end{equation}
where $\bZ$ denotes the noise vector.
\end{theorem}
This simple but important bound (see, e.g. \cite{TarokhVardyZeger99_universal}\cite{WozencraftJacobs}) is based on the fact that the pdf of the noise vector has spherical symmetry and decreases with the radius. An immediate corollary is the following bound for lattices (or more generally, any IC with equal-volume Voronoi cells):
\begin{equation}\label{eqn:SphereBoundConstantVoronoi}
    P_e(n,\NLD) \geq P_e^{SB}(n,\NLD) \triangleq \Pr\{\|\bZ\|>\reff\} = \int_{\reff}^{\infty} f_R(r') dr',
\end{equation}
where $\reff$ is the radius of a hypersphere with the same volume as a Voronoi cell, and $f_R(r)$ is the pdf of the norm of the noise vector, i.e. a (normalized) Chi distribution with $n$ degrees of freedom.

Note that this bound holds for any point $s$ in the IC, therefore it holds for the average error probability $P_e(n,\NLD)$ (and trivially for the maximal error probability as well).
In \cite{Poltyrev94_CodingWithoutRestrictions} the argument is extended to IC's which not necessarily obey the constant volume condition in the following manner: first, it is claimed that there must exist a Voronoi region with volume that is at less than the average volume $\gamma^{-1}$, so the bound holds for $P_e^{\max}(\S)$.
In order to apply the bound to the average error probability, a given IC $\S$ with average error probability $\eps$ is expurgated to get another IC $\S'$ with \emph{maximal} error probability  at most $2\eps$. Applying the previous argument for the maximal error probability of $\S'$ gives a bound on the average error probability of $\S$. The expurgation process, in addition to the factor of 2 in the error probability, also incurs a factor of 2 loss in the density $\gamma$. When evaluating the asymptotic exponential behavior of the error probability these factors have no meaning, but if we are interested (as in the case in this paper) in the bound values for finite $n$, and in the asymptotic behavior of $\NLD_\eps(n)$, these factors weaken the sphere bound significantly. In Section \ref{sec:NewBounds} we show that \eqref{eqn:SphereBoundConstantVoronoi} holds verbatim for any finite dimensional IC, and for the average error probability as well.

The sphere bound \eqref{eqn:SphereBoundConstantVoronoi} includes a simple (but with no known closed-form solution) 1D integral and can be evaluated numerically. An alternative for the numerical integration was proposed in \cite{TarokhVardyZeger99_universal}, where the integral was transformed into a sum of $n/2$ elements to allow the exact calculation of the bound. While the result gives an alternative to numeric integration, it does not shed any light on the asymptotic behavior of the bound with growing $n$.

\subsection{Known Asymptotic Bounds at Fixed $\NLD$ (Error Exponent)}

The error exponent $\E(\NLD)$ for the unconstrained AWGN was defined in \eqref{eqn:errExpDef}.
The nonasymptotic bounds in the previous subsection can lead to upper and lower bounds on the exponent.

The asymptotic evaluation of Poltyrev's achievability bound (Theorem~\ref{thm:AchievabilityPoltyrev}) is hard: in \cite{Poltyrev94_CodingWithoutRestrictions}, in order to provide a lower bound on the error exponent, a suboptimal value for $r$ is chosen for finite $n$  $\left(r = \sqrt n \sigma e^{-(\NLD^* - \NLD)}
\right)$. The resulting bound is the random coding exponent for this setting $\E_r(\NLD)$, given by
\begin{equation}\label{eqn:Er}
    \E_r(\NLD) = \left\{
                     \begin{array}{ll}
                       \NLD^*  -\NLD+ \log \frac{e}{4}, & \hbox{$\NLD \leq \NLD_{cr}$;} \\
                       \frac{1}{2}\left[e^{2(\NLD^*-\NLD)} -1 -2(\NLD^* -\NLD) \right], & \hbox{$\NLD_{cr} \leq \NLD < \NLD^*  $;} \\
                       0, & \hbox{$\NLD \geq \NLD^*,$}
                     \end{array}
                   \right.
\end{equation}
where $\NLD_{cr} = \frac{1}{2}\log\frac{1}{4 \pi e \sigma^2}$. Poltyrev also provided an expurgation-type argument to improve the error exponent at low NLD values (below $\NLD_{ex} \triangleq \NLD^*-\log 2$). This NLD region is outside the focus of the current paper.

An upper bound on the error exponent is the sphere packing exponent. It is given by \cite{Poltyrev94_CodingWithoutRestrictions}:
\begin{align}
    \E_{sp}(\NLD)
    &= \frac{1}{2}\left[e^{2(\NLD^*-\NLD)} -1 -2(\NLD^* -\NLD) \right],\label{eqn:Esp}
\end{align}
which is derived from the sphere bound (see \cite[Appendix C]{Poltyrev94_CodingWithoutRestrictions}).

The upper and lower bounds on the error exponent only hint on the value of $P_e(n,\NLD)$:
\begin{equation}
    e^{-n(\E_{sp}(\NLD) + o(1))} \leq P_e(n,\NLD) \leq e^{-n(\E_{r}(\NLD) + o(1))}.
\end{equation}
Even when the error exponent bounds coincide (above the critical NLD $\NLD_{cr}$), the optimal error probability $P_e(n,\NLD)$ is known only up to an unknown sub-exponential term. In Section~\ref{sec:Properties} we present a significantly tighter asymptotic analysis and show, for example, that at NLD above $\NLD_{cr}$, $P_e(n,\NLD)$ is known, asymptotically, up to \emph{a constant}.

\fi

\section{Bounds for Finite Dimensional IC's}\label{sec:NewBounds}

\if \NewBounds 1
In this section we analyze the optimal performance of finite dimensional infinite constellations in Gaussian noise. We describe two new achievability bounds, both based on lattices: The first bound is based on a simple `typicality' decoder, and the second one based on the ML decoder. Both bounds result in simpler expressions than Poltyrev's bound (Theorem \ref{thm:AchievabilityPoltyrev}). The first bound is simpler to derive but proves to be weak. The second bound gives the exact value of the bound as Poltyrev's (Theorem \ref{thm:AchievabilityPoltyrev}), without the need for 3D integration and an additional numeric optimization, but only a single 1D integral (which can be analyzed further - see Section~\ref{sec:Properties}). As for converse bounds, we extend the validity of the sphere bound to the most general case of IC's (not only those with equal-volume Voronoi cells) and average error probability.

\subsection{Typicality Decoder Based Bound}
\begin{theorem}\label{thm:AchievabilityTypicality}
For any $r>0$,
\begin{equation}
  P_e(n,\NLD) \leq P_e^{TB} \triangleq e^{n\NLD} V_n r^n + \Pr\left\{ \|\bZ\| > r \right\},
\end{equation}
and the optimal value for $r$ is given by
\begin{equation}
    r^* = \sigma\sqrt{n (1 + 2\NLD^* -2\NLD)}.
\end{equation}
\begin{proof}
Let $\Lambda$ be a lattice that is used as an IC for transmission over the unconstrained AWGN. We consider a suboptimal decoder, and therefore the performance of the optimal ML decoder can only be better. The decoder, called a \emph{typicality decoder}, shall operate as follows.
Suppose that $\lambda \in \Lambda$ is sent, and the point $\by = \lambda + \bz$ is received, where $\bz$ is the additive noise. Let $r$ be a parameter for the decoder, which will be determined later on. If there is only a single point in the ball $\ballx{\by}{r}$, then this will be the decoded word. If there are no codewords in the ball, or more than one codeword in the ball, an error is declared (one of the code points is chosen at random).
\begin{lem} The average error probability of a lattice $\Lambda$ (with the typicality decoder) is bounded by
\begin{equation}\label{eqn:Balls}
    P_e(\Lambda) \leq
\Pr \left\{ \bZ \notin \ball{r}\right\}
     + \sum_{\lambda \in \Lambda \setminus \{0\}} \Pr \left\{ \bZ \in \ballx{\lambda}{r}\cap \ball{r}\right\},
\end{equation}
where $\bZ$ denotes the noise vector.
\begin{proof}
Since $\Lambda$ is a lattice we can assume without loss of generality that the zero point was sent. We divide the error events to two cases. First, if the noise falls outside the ball of radius $r$ (centered at the origin), then there surely will be erroneous decoding since the transmitted (0) point is outside the ball. The remaining error cases are where the noise $\bZ$ is within $\ball{r}$, and the noise falls in the typical ball of some other lattice point (that is different than the transmitted zero point). We therefore get

\begin{align}
  P_e(\Lambda)  &\leq \Pr \left\{ \bZ \notin \ball{r}\right\}  + \Pr\left\{ \bZ \in \ball{r} \bigcap \left[ \bigcup_{\lambda \in \Lambda \setminus \{0\}}\ballx{\lambda}{r}\right]\right\}\nonumber\\
  &=\Pr \left\{ \bZ \notin \ball{r}\right\}  + \Pr\left\{ \bZ \in  \bigcup_{\lambda \in \Lambda \setminus \{0\}}\ballx{\lambda}{r} \cap \ball{r}\right\}\nonumber\\
  &\leq\Pr \left\{ \bZ \notin \ball{r}\right\}  + \sum_{\lambda \in \Lambda \setminus \{0\}}
  \Pr\left\{ \bZ \in \ballx{\lambda}{r} \cap \ball{r}\right\},\label{eqn:pre-MH}
\end{align}
where the last inequality follows from the union bound.
\end{proof}
\end{lem}

We use the Minkowski-Hlawka theorem \cite{Lekkerkerker87}\cite{HlawkaGeometric1991}:
\footnote{
The MH theorem is usually written as \eqref{eqn:MH} with an $\epsilon$ added to the RHS that is arbitrarily small (e.g. \cite[Lemma 3, p. 65]{HlawkaGeometric1991}, and \cite[Theorem 1, p. 200]{Lekkerkerker87}).
The version \eqref{eqn:MH} follows from a slightly improved version of the theorem due to Siegel, often called the Minkowski-Hlawka-Siegel (MHS) theorem, see \cite[Theorem 5, p. 205]{Lekkerkerker87}.
}
\begin{theorem}[MH]\label{thm:MH}
Let $f: \Reals^n \ra \Reals^+$ be a nonnegative integrable function with bounded support. Then for every $\gamma>0$, there exist a lattice $\Lambda$ with $\det \Lambda = \gamma^{-1}$ that satisfies
\begin{equation}\label{eqn:MH}
    \sum_{\lambda\in \Lambda \setminus \{0\}} f(\lambda) \leq \gamma\int_{\Reals^n} f(\lambda) d\lambda.
\end{equation}
\end{theorem}

Since $\Pr\left\{ \bZ \in \ballx{\lambda}{r} \cap \ball{r}\right\} = 0$ for any $\lambda$ s.t. $\|\lambda\| > 2r$ we may apply the MH theorem to the sum in \eqref{eqn:pre-MH}. We deduce that for any $\gamma>0$, there must exist a lattice $\Lambda$ with density $\gamma$, s.t.
\begin{equation}
\sum_{\lambda \in \Lambda \setminus \{0\}}\Pr\left\{ \bZ \in \ballx{\lambda}{r} \cap \ball{r}\right\}
\leq  \gamma\int_{\Reals^n} \Pr\left\{ \bZ \in \ballx{\lambda}{r} \cap \ball{r}\right\}d\lambda. \label{eqn:MH_application_typicality}
\end{equation}
We further examine the resulting integral:
\begin{align}
  \int_{\Reals^n} &\Pr\left\{ \bZ \in \ballx{\lambda}{r} \cap \ball{r}\right\}d\lambda\nonumber\\
  &=  \int_{\Reals^n} \int_{\ballx{\lambda}{r} \cap \ball{r}} f_\bZ(\bz) d\bz d\lambda\nonumber\\
  &\leq  \int_{\Reals^n} \int_{\ballx{\lambda}{r}} f_\bZ(\bz) d\bz d\lambda\nonumber\\
  &=  \int_{\Reals^n} \int_{\ball{r}} f_\bZ(\bz'+\lambda) d\bz' d\lambda\nonumber\\
  &=  \int_{\ball{r}} 1 d\bz'\nonumber\\
  &=  V_n r^n.
\end{align}

Combined with \eqref{eqn:Balls} we get that there exist a lattice $\Lambda$ with density $\gamma$, for which
\begin{align}\label{eqn:pe_lattice}
  P_e(\Lambda)  &\leq \gamma V_n r^n + \Pr\left\{ \|\bZ\| > r \right\},
\end{align}
where $r>0$ and $\gamma=e^{n\NLD}$ can be chosen arbitrarily.

The optimal value for $r$ follows from straightforward optimization of the RHS of \eqref{eqn:pe_lattice}: we first write
\begin{align*}
    \Pr\left\{ \|\bZ\| > r \right\}
    &= \Pr\left\{ \frac{1}{\sigma^2}\sum_{i=1}^nZ_i^2 > \frac{r^2}{\sigma^2} \right\}.
\end{align*}
We note that the sum $\frac{1}{\sigma^2}\sum_{i=1}^nZ_i^2$ is a sum of $n$ i.i.d. standard Gaussian RV's, which is exactly a $\chi^2$ random variable with $n$ degrees of freedom. The pdf of this RV is well known, and given by
\begin{equation*}
    f_{\chi^2_n}(x) = \frac{2^{-n/2}}{\Gamma(n/2)} x^{n/2 -1} e^{-x/2},
\end{equation*}
where $\Gamma(\cdot)$ is the Gamma function. Equipped with this, the RHS of \eqref{eqn:pe_lattice} becomes
\begin{equation*}
    e^{n\NLD}V_n r^n + \int_{\tfrac{r^2}{\sigma^2}}^\infty \frac{2^{-n/2}}{\Gamma(n/2)} x^{n/2 -1} e^{-x/2}.
\end{equation*}
Differentiating w.r.t. $r$ and equating to zero gives
$$
   n e^{n\NLD}V_n r^{n-1} -\frac{2r}{\sigma^2}
    \frac{2^{-n/2}}{\Gamma(n/2)} (r^2/\sigma^2)^{n/2 -1} e^{-\tfrac{r^2}{2\sigma^2}}=0.
    $$
We plug in the expression for $V_n = \frac{\pi^{n/2}}{\frac{n}{2}\Gamma(n/2)}$ and get
$$
   n e^{n\NLD}\frac{\pi^{n/2}}{\frac{n}{2}\Gamma(n/2)} r^{n-1} -\frac{2r}{\sigma^2}
    \frac{2^{-n/2}}{\Gamma(n/2)} (r^2/\sigma^2)^{n/2 -1} e^{-\tfrac{r^2}{2\sigma^2}}=0,$$
which simplifies to the required $r = \sigma\sqrt{n(1+2\NLD^*-2\NLD)}$.
\end{proof}
\end{theorem}

\subsection{ML Decoder Based Bound}

The second achievability bound is based on the ML decoder (using a different technique than Poltyrev \cite{Poltyrev94_CodingWithoutRestrictions}):

\begin{theorem}\label{thm:AchievabilityMaxLike}
For any $r>0$ and dimension $n$, there exist a lattice $\Lambda$ with error probability
\begin{equation}\label{eqn:AchievabilityMaxLike}
  P_e(n,\NLD) \leq P_e^{MLB}(n,\NLD) \triangleq e^{n\NLD} V_n \int_0^r f_R(\tilde r) \tilde r^n d\tilde r + \Pr\left\{ \|\bZ\| > r \right\},
\end{equation}

and the optimal value for $r$ is given by
\begin{equation}
    r^* = \reff = e^{-\NLD}V_n^{-1/n}.
\end{equation}
\end{theorem}
Before the proof, note that this specific value for $r$ gives a new interpretation to the bound: the term $\Pr\left\{ \|\bZ\| > r \right\}$ is exactly the sphere bound \eqref{eqn:SphereBound1}, and the other term can be thought of as a `redundancy' term. Making this value small results in tightening of the gap between the bounds.

\begin{proof}
    Suppose that the zero lattice point was sent, and the noise vector is $\bz\in\Reals^n$. An error event occurs (for a ML decoder) when there is a nonzero lattice point $\lambda \in \Lambda$ whose Euclidean distance to $\bz$ is less than the distance between the zero point and noise vector.
    We denote by $\Err$ the error event, condition on the radius $R$ of the noise vector and get
    \begin{align}
        P_e(\Lambda) &= \Pr\{\Err\} = \nonumber\\
        &= \EE_R\left[\Pr\left\{\Err \mid\|\bZ\|=R\right\}\right]\nonumber\\
        &= \int_0^\infty f_R(r) \Pr\left\{\Err \mid\|\bZ\|=r\right\} dr \nonumber\\
        &\leq \int_0^{r^*} f_R(r) \Pr\left\{\Err \mid\|\bZ\|=r\right\} dr + \Pr\{\|\bZ\|>r^*\},\label{eqn:new_bound_pe1}
    \end{align}
    where the last inequality follows by upper bounding the probability by $1$. It holds for any $r^*>0$. %

    We examine the conditional error probability $\Pr\left\{\Err \mid\|\bz\|=r\right\}$:
   \begin{align}
    \Pr\left\{\Err \mid\|\bZ\|=r\right\}&=\Pr\left\{\bigcup_{\lambda \in \Lambda \setminus \{0\}} \|\bZ-\lambda\| \leq \|\bZ\|\ \middle|\  \|\bZ\|=r\right\}\nonumber\\
    &\leq \sum_{\lambda \in \Lambda \setminus \{0\}}  \Pr\left\{\|\bZ-\lambda\| \leq \|\bZ\|\ \middle|\  \|\bZ\|=r\right\}\nonumber\\
    &= \sum_{\lambda \in \Lambda \setminus \{0\}}  \Pr\left\{\lambda \in \ballx{\bZ}{\|\bZ\|}\ \middle|\  \|\bZ\|=r\right\},
   \end{align}
   where the inequality follows from the union bound. Plugging into the left term in \eqref{eqn:new_bound_pe1} gives
   \begin{align}
        &\int_0^{r^*} f_R(r)\sum_{\lambda \in \Lambda \setminus \{0\}}  \Pr\left\{\lambda \in \ballx{\bZ}{\|\bZ\|}\ \middle|\  \|\bZ\|=r\right\}dr \nonumber\\
        =& \sum_{\lambda \in \Lambda \setminus \{0\}}   \int_0^{r^*} f_R(r) \Pr\left\{\lambda \in \ballx{\bZ}{\|\bZ\|}\ \middle|\  \|\bZ\|=r\right\}dr.
   \end{align}

    Note that the last integral has a bounded support (w.r.t. $\lambda$) - it is always zero if $\|\lambda\|\geq 2r^*$. Therefore we can apply the Minkowski-Hlawka theorem as in Theorem \ref{thm:AchievabilityTypicality} and get that for any $\gamma>0$ there exists a lattice $\Lambda$ with density $\gamma$, whose error probability is upper bounded by
    \begin{align*}
        P_e(\Lambda) & \leq \gamma\int_{\lambda\in\Reals^n} \int_0^{r^*} f_R(r) \Pr\left\{\lambda \in \ballx{\bZ}{\|\bZ\|}\ \middle|\  \|\bZ\|=r\right\}drd\lambda +\Pr\{\|\bZ\|>r^*\}.
    \end{align*}
    We continue with
\begin{align*}
    &\int_{\lambda\in\Reals^n} \int_0^{r^*} f_R(r) \Pr\left\{\lambda \in \ballx{\bZ}{\|\bZ\|}\ \middle|\  \|\bZ\|=r\right\}drd\lambda\\
    &= \int_0^{r^*} f_R(r) \int_{\lambda\in\Reals^n} \Pr\left\{\lambda \in \ballx{\bZ}{\|\bZ\|}\ \middle|\  \|\bZ\|=r\right\}d\lambda dr\\
    &= \int_0^{r^*} f_R(r) \int_{\lambda\in\Reals^n}
    \EE\left[1_{\{\lambda \in \ballx{\bZ}{\|\bZ\|}\}}\ \middle|\  \|\bZ\|=r\right] d\lambda dr\\
    &= \int_0^{r^*} f_R(r) \EE\left[\int_{\lambda\in\Reals^n} 1_{\{\lambda \in \ballx{\bZ}{\|\bZ\|}\}}d\lambda\ \middle|\  \|\bZ\|=r\right]  dr\\
    &= \int_0^{r^*} f_R(r) \EE\left[\|\bZ\|^n V_n \middle|\  \|\bZ\|=r\right]  dr\\
    &=  V_n \int_0^{r^*} f_R(r) r^n dr,
\end{align*}
and we obtain \eqref{eqn:AchievabilityMaxLike}.

To find the optimal value for $r$ (the one that minimizes the RHS of \eqref{eqn:AchievabilityMaxLike}), we see that:
\begin{equation}
    \Pr\left\{ \|\bZ\| > r \right\} = \int_r^\infty f_R(\tilde r) d\tilde r.
\end{equation}
Differentiating the RHS of \eqref{eqn:AchievabilityMaxLike} w.r.t. $r$ in order to find the minimum gives
\begin{equation}
    e^{n\NLD}V_n f_R(r) r^n - f_R(r) =0,
\end{equation}
and $r^* = \reff = e^{-\NLD}V_n^{-1/n}$ immediately follows.
\end{proof}

\subsection{Equivalence of the ML bound with Poltyrev's bound}
In Theorems \ref{thm:AchievabilityTypicality} and \ref{thm:AchievabilityMaxLike} we provided a new upper bounds on the error probability that were simpler than Poltyrev's original bound (Theorem \ref{thm:AchievabilityPoltyrev}). For example, in order to compute Poltyrev's bound, one has to apply 3D numerical integration, and numerically optimize w.r.t. $r$. In contrast, both new bounds requires only a single integration, and the optimal value for $r$ has a closed-form expression so no numerical optimization is required.

It appears that the simplicity of the bound in Theorem \ref{thm:AchievabilityMaxLike} does not come at a price of a weaker bound. In fact, it proves to be equivalent to Poltyrev's bound:
\begin{theorem}\label{thm:Equivalence}
    Poltyrev's bound (Theorem \ref{thm:AchievabilityPoltyrev}) for the error probability, for the optimal value of $r$, is equal to the ML bound from Theorem \ref{thm:AchievabilityMaxLike}:
    \begin{align}
        &\min_{r>0} \left\{e^{n\NLD} n V_n \int_0^{2r} w^{n-1}\Pr\{\bZ \in D(r,w)\} dw + \Pr\{\|\bZ\|>r\}\right\}\nonumber\\
        =&e^{n\NLD} V_n \int_0^{r^*} f_R(\rho) \rho^n d\rho + \Pr\left\{ \|\bZ\| > r^* \right\},\label{eqn:BoundsEquivalence}
    \end{align}
where $r^* = \reff = e^{-\NLD}V_n^{-1/n}$.

In fact, we can strengthen \eqref{eqn:BoundsEquivalence} and show that
\begin{align}
    \gamma n V_n \int_0^{2r} w^{n-1}\Pr\{\bZ \in D(r,w)\} dw  = \gamma V_n \int_0^r f_R(\rho) \rho^n d\rho \label{eqn:BoundsEquivalenceStrong}
\end{align}
for any $r>0$.
\begin{proof} Appendix~\ref{app:Equivalence}.\end{proof}
\end{theorem}
Note that proving \eqref{eqn:BoundsEquivalenceStrong} shows that both bounds are equivalent, regardless of the value of $r$. Consequently, the optimal value for $r$ in Poltyrev's bound is also found. In \cite{Poltyrev94_CodingWithoutRestrictions} the optimal value (denoted there $\overline d_c^*(n,\delta)$) was given as the solution to an integral equation, and was only evaluated asymptotically.

\subsection{The Sphere Bound for Finite Dimensional Infinite Constellations}\label{ssec:SphereBoundIC}

The sphere bound \eqref{eqn:SphereBoundConstantVoronoi} applies to infinite constellations with fixed Voronoi cell volume. Poltyrev \cite{Poltyrev94_CodingWithoutRestrictions} extended it to general IC's with the aid of an \emph{expurgation} process, without harming the tightness of the error exponent bound. When the dimension $n$ is finite, the expurgation process incurs a non-negligible loss (a factor of 2 in the error probability and in the density). In this section we show that the sphere bound applies \emph{without any loss} to general finite dimensional IC's and average error probability.

We first concentrate on IC's with some mild regularity assumptions:
\begin{defn}[Regular IC's]\label{def:regularIC}
An IC $\S$ is called \emph{regular}, if:
\begin{enumerate}
  \item There exists a radius $r_0>0$, s.t. for all $s\in\S$, the Voronoi cell $W(s)$ is contained in $\ballx{s}{r_0}$.
  \item The density $\gamma(\S)$ is given by $\lim_{a\ra\infty} \frac{M(\S,a)}{a^n}$ (rather than $\limsup$ in the original definition).
\end{enumerate}
\end{defn}

For $s\in\S$, we denote by $v(s)$ the volume of the Voronoi cell of $s$, $|W(s)|$.

\begin{defn}[Average Voronoi cell volume]
For a regular IC $\S$, the average Voronoi cell volume is defined by
\begin{equation}
    v(\S) \triangleq \limsup_{a\ra\infty} \EE_a[ v(s)].
\end{equation}
\end{defn}

\begin{lem}\label{lem:avgVol}
For a regular IC $\S$, the average volume is given by the inverse of the density:
\begin{equation}
    \gamma(\S) = \frac{1}{v(\S)}.
\end{equation}
\begin{proof} Appendix \ref{app:AvgVol}.
\end{proof}
\end{lem}

For brevity, let $\SPB(v)$ denote the probability that the noise vector $\bZ$ leaves a sphere of volume $v$. With this notation, the sphere bound reads
\begin{equation}
    P_e(s) \geq \SPB(v(s)),
\end{equation}
and holds for any individual point $s\in \S$. We also note the following:
\begin{lem}\label{lem:SPBconvexity}
The equivalent sphere bound $\SPB(v)$ is convex in the Voronoi cell volume $v$.
\proof Appendix \ref{app:SPBconvexity}.
\end{lem}

We now show that the above equation holds for the average volume and error probability as well.
\begin{theorem}\label{thm:ConverseRegular}
Let $\S$ be a regular (finite dimensional) IC with NLD $\NLD$, and let $v(\S)$ be the average Voronoi cell volume of $\S$ (so the density of $\S$ is $\gamma = v(\S)^{-1}$). Then the average error probability of $\S$ is lower bounded by
\begin{equation}
    P_e(\S) \geq \SPB(v(\S)) = \SPB(\gamma^{-1}) = P_e^{SB}(n,\NLD).
\end{equation}
\begin{proof}
We start with the definition of the average error probability and get
\begin{align}
P_e(S) &= \limsup_{a\ra\infty} \EE_a[P_e(s)] \nonumber\\
 &\overset{(a)}{\geq} \limsup_{a\ra\infty}\EE_a[\SPB(v(s))]\nonumber\\
 &\overset{(b)}{\geq} \limsup_{a\ra\infty} \SPB(\EE_a[v(s)])\nonumber\\
 &\overset{(c)}{=} \SPB(\limsup_{a\ra\infty} \EE_a[v(s)])\nonumber\\
 &= \SPB(v(\S)).
\end{align}

$(a)$ follows from the sphere bound for each individual point $s\in\S$,
$(b)$ follows from the Jensen inequality and the convexity of $\SPB(\cdot)$ (Lemma \ref{lem:SPBconvexity}), and
$(c)$ follows from the fact that $\SPB(\cdot)$ is continuous.
\end{proof}
\end{theorem}
As a consequence, we get that the sphere bound holds for regular IC's as well, without the need for expurgation (as in \cite{Poltyrev94_CodingWithoutRestrictions}).

\bigskip

So far the discussion was constrained to regular IC's only. This excludes constellations with infinite Voronoi regions (e.g. contains points only in half of the space), and also constellations in which the density oscillates with the cube size $a$ (and the formal limit $\gamma$ does not exist). We now extend the proof of the converse for any IC, without the regularity assumptions. The proof is based on the following regularization process:
\begin{lem}[Regularization]\label{lem:regularization}
Let $\S$ be an IC with density $\gamma$ and average error probability $P_e(\S)=\eps$. Then for any $\xi>0$ there exists a \emph{regular} IC $\S'$ with density $\gamma' \geq \gamma / (1+\xi)$, and average error probability $P_e(\S')=\eps'\leq \eps(1+\xi)$.
\begin{proof} Appendix \ref{app:regularization}.
\end{proof}
\end{lem}

\begin{theorem}[Sphere Bound for Finite Dimensional IC's]\label{thm:ConverseFiniteDimIC}
Let $\S$ be a finite dimensional IC with density $\gamma$. Then the average error probability of $\S$ is lower bounded by
\begin{equation}\label{eqn:ConverseFiniteDimIC}
    P_e(\S)  \geq \SPB(\gamma^{-1}) = P_e^{SB}(n,\NLD)
\end{equation}

\begin{proof}
Let $\xi>0$. By the regularization lemma (Lemma~\ref{lem:regularization}) there exists a regular IC $\S'$ with $\gamma' \geq \gamma / (1+\xi)$, and $P_e(\S')\leq P_e(\S)(1+\xi)$. We apply Theorem~\ref{thm:ConverseRegular} to $\S'$ and get that
\begin{equation}
     P_e(\S)(1+\xi) \geq P_e(\S') \geq \SPB(\gamma'^{-1}) \geq \SPB((1+\xi)\gamma^{-1}),
\end{equation}
or
\begin{equation}
     P_e(\S) \geq \frac{1}{1+\xi}\SPB((1+\xi)\gamma^{-1}),
\end{equation}
for all $\xi>0$. Since $\SPB(\cdot)$ is continuous, we may take the limit $\xi \ra 0$ and get to \eqref{eqn:ConverseFiniteDimIC}.
\end{proof}
\end{theorem}

\subsection{Numerical Comparison}
Here we numerically compare the bounds in this section with Poltyrev's achievability bound (Theorem~\ref{thm:AchievabilityPoltyrev}). As shown in the previous subsection, the bounds in Theorems \ref{thm:AchievabilityPoltyrev} and \ref{thm:AchievabilityMaxLike} are equivalent. However, as discussed following the statement of Theorem~\ref{thm:AchievabilityPoltyrev} above, in \cite{Poltyrev94_CodingWithoutRestrictions} the suboptimal value for $r$ is used.

We therefore refer to the achievability bound in Theorem \ref{thm:AchievabilityPoltyrev} (or Theorem \ref{thm:AchievabilityMaxLike}) with $r =\sqrt n \sigma e^{\NLD^*-\NLD}$ as `Poltyrev's bound'. The results are shown in Figures \ref{fig:AllBoundsDelta15} and \ref{fig:AllBoundsDelta2}. The exponential behavior of the bounds (the asymptotic slope of the curves in the log-scale graph) is clearly seen in the figures: at NLD above $\NLD_{cr}$, the sphere bound and the ML and Poltyrev's achievability bounds have the same exponent, while for NLD below $\NLD_{cr}$ the exponent of the sphere bound is better. In both cases the typicality bound has a weaker exponent. These observations are corroborated analytically in Section~\ref{sec:Properties} below.

\begin{figure}
  \psfrag{pesb}{$P_e^{SB}$}
\psfrag{petb}{$P_e^{TB}$}
\psfrag{pemlb}{$P_e^{MLB}$}
\psfrag{pemlbrsubopt}{$P_e^{MLB}$ with $r =\sqrt n \sigma e^{\NLD^*-\NLD}$}
  \includegraphics[width=6in]{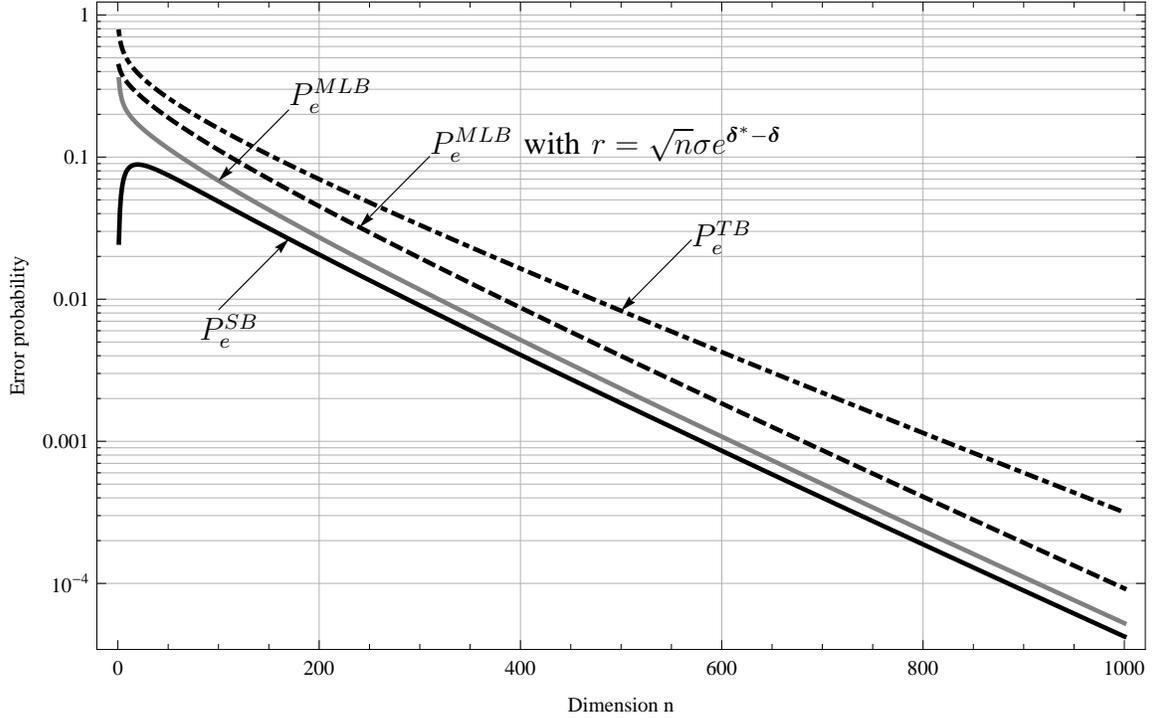}\\
  \caption{Numerical evaluation of the bounds for $\NLD = -1.5\mathsf{nat}$ with $\sigma^2=1$ (0.704db from capacity). From bottom to top: Solid - the sphere bound (Theorem~\ref{thm:SphereBound}). Gray - the ML bound (Theorem \ref{thm:AchievabilityMaxLike}). Dashed - Poltyrev's bound (Theorem~\ref{thm:AchievabilityPoltyrev}). Dot-dashed - the typicality-based achievability bound (Theorem \ref{thm:AchievabilityTypicality}).
    }\label{fig:AllBoundsDelta15}
\end{figure}

\begin{figure}
    \psfrag{pesb}{$P_e^{SB}$}
    \psfrag{petb}{$P_e^{TB}$}
    \psfrag{pemlb}{$P_e^{MLB}$}
    \psfrag{pemlbrsubopt}{$P_e^{MLB}$ with $r =\sqrt n \sigma e^{\NLD^*-\NLD}$}
    \includegraphics[width=6in]{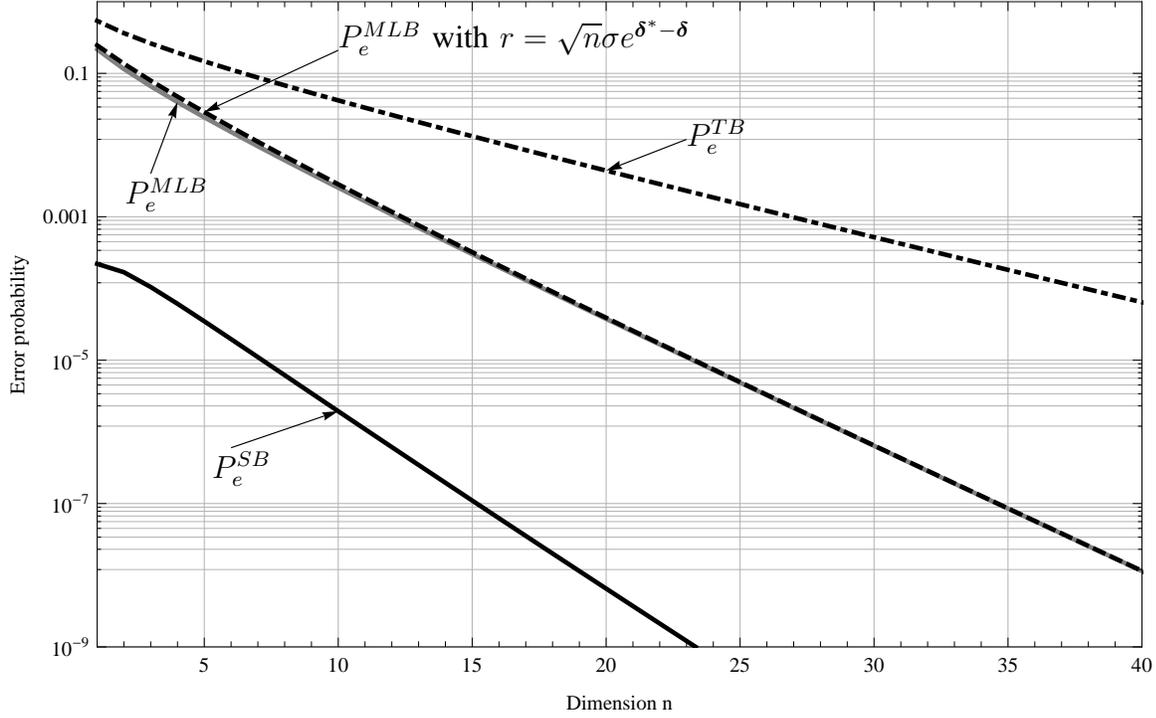}\\
    \caption{
        Numerical evaluation of the bounds for $\NLD = -2\mathsf{nat}$ with $\sigma^2=1$ (5.05db from capacity). From bottom to top: Solid - the sphere bound (Theorem~\ref{thm:SphereBound}). Gray - the ML-based achievability bound (Theorem \ref{thm:AchievabilityMaxLike}). Dashed - Poltyrev's bound (Theorem~\ref{thm:AchievabilityPoltyrev}). Dot-dashed - the typicality-based achievability bound (Theorem \ref{thm:AchievabilityTypicality}).
    }\label{fig:AllBoundsDelta2}
\end{figure}
\fi

\section{Analysis and Asymptotics at Fixed NLD $\NLD$}\label{sec:Properties}
\if \Properties 1

In this section we analyze the bounds presented in the previous section with two goals in mind:
\begin{enumerate}
  \item To derive tight analytical bounds (that require no integration) that allow easy evaluation of the bounds, both upper and lower.
  \item To analyze the bounds asymptotically (for fixed $\NLD$) and refine the error exponent results for the setting.
\end{enumerate}

In \ref{ssec:SphereBoundAnalysis} we present the refined analysis of the sphere bound. While the sphere bound $P_e^{SB}$ will present the same asymptotic form for any $\NLD$, the ML bound $P_e^{MLB}$ has a different behavior above and below $\NLD_{cr}$. In \ref{ssec:MaxLikeAnalysisAbove} we focus on the ML bound above $\NLD_{cr}$. The tight results from \ref{ssec:SphereBoundAnalysis} and \ref{ssec:MaxLikeAnalysisAbove} reveal that (above $\NLD_{cr}$) the optimal error probability $P_e(n,\NLD)$ is known asymptotically up to a constant. This is discussed in \ref{ssec:tightness}. In \ref{ssec:MaxLikeAnalysisBelowDeltaCr} we focus on the ML bound below $\NLD_{cr}$, and in \ref{ssec:MaxLikeAnalysisAtDeltaCr} we consider the special case of $\NLD = \NLD_{cr}$. In \ref{ssec:TypicalityPreciseAsymptotics} we study the asymptotics of the typicality bound $P_e^T(n,\NLD)$ and in \ref{ssec:AsymptoticsMLPoltyrev} we analyze Poltyrev's bound, i.e. the ML bound with $r$ set to $r =\sqrt n \sigma e^{\NLD^*-\NLD}$ instead of $\reff$.

\bigskip

The fact that the ML bound behaves differently above and below $\NLD_{cr}$ can be explained by the following. Consider the first term in the ML bound, $e^{n\NLD} V_n \int_0^{\reff} f_R(r) r^n dr$. Loosely speaking, the value of this integral is determined (for large $n$) by the value of the integrand with the most dominant exponent. When $\NLD > \NLD_{cr}$, the dominating value for the integral is at $r = \reff$. For $\NLD < \NLD_{cr}$, the dominating value is approximately at $r = \sqrt{2n\sigma^2}$. Note that this value does not depend on $\NLD$, so the dependence in $\NLD$ comes from the term $e^{n\NLD}$ alone, and the exponential behavior of the bound is of a straight line. Since we are interested in more than merely the exponential behavior of the bound, we use more refined machinery in order to analyze the bounds.

Poltyrev \cite{Poltyrev94_CodingWithoutRestrictions} used an expurgation technique in order to improve the error exponent for lower NLD values (below $\NLD_{ex} = \NLD^*-\log 2$). The refined tools used here can also be applied to the expurgation bound in order to analyze its sub-exponential behavior. However, in this region the ratio between the upper and lower bounds grows exponentially, and therefore the sub-exponential analysis of the expurgation bound is of little interest and is not included in this paper.

\bigskip

\subsection{Analysis of the Sphere Bound}\label{ssec:SphereBoundAnalysis}
The sphere bound (Theorem \ref{thm:SphereBound}) is a simple bound based on the geometry of the coding problem. However, the resulting expression, given by an integral that has no elementary form, is generally hard to evaluate. There are several approaches for evaluating this bound:
\begin{itemize}
  \item Numeric integration is only possible for small - moderate values of $n$. Moreover, the numeric evaluation does not provide any hints about the asymptotical behavior of the bound.
  \item Tarokh et al. \cite{TarokhVardyZeger99_universal} were able to represent the integral in the bound as a sum of $n/2$ elements. This result indeed helps in numerically evaluating the bound, but does not help in understanding its asymptotics.
  \item Poltyrev \cite{Poltyrev94_CodingWithoutRestrictions} used large-deviation techniques to derive the sphere packing error exponent, i.e.
  \begin{equation}
      \lim_{n\ra\infty} -\frac{1}{n}\log P_e(n,\NLD) \leq \E_{sp}(\NLD) = \frac{1}{2}\left[e^{2(\NLD^*-\NLD)} -1 -2(\NLD^* -\NLD) \right].
  \end{equation}
 The error exponent, as its name suggests, only hints on the exponential behavior of the bound, but does not aid in evaluating the bound itself or in more precise asymptotics.
\end{itemize}

Here we derive non-asymptotic, analytical bounds based on the sphere bound. These bounds allow easy evaluation of the bound, and give rise to more precise asymptotic analysis for the error probability (where $\NLD$ is fixed).

\begin{theorem}\label{thm:SphereBoundAnalysis}
Let $r^* \triangleq \reff = e^{-\NLD} V_n^{-1/n}$, $\rho^* \triangleq \frac{\reff^{2}}{n \sigma^2} $ and $\Upsilon \triangleq \frac{n (\rho^*-1 + \frac{2}{n})}{\sqrt{2(n-2)}}$.
Then for any NLD $\NLD < \NLD^*$ and for any dimension $n>2$, the sphere bound $P_e^{SB}(n,\NLD)$ is lower bounded by
\begin{align}
    P_e^{SB}(n,\NLD) &\geq   e^{n(\NLD^*-\NLD)}e^{n/2}
    e^{ - \frac{n}{2}\rho^*}\cdot
    e^\frac{\Upsilon^2}{2}
    \sqrt{\frac{n^2\pi}{n-2}}
    Q\left(\Upsilon\right)\label{eqn:SphereBoundLowerBoundQfunc}\\
    &\geq\frac{e^{n(\NLD^*-\NLD)}e^{n/2}e^{-\frac{n}{2}\rho^*}}{\rho^*-1+\frac{2}{n}}
    \left(\frac{1}{1+\Upsilon^{-2}}\right)\label{eqn:SphereBoundLowerBoundAnalytic},
\end{align}
\emph{upper} bounded by
\begin{equation}
         P_e^{SB}(n,\NLD) \leq \frac{e^{n(\NLD^*-\NLD)}e^{n/2}e^{-\frac{n}{2}\rho^*}}{\rho^*-1+\frac{2}{n}} \label{eqn:SphereBoundUpperBoundAnalytic},
\end{equation}
and for fixed $\NLD$, given asymptotically by
\begin{equation}
         P_e^{SB}(n,\NLD) = e^{-n E_{sp}(\NLD)} \frac{(n\pi)^{-\frac{1}{2}e^{2(\NLD^* - \NLD)}}}{e^{2(\NLD^* - \NLD)}-1} \left(1+O\left(\frac{\log^2n}{n}\right)\right).\label{eqn:SphereBoundPreciseAsymptotics}
\end{equation}

Some notes regarding the above results:
\begin{itemize}
    \item Eq. \eqref{eqn:SphereBoundLowerBoundQfunc} provides a lower bound in terms of the $Q$ function, and \eqref{eqn:SphereBoundLowerBoundAnalytic} gives a slightly looser bound, but is based on elementary functions only.
    \item The upper bound \eqref{eqn:SphereBoundUpperBoundAnalytic} on the sphere bound has no direct meaning in terms of bounding the error probability $P_e(n,\NLD)$ (since the sphere bound is a lower bound). However, it used for evaluating the sphere bound itself (i.e. to derive \eqref{eqn:SphereBoundPreciseAsymptotics}), and it will prove useful in \emph{upper} bounding $P_e(n,\NLD)$ in Theorem \ref{thm:MaxLikeAnalysis} below.
    \item A bound of the type \eqref{eqn:SphereBoundUpperBoundAnalytic}, i.e. an upper bound on the probability that the noise leaves a sphere, can be derived using the Chernoff bound as was done by Poltyrev \cite[Appendix~B]{Poltyrev94_CodingWithoutRestrictions}. However, while Poltyrev's technique indeed gives the correct exponential behavior, it falls short of attaining the sub-exponential terms, and therefore \eqref{eqn:SphereBoundUpperBoundAnalytic} is tighter. Moreover, \eqref{eqn:SphereBoundUpperBoundAnalytic} leads to the exact precise asymptotics \eqref{eqn:SphereBoundPreciseAsymptotics}.
    \item \eqref{eqn:SphereBoundPreciseAsymptotics} gives an asymptotic bound that is significantly tighter than the error exponent term alone. The asymptotic form \eqref{eqn:SphereBoundPreciseAsymptotics} applies to \eqref{eqn:SphereBoundLowerBoundQfunc}, \eqref{eqn:SphereBoundLowerBoundAnalytic} and \eqref{eqn:SphereBoundUpperBoundAnalytic} as well.
    \item Note that $\rho^*$ is a measure that can also quantify the gap from capacity (see \ref{ssec:MeasuringGapFromCapacity}). It is an alternative to $\Delta\NLD = \NLD^*-\NLD$ (or to $\mu = e^{2\Delta\NLD}$). The measures are not equivalent, but as $n \ra \infty$ we have $\rho^* = e^{2(\NLD^* - \NLD)} + o(1)$, see \eqref{eqn:rhoStarApprox1} and \eqref{eqn:rhoStarApprox2} below.
\end{itemize}

\bigskip

\begin{proof}
We write the sphere bound explicitly:
\begin{align}
    P_e^{SB}(n,\NLD)
    &= \Pr\{\|\bZ\|>r^*\} \nonumber\\
    &= \int_{r^*}^{\infty} f_R(r') dr'\nonumber\\
    &=\int_{r^{*2}/\sigma^2}^{\infty} f_{\chi^2_n}(\rho) d\rho\nonumber\\
    &=\frac{2^{-\frac{n}{2}}}{\Gamma\left(\frac{n}{2}\right)} \int_{r^{*2}/\sigma^2}^{\infty} \rho^{\frac{n}{2}-1}e^{-\rho/2}d\rho\nonumber\\
    &=\frac{2^{-\frac{n}{2}}n^{n/2}}{\Gamma\left(\frac{n}{2}\right)} \int_{\rho^*}^{\infty} \rho^{\frac{n}{2}-1}e^{-n\rho/2}d\rho.\label{eqn:SPBintegral}
\end{align}

In order to evaluate the integral in \eqref{eqn:SPBintegral} we require the following lemma.

\begin{lem}\label{lem:SpherePackingIntegralBounds} Let $n>2$ and $x>1-\frac{2}{n}$. Then the integral $\int_x^{\infty} \rho^{\frac{n}{2}-1}e^{-n\rho/2}d\rho$ can be bounded from above by
\begin{align}
    \int_x^{\infty} \rho^{\frac{n}{2}-1}e^{-n\rho/2}d\rho
    &\leq
    \frac{2x^{\frac{n}{2}}e^{ - \frac{n x}{2}}}{n(x-1+\frac{2}{n})}
    \label{eqn:SpherePackingIntegralUB}
\end{align}
and from below by
\begin{align}
    \int_x^{\infty} \rho^{\frac{n}{2}-1}e^{-n\rho/2}d\rho
    &\geq
    2x^{\frac{n}{2}}e^{ - \frac{n x}{2}}
    \exp\left[\frac{\Upsilon^2}{2}\right]
    \sqrt{\frac{\pi}{n-2}}
    Q\left(\Upsilon\right)\label{eqn:SpherePackingIntegralLB1} \\
    &\geq
    \frac{2x^{\frac{n}{2}}e^{ - \frac{n x}{2}}}{n(x-1+\frac{2}{n})}
    \left(\frac{1}{1+\Upsilon^{-2}}\right),\label{eqn:SpherePackingIntegralLB2}\\
    &\geq
    \frac{2x^{\frac{n}{2}}e^{ - \frac{n x}{2}}}{n(x-1+\frac{2}{n})}
    \left(1-\frac{1}{\Upsilon^2}\right),\label{eqn:SpherePackingIntegralLB3}
\end{align}
where $\Upsilon \triangleq \frac{n(x-1+\frac{2}{n})}{\sqrt{2(n-2)}}$.
\begin{proof} Appendix~\ref{app:IntegralBounds}.
\end{proof}
\end{lem}

We continue the proof of the theorem: \eqref{eqn:SphereBoundLowerBoundQfunc} follows by plugging \eqref{eqn:SpherePackingIntegralLB1} into \eqref{eqn:SPBintegral} with $x = \rho^*$. It can be shown that $\rho^*\geq 1$ for all $\NLD < \NLD^*$ so the condition $x > 1-\frac{2}{n}$ is met. \eqref{eqn:SphereBoundLowerBoundAnalytic} follows similarly using \eqref{eqn:SpherePackingIntegralLB2} and the definition of $\NLD^*$. The upper bound \eqref{eqn:SphereBoundUpperBoundAnalytic} follows using \eqref{eqn:SpherePackingIntegralUB}.

To derive \eqref{eqn:SphereBoundPreciseAsymptotics} we first note the following asymptotic results:
\begin{equation}\label{eqn:VnApproxMult}
    V_n = \frac{\pi^{n/2}}{\frac{n}{2}\Gamma(\frac{n}{2})} = \left(\frac{2\pi e}{n}\right)^{n/2}\frac{1}{\sqrt{n\pi}}\left(1+O\left(\frac{1}{n}\right)\right),
\end{equation}
\begin{align}
    \rho^* = \frac{e^{-2\NLD} V_n^{-2/n}}{n \sigma^2}
    &= e^{2(\NLD^*-\NLD)} (n\pi)^{1/n}
    \left(1+O\left(\frac{1}{n^2}\right)\right)\label{eqn:rhoStarApprox1}\\
    &= e^{2(\NLD^*-\NLD)}
    \left(1+\frac{1}{n}\log(n\pi) + O\left(\frac{\log^2n}{n^2}\right)\right)\label{eqn:rhoStarApprox2},
\end{align}
\begin{equation}\label{eqn:UpsilonApprox}
    \Upsilon = \frac{n (\rho^*-1 + \frac{2}{n})}{\sqrt{2(n-2)}} = \sqrt{\frac{n}{2}}
\left(e^{2(\NLD^*-\NLD)} - 1 \right)\left(1+O\left(\frac{\log n}{n}\right)\right) = \Theta(\sqrt n).
\end{equation}
For \eqref{eqn:VnApproxMult} see Appendix \ref{app:Vn}. \eqref{eqn:rhoStarApprox1} follows from \eqref{eqn:VnApproxMult} and the definition of $\NLD^*$. \eqref{eqn:rhoStarApprox2} follows by writing $(n\pi)^{1/n} = e^{\frac{1}{n}\log(n\pi)}$ and the Taylor approximation. \eqref{eqn:UpsilonApprox} follows directly from \eqref{eqn:rhoStarApprox2}.

We evaluate the term $e^{-\frac{n}{2}\rho^*}$ in \eqref{eqn:SphereBoundLowerBoundAnalytic} and \eqref{eqn:SphereBoundUpperBoundAnalytic}:
\begin{align}
    e^{-\frac{n}{2}\rho^*}
    &= \exp\left[-\frac{n}{2}e^{2(\NLD^*-\NLD)}\left(1+\frac{1}{n}\log(n\pi) + O\left(\frac{\log^2n}{n^2}\right)\right)\right]\nonumber\\
    &= e^{-\frac{n}{2}e^{2(\NLD^*-\NLD)}}
    \exp\left[-\frac{1}{2}e^{2(\NLD^*-\NLD)}\log(n\pi) + O\left(\frac{\log^2n}{n}\right)\right]\nonumber\\
    &= e^{-\frac{n}{2}e^{2(\NLD^*-\NLD)}}(n\pi)^{-\frac{1}{2}e^{2(\NLD^*-\NLD)}}
    \left(1+ O\left(\frac{\log^2n}{n}\right)\right).\label{eqn:exprhoApprox}
\end{align}
Plugging \eqref{eqn:rhoStarApprox2}, \eqref{eqn:UpsilonApprox} and \eqref{eqn:exprhoApprox} into \eqref{eqn:SphereBoundLowerBoundAnalytic} and \eqref{eqn:SphereBoundUpperBoundAnalytic}, along with the definition of $\E_{sp}(\NLD)$, leads to the desired \eqref{eqn:SphereBoundPreciseAsymptotics}.
\end{proof}
\end{theorem}

In Fig.~\ref{fig:SphereBoundPlot} we demonstrate the tightness of the bounds and precise asymptotics of Theorem~\ref{thm:SphereBoundAnalysis}. In the figure the sphere bound is presented with its bounds and approximations. The lower bound \eqref{eqn:SphereBoundLowerBoundQfunc} is the tightest lower bound (but is based on the non-analytic $Q$ function). The analytic lower bound \eqref{eqn:SphereBoundLowerBoundAnalytic} is slightly looser than \eqref{eqn:SphereBoundLowerBoundQfunc}, but is tight enough in order to derive the precise asymptotic form \eqref{eqn:SphereBoundPreciseAsymptotics}. The upper bound \eqref{eqn:SphereBoundUpperBoundAnalytic} of the sphere bound is also tight. The error exponent itself (without the sub-exponential terms) is clearly way off, compared to the precise asymptotic form \eqref{eqn:SphereBoundPreciseAsymptotics}.

\begin{figure}
\psfrag{pesb}{Sphere bound $P_e^{SB}$}
\psfrag{SBUB}{Upper bound \eqref{eqn:SphereBoundUpperBoundAnalytic}}
\psfrag{SBLB}{Lower bound \eqref{eqn:SphereBoundLowerBoundAnalytic}}
\psfrag{SBLBQ}{Lower bound \eqref{eqn:SphereBoundLowerBoundQfunc}}
\psfrag{Asym}{Asymptotic form \eqref{eqn:SphereBoundPreciseAsymptotics}}
\psfrag{ErrExp}{$e^{-n \E_{sp}(\NLD)}$}
  \includegraphics[width=6in]{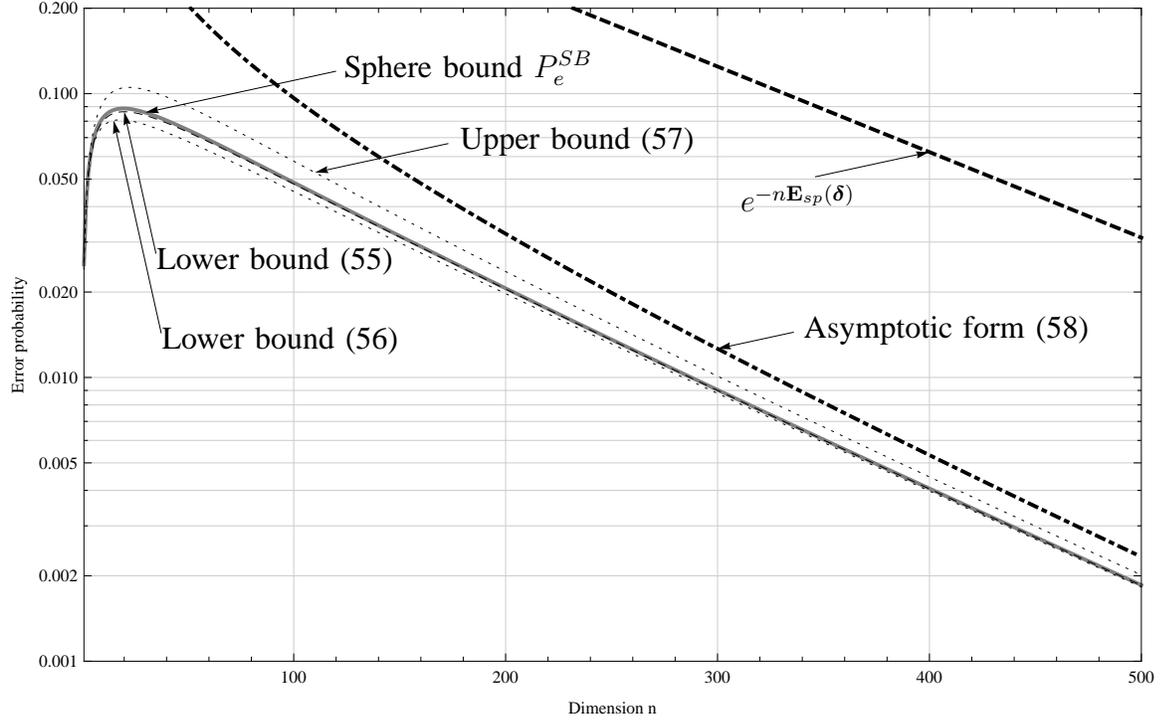}\\
  \caption{Numerical evaluation of the sphere bound and its bounds and approximation in Theorem~\ref{thm:SphereBoundAnalysis} vs the dimension $n$. Here $\NLD = -1.5\mathsf{nat}$ and $\sigma^2=1$ (0.704db from capacity).
  The tight bounds \eqref{eqn:SphereBoundLowerBoundQfunc}, \eqref{eqn:SphereBoundLowerBoundAnalytic} and  \eqref{eqn:SphereBoundUpperBoundAnalytic} lead to the asymptotic form \eqref{eqn:SphereBoundPreciseAsymptotics}. The error exponent term alone is evidently way off compared to \eqref{eqn:SphereBoundPreciseAsymptotics}.
}\label{fig:SphereBoundPlot}
\end{figure}

\bigskip

\subsection{Analysis of the ML Bound Above $\NLD_{cr}$}\label{ssec:MaxLikeAnalysisAbove}

In order to derive the random coding exponent $\E_r(\NLD)$, Poltyrev's achievability bound (Theorem~\ref{thm:AchievabilityPoltyrev}) was evaluated asymptotically by setting a suboptimal value $\sqrt n \sigma e^{-(\NLD^*-\NLD)}$ for the parameter $r$. While setting this value still gives the correct exponential behavior of the bound, a more precise analysis (in the current and following subsections) using the optimal value for $r$ as in Theorem~\ref{thm:AchievabilityMaxLike} gives tighter analytical and asymptotic results.

\begin{theorem}\label{thm:MaxLikeAnalysis}
Let $r^* \triangleq \reff = e^{-\NLD} V_n^{-1/n}$, $\rho^* \triangleq \frac{\reff^{2}}{n \sigma^2} $,
$\Upsilon \triangleq \frac{n (\rho^*-1 + \frac{2}{n})}{\sqrt{2(n-2)}}$ and $\Psi \triangleq \frac{\sqrt n \left(2-\rho^*+\tfrac{2}{n}\right)}{2\sqrt {\rho^*}}$.

Then for any NLD $\NLD$ and for any dimension $n>2$ where
$1-\frac{2}{n} <\rho^* < 2-\frac{2}{n}$, the ML bound $P_e^{MLB}(n,\NLD)$ is upper bounded by
\begin{align}
   P_e^{MLB}(n,\NLD)  &\leq\frac{e^{n(\NLD^*-\NLD)}e^{n/2}e^{-\frac{n}{2}\rho^*}}{\left(2-\rho^*-\frac{2}{n}\right)\left(\rho^*-1+\frac{2}{n}\right)},
  \label{eqn:MaxLikeUpperBoundAnalytic}
\end{align}
lower bounded by
\begin{align}
    P_e^{MLB}(n,\NLD) &\geq e^{n(\NLD^*-\NLD)}e^{n/2}    e^{-n\rho^*/2} \left[
    e^{\Psi^2/2} \sqrt\frac{n\pi}{2 \rho^*}Q(\Psi)
     +
    e^{\Upsilon^2/2} \sqrt{\frac{n^2\pi}{n-2}} Q\left(\Upsilon\right)\right]\label{eqn:MaxLikeLowerBoundQfunc}\\
    &\geq
    e^{n(\NLD^*-\NLD)}e^{n/2}    e^{-n\rho^*/2} \left[
    \frac{1}{2-\rho^*+\frac{2}{n}}\cdot\frac{1}{1+\Psi^{-2}}
     +
    \frac{1}{\rho^*-1+\frac{2}{n}}\cdot\frac{1}{1+\Upsilon^{-2}}
\right],
    \label{eqn:MaxLikeLowerBoundAnalytic}
\end{align}
and for $\NLD_{cr}<\NLD<\NLD^*$, given asymptotically by
\begin{equation}
    P_e^{MLB}(n,\NLD) = \frac{
    e^{-n\E_{r}(\NLD)}(n\pi)^{-\frac{1}{2}e^{2(\NLD^*-\NLD)}}}
    {\left(2-e^{2(\NLD^*-\NLD)}\right)\left(e^{2(\NLD^*-\NLD)}-1\right)}
    \left(1+ O\left(\frac{\log^2n}{n}\right)\right).
     \label{eqn:MaxLikeLowerBoundAsymptotic}
\end{equation}

Some notes regarding the above results:
\begin{itemize}
  \item For large $n$, the condition $\rho^* < 2-\frac{2}{n}$ translates to the fact that $\NLD_{cr}<\NLD$. $\rho^*>1-\frac{2}{n}$ holds for all $\NLD<\NLD^*$. The case of $\NLD \leq \NLD^*$ is addressed later on in the current section.
 \item The lower bounds \eqref{eqn:MaxLikeLowerBoundQfunc} and \eqref{eqn:MaxLikeLowerBoundAnalytic}
  have no direct meaning in terms of bounding the error probability $P_e(n,\NLD)$ (since they lower bound an upper bound). However, they are useful for evaluating the achievability bound itself (i.e. to derive \eqref{eqn:MaxLikeLowerBoundAsymptotic}).
  \item \eqref{eqn:MaxLikeLowerBoundAsymptotic} gives an asymptotic bound that is significantly tighter than the error exponent term alone. It holds above the $\NLD_{cr}$ only, where below $\NLD_{cr}$ and exactly at $\NLD_{cr}$ we have Theorems \ref{thm:MaxLikeAnalysisBelowDeltaCr} and \ref{thm:MaxLikeAnalysisAtDeltaCr} below.
      The asymptotic form \eqref{eqn:MaxLikeLowerBoundAsymptotic} applies to \eqref{eqn:MaxLikeUpperBoundAnalytic}, \eqref{eqn:MaxLikeLowerBoundQfunc} and \eqref{eqn:MaxLikeLowerBoundAnalytic} as well.
\end{itemize}
\end{theorem}

\begin{proof}
The proof relies on a precise analysis of the ML bound:
\begin{equation}
    e^{n\NLD} V_n \int_0^{r^*} f_R(r) r^n dr + \Pr\left\{ \|\bZ\| > r^* \right\}.
\end{equation}
The second term is exactly the sphere bound, which allows the usage of the analysis of Theorem \ref{thm:SphereBoundAnalysis}. We therefore proceed with analyzing the first term:
\begin{align}
    e^{n\NLD} V_n \int_0^{r^*} f_R(r) r^n dr
    &= e^{n\NLD} V_n\sigma^n \int_0^{\frac{r^*}{\sigma}} f_{\chi_n}(y) \rho^n dy\nonumber\\
    &= e^{n\NLD} V_n\sigma^n \frac{2^{1-\frac{n}{2}}}{\Gamma\left[\frac{n}{2}\right]} \int_0^{\frac{r^*}{\sigma}} e^{-y^2/2} y^{2n-1} dy\nonumber\\
    &= e^{n\NLD} V_n\sigma^n \frac{2^{-\frac{n}{2}}}{\Gamma\left[\frac{n}{2}\right]} \int_0^{\frac{r^{*2}}{\sigma^2}} e^{-t/2} t^{n-1} dt\nonumber\\
    &= \frac{n}{2} e^{n(\NLD+\NLD^*)} V_n^2e^{n/2}
    \sigma^{2n}n^n \int_0^{\rho^*} e^{-n\rho/2}\rho^{n-1} d\rho\label{eqn:preMLIntegralBounds}
\end{align}

We need the following lemma:
\begin{lem}\label{lem:MaxLikeIntegralBounds} Let $0<x<2-\frac{2}{n}$. Then the integral $\int_0^{x} e^{-n \rho/2}  \rho^{n-1} d\rho$ is upper bounded by
\begin{equation}\label{eqn:MLIntegralUpperBound}
    \int_0^{x} e^{-n \rho/2}  \rho^{n-1} d\rho \leq \frac{2x^ne^{-nx/2}}{n\left(2-x-\tfrac{2}{n}\right)}\left(1 - e^{-n\left(1-\tfrac{1}{n} -\tfrac{x}{2}\right)}\right),
\end{equation}
and is lower bounded by
\begin{align}
    \int_0^{x} e^{-n \rho/2}  \rho^{n-1} d\rho
    &\geq x^n e^{-nx/2} e^{\Psi^2/2} \sqrt\frac{2\pi}{n x}Q(\Psi)
    \label{eqn:MLIntegralLowerBoundQfunc}\\
    &\geq \frac{2x^n e^{-nx/2}}{n\left(2-x+\tfrac{2}{n}\right)}\cdot\frac{1}{1+\Psi^{-2}},
    \label{eqn:MLIntegralLowerBoundAnalytic}
\end{align}
where $\Psi \triangleq \frac{\sqrt n \left(2-x+\tfrac{2}{n}\right)}{2\sqrt x}$.
\begin{proof} Appendix~\ref{app:IntegralBounds}.
\end{proof}
\end{lem}

To prove the upper bound \eqref{eqn:MaxLikeUpperBoundAnalytic} we use \eqref{eqn:MLIntegralUpperBound} with  $x = \rho^*$ to bound \eqref{eqn:preMLIntegralBounds}:
\begin{align}
    e^{n\NLD} V_n \int_0^{r^*} f_R(r) r^n dr
    &=
    \frac{n}{2} e^{n(\NLD+\NLD^*)} V_n^2e^{n/2}
    \sigma^{2n}n^n
    \int_0^{\rho^*} e^{-n \rho/2}\rho^{n-1} d\rho\\
    &\leq
    \frac{n}{2} e^{n(\NLD+\NLD^*)} V_n^2e^{n/2}
    \sigma^{2n}n^n
    \cdot\frac{2\rho^{*n}e^{-\tfrac{n}{2}\rho^*}}{n\left(2-\rho^*-\tfrac{2}{n}\right)}\\
    &= \frac{e^{n(\NLD^*-\NLD)} e^{n/2}e^{-\tfrac{n}{2}\rho^*}}{2-\rho^*-\tfrac{2}{n}}.
\end{align}

We combine the above with the upper bound \eqref{eqn:SphereBoundUpperBoundAnalytic} on the sphere bound and get
\begin{equation}
    e^{n\NLD} V_n \int_0^{r^*} f_R(r) r^n dr + \Pr\left\{ \|\bZ\| > r^* \right\} \leq \frac{e^{n(\NLD^*-\NLD)} e^{n/2}e^{-\tfrac{n}{2}\rho^*}}{2-\rho^*-\tfrac{2}{n}} +
    \frac{e^{n(\NLD^*-\NLD)}e^{n/2}e^{-\frac{n}{2}\rho^*}}{\rho^*-1+\frac{2}{n}},
\end{equation}
which immediately leads to \eqref{eqn:MaxLikeUpperBoundAnalytic}.

In order to attain the lower bound \eqref{eqn:MaxLikeLowerBoundQfunc} we use \eqref{eqn:MLIntegralLowerBoundQfunc} with $x=\rho^*$ and get that
$e^{n\NLD} V_n \int_0^{r^*} f_R(r) r^n dr$ is lower bounded by
\begin{align*}
    &\frac{n}{2} e^{n(\NLD+\NLD^*)} V_n^2e^{n/2}
    \sigma^{2n}n^n\cdot \rho^{*n} e^{-n\rho^*/2} e^{\Psi^2/2} \sqrt\frac{2\pi}{n \rho^*}Q(\Psi)\\
    &=e^{n(\NLD^*-\NLD)}e^{n/2}
     e^{-n\rho^*/2} \cdot e^{\Psi^2/2} \sqrt\frac{n\pi}{2 \rho^*}Q(\Psi).
\end{align*}
Eq. \eqref{eqn:MaxLikeLowerBoundQfunc} follows by using the lower bound \eqref{eqn:SphereBoundLowerBoundQfunc} on the sphere bound. The analytic bound \eqref{eqn:MaxLikeLowerBoundAnalytic} follows from \eqref{eqn:MLIntegralLowerBoundAnalytic}.

The asymptotic form \eqref{eqn:MaxLikeLowerBoundAsymptotic} follows by the fact that $\Psi =  \Theta(\sqrt n)$, and by plugging \eqref{eqn:rhoStarApprox2} and \eqref{eqn:UpsilonApprox} into the analytical bounds \eqref{eqn:MaxLikeUpperBoundAnalytic} and \eqref{eqn:MaxLikeLowerBoundAnalytic}.
\end{proof}

In Fig.~\ref{fig:MaxLikeBoundPlotDelta15} we demonstrate the tightness of the bounds and precise asymptotics in Theorem~\ref{thm:MaxLikeAnalysis}. In the figure the ML bound is presented with its bounds and approximations. The image is similar to the Fig.~\ref{fig:SphereBoundPlot}, referring to the sphere bound. The lower bound \eqref{eqn:MaxLikeLowerBoundQfunc} is the tightest lower bound (but is based on the non-analytic $Q$ function). The analytic lower bound \eqref{eqn:MaxLikeLowerBoundAnalytic} is slightly looser than \eqref{eqn:MaxLikeLowerBoundQfunc}, but is tight enough in order to derive the precise asymptotic form \eqref{eqn:MaxLikeLowerBoundAsymptotic}. The upper bound \eqref{eqn:MaxLikeUpperBoundAnalytic} of the sphere bound is also tight. The error exponent itself (without the sub-exponential terms) is clearly way off, compared to the precise asymptotic form \eqref{eqn:MaxLikeLowerBoundAsymptotic}.

\begin{figure}
\psfrag{pemlb}{ML bound $P_e^{MLB}$}
\psfrag{MLBUB}{Upper bound \eqref{eqn:MaxLikeUpperBoundAnalytic}}
\psfrag{MLBLB}{Lower bound \eqref{eqn:MaxLikeLowerBoundAnalytic}}
\psfrag{MLBLBQ}{Lower bound \eqref{eqn:MaxLikeLowerBoundQfunc}}
\psfrag{Asym}{Asymptotic form \eqref{eqn:MaxLikeLowerBoundAsymptotic}}
\psfrag{ErrExp}{$e^{-n \E_{r}(\NLD)}$}
  \includegraphics[width=6in]{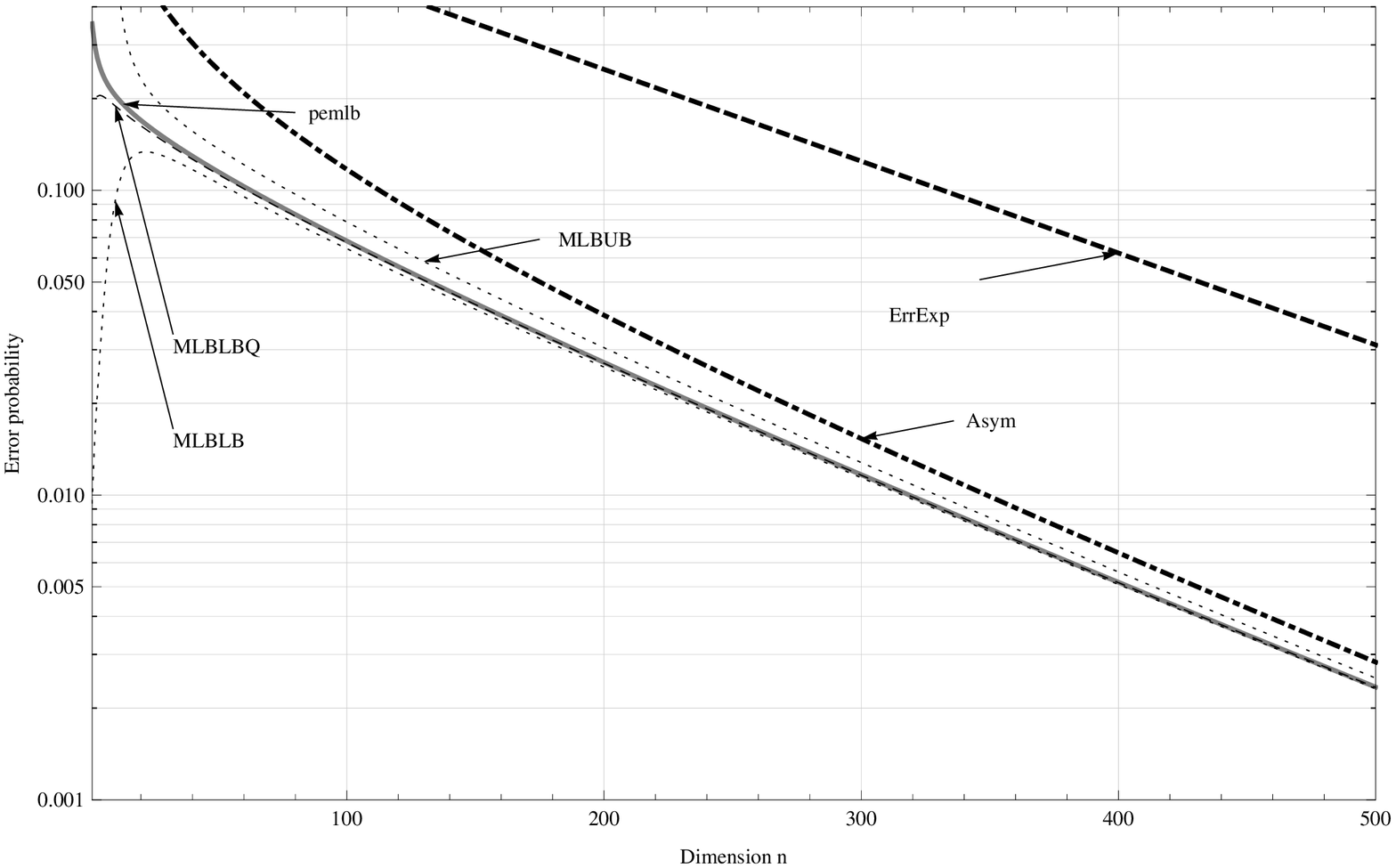}\\
  \caption{Numerical evaluation of the ML bound and its bounds and approximation in Theorem~\ref{thm:MaxLikeAnalysis} vs the dimension $n$. Here $\NLD = -1.5\mathsf{nat}$ (0.704db from capacity).
  The tight bounds \eqref{eqn:MaxLikeUpperBoundAnalytic}, \eqref{eqn:MaxLikeLowerBoundQfunc} and  \eqref{eqn:MaxLikeLowerBoundAnalytic} lead to the asymptotic form \eqref{eqn:MaxLikeLowerBoundAsymptotic}. The error exponent term alone is evidently way off compared to \eqref{eqn:MaxLikeLowerBoundAsymptotic}.
  }\label{fig:MaxLikeBoundPlotDelta15}
\end{figure}

\subsection{Tightness of the Bounds Above $\NLD_{cr}$}\label{ssec:tightness}
\begin{cor}
For $\NLD_{cr} < \NLD < \NLD^*$ the ratio between the upper and lower bounds on $P_e(n,\NLD)$ converges to a constant, i.e.
\begin{equation}
    \frac{P_e^{MLB}(n,\NLD)}{P_e^{SB}(n,\NLD)} = \frac{1}{\left(2-e^{2(\NLD^*-\NLD)}\right)} + O\left(\frac{\log n}{n}\right).
\end{equation}
\begin{proof}
The proof follows from Theorems \ref{thm:SphereBoundAnalysis} and \ref{thm:MaxLikeAnalysis}. Note that the result is tighter than the ratio of the asymptotic forms \eqref{eqn:SphereBoundPreciseAsymptotics} and \eqref{eqn:MaxLikeLowerBoundAsymptotic} (i.e. $O(\tfrac{\log n}{n})$ and not $O(\tfrac{\log^2 n}{n})$) since the term that contributes the $\log^2n$ term is $e^{-\frac{n}{2}\rho^*}$ which is common for both upper and lower bounds.
\end{proof}
\end{cor}

\subsection{The ML Bound Below $\NLD_{cr}$}\label{ssec:MaxLikeAnalysisBelowDeltaCr}

Here we provide the asymptotic behavior of the ML bound at NLD values \emph{below} $\NLD_{cr}$.

\begin{theorem}\label{thm:MaxLikeAnalysisBelowDeltaCr}
For any $\NLD < \NLD_{cr}$, the ML bound can be approximated by
\begin{equation}
    P_e^{MLB}(n,\NLD) = \frac{e^{-n\E_r(\NLD)}}{\sqrt{2 \pi n}}\left(1+O\left(\tfrac{1}{n}\right)\right).\label{eqn:MaxLikeAnalysisBelowDeltaCr}
\end{equation}

\begin{proof}
We start as in the proof of Theorem \ref{thm:AchievabilityMaxLike} to have
\begin{equation}
    e^{n\NLD} V_n \int_0^{r^*} f_R(r) r^n dr
    = \frac{n}{2} e^{n\NLD} V_n^2\sigma^n (2\pi)^{-\frac{n}{2}}n^n \int_0^{\rho^*} e^{-n\rho/2} \rho^{n-1} d\rho.\label{eqn:preLaplace}
\end{equation}
We continue by approximating the integral as follows:
\begin{lem}\label{lem:MaxLikeIntegralLaplace}
Let $x >2$. Then the integral $\int_0^{x} e^{-n \rho/2}  \rho^{n-1} d\rho$ can be approximated by
\begin{align}
    \int_0^{x} e^{-n \rho/2}  \rho^{n-1} d\rho
    =\sqrt{\frac{2\pi}{n}}e^{-n}2^n \left(1+O\left(\tfrac{1}{n}\right)\right).
    \label{eqn:MLIntegralBelowDeltaCrAsymptotic}
\end{align}
\begin{proof}
The proof relies on the fact that the integrand is maximized at the interior of the interval $[0,x]$. Note that the result does not depend on $x$.

We first rewrite the integral to the form
\begin{equation}
    \int_0^{x} \frac{1}{\rho} e^{-n (\rho/2-\log\rho)} d\rho
    = \int_0^{x} g(\rho) e^{-n G(\rho)} d\rho,
\end{equation}
where $g(\rho) \triangleq \frac{1}{\rho}$ and $G(\rho)\triangleq \rho/2-\log\rho$.

When $n$ grows, the asymptotical behavior of the integral is dominated by the value of the integrand at $\tilde \rho = 2$ (which minimizes $G(\rho)$). This is formalized by Laplace's method of integration (see, e.g. \cite[Sec. 3.3]{BarndorffNielsenCox_Asymptotic}):
\begin{align*}
\int_0^{x} g(\rho) e^{-n G(\rho)} d\rho
    &= g(\tilde\rho)e^{-nG(\tilde\rho)} \sqrt\frac{2\pi}{n \frac{\partial^2G(\tilde\rho)}{\partial\rho^2}|_{\rho = \tilde\rho}}\left(1+O\left(\tfrac{1}{n}\right)\right)\\
    &= \frac{1}{2}e^{n(1-\log 2)} \sqrt\frac{2\pi}{n \cdot\tfrac{1}{4}}\left(1+O\left(\tfrac{1}{n}\right)\right),
\end{align*}
which leads to \eqref{eqn:MLIntegralBelowDeltaCrAsymptotic}.
\end{proof}
\end{lem}

Before we apply the result of the lemma to \eqref{eqn:preLaplace}, we note that whenever $\NLD$ is below the critical  $\NLD_{cr}$, $\rho^* > e^{2(\NLD^* - \NLD)} = 2e^{2(\NLD_{cr} - \NLD)} > 2$ for all $n$. Therefore for all $n$ we have
\begin{equation}
    \int_0^{2e^{2(\NLD_{cr} - \NLD)}} e^{-n\rho/2} \rho^{n-1} d\rho
    \leq \int_0^{\rho^*} e^{-n\rho/2} \rho^{n-1} d\rho
    \leq \int_0^{\infty} e^{-n\rho/2} \rho^{n-1} d\rho.
\end{equation}
We apply Lemma \ref{lem:MaxLikeIntegralLaplace} to both sides of the equation and conclude that
\begin{equation}
    \int_0^{\rho^*} e^{-n\rho/2} \rho^{n-1} d\rho = \sqrt{\frac{2\pi}{n}}e^{-n}2^n \left(1+O\left(\tfrac{1}{n}\right)\right).
\end{equation}
The proof of the theorem is completed using the approximation \eqref{eqn:VnApproxMult} for $V_n$.

It should be noted that the sphere bound part of the achievability bound vanishes with a stronger exponent ($E_{sp}(\NLD)$), and therefore does not contribute to the asymptotic value.
\end{proof}\end{theorem}

In Fig.~\ref{fig:MaxLikeBoundPlotDelta18} we demonstrate the tightness of the precise asymptotics in Theorem~\ref{thm:MaxLikeAnalysisBelowDeltaCr}. Here too the precise asymptotic form 
 is significantly tighter than the error exponent only.

\begin{figure}
  \psfrag{pemlb}{ML bound $P_e^{MLB}$}
  \psfrag{Asym}{Asymptotic form \eqref{eqn:MaxLikeAnalysisBelowDeltaCr}}
  \psfrag{ErrExp}{$e^{-n \E_{r}(\NLD)}$}
  \includegraphics[width=6in]{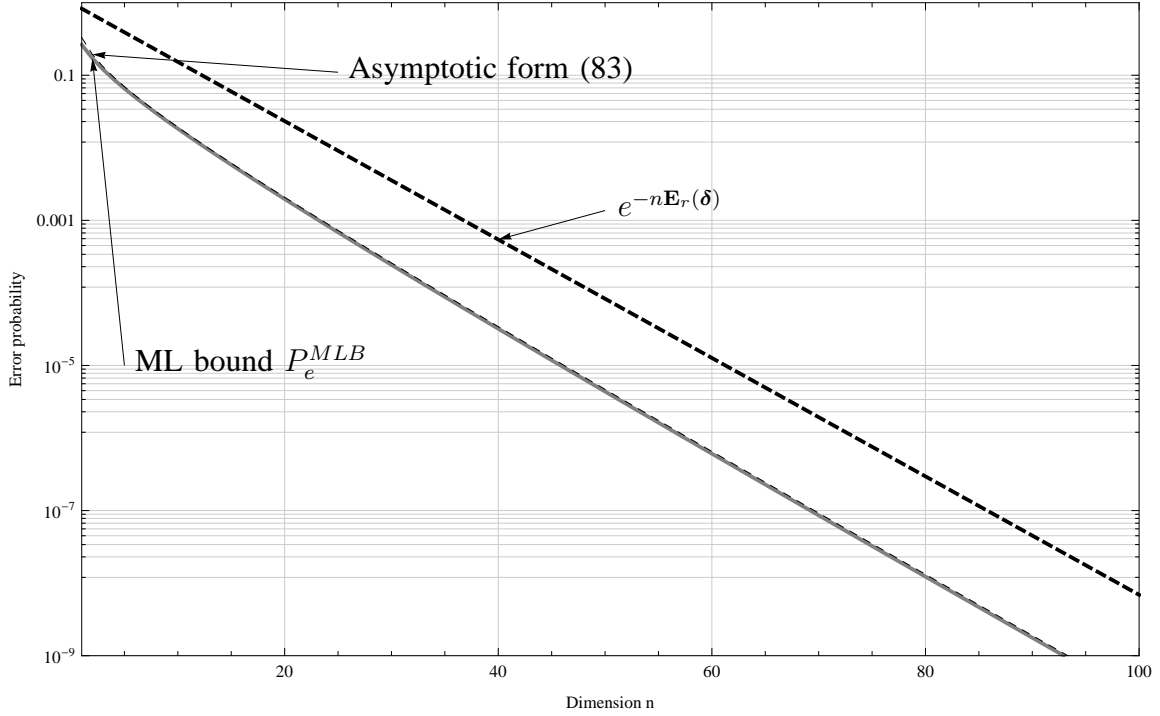}\\
  \caption{Numerical evaluation of the ML bound and its approximation in Theorem~\ref{thm:MaxLikeAnalysisBelowDeltaCr} vs the dimension $n$. Here $\NLD = -1.8\mathsf{nat}$ (3.31db from capacity).  The precise asymptotic form \eqref{eqn:MaxLikeAnalysisBelowDeltaCr} is clearly tighter than the error exponent only.
  }\label{fig:MaxLikeBoundPlotDelta18}
\end{figure}

\subsection{The ML Bound at $\NLD_{cr}$}\label{ssec:MaxLikeAnalysisAtDeltaCr}
In previous subsections we provided asymptotic forms for the upper bound on $P_e(n,\NLD)$, for $\NLD > \NLD_{cr}$ and for $\NLD < \NLD_{cr}$ (Theorems \ref{thm:MaxLikeAnalysis} and \ref{thm:MaxLikeAnalysisBelowDeltaCr} respectively). Unfortunately, neither theorem holds for $\NLD_{cr}$ exactly. We now analyze the upper bound at $\NLD_{cr}$, and show that its asymptotic form is different at this point. As a consequence, at the critical NLD, the ratio between the upper and lower bounds on $P_e(n,\NLD)$ is of the order of $\sqrt n$ (this ratio above $\NLD_{cr}$ is a constant, and below $\NLD_{cr}$ the ratio increases exponentially since the error exponents are different).

\bigskip

\begin{theorem}\label{thm:MaxLikeAnalysisAtDeltaCr}
    At $\NLD = \NLD_{cr}$, the ML bound is given asymptotically by
\begin{align}
    P_e^{MLB}(n,\NLD_{cr})
    &= e^{-n\E_r(\NLD_{cr})}
    \frac{1}{2\pi}
    \left[\sqrt{\frac{\pi}{2 n}} + \frac{\log(n\pi e^2)}{n}\right]
    \left(1+O\left(\tfrac{\log^2 n}{n}\right)\right)
    \label{eqn:MaxLikeAnalysisAtDeltaCr}\\
    &= e^{-n\E_r(\NLD_{cr})}
    \frac{1}{\sqrt{8\pi n}}\left(1+O\left(\tfrac{\log n}{\sqrt n}\right)\right)
    \label{eqn:MaxLikeAnalysisAtDeltaCrLoose}
\end{align}

\begin{proof} Appendix \ref{app:delta_cr}.\end{proof}
\end{theorem}

In Fig.~\ref{fig:MaxLikeBoundPlotDeltaCr} we demonstrate the tightness of the precise asymptotics of Theorem~\ref{thm:MaxLikeAnalysisAtDeltaCr}.

\begin{figure}
  \psfrag{pemlb}{ML bound $P_e^{MLB}$}
  \psfrag{Asym}{Tight asymptotic form \eqref{eqn:MaxLikeAnalysisAtDeltaCr}}
  \psfrag{Asym2}{Loose asymptotic form \eqref{eqn:MaxLikeAnalysisAtDeltaCrLoose}}
  \psfrag{ErrExp}{$e^{-n \E_{r}(\NLD_{cr})}$}
  \includegraphics[width=6in]{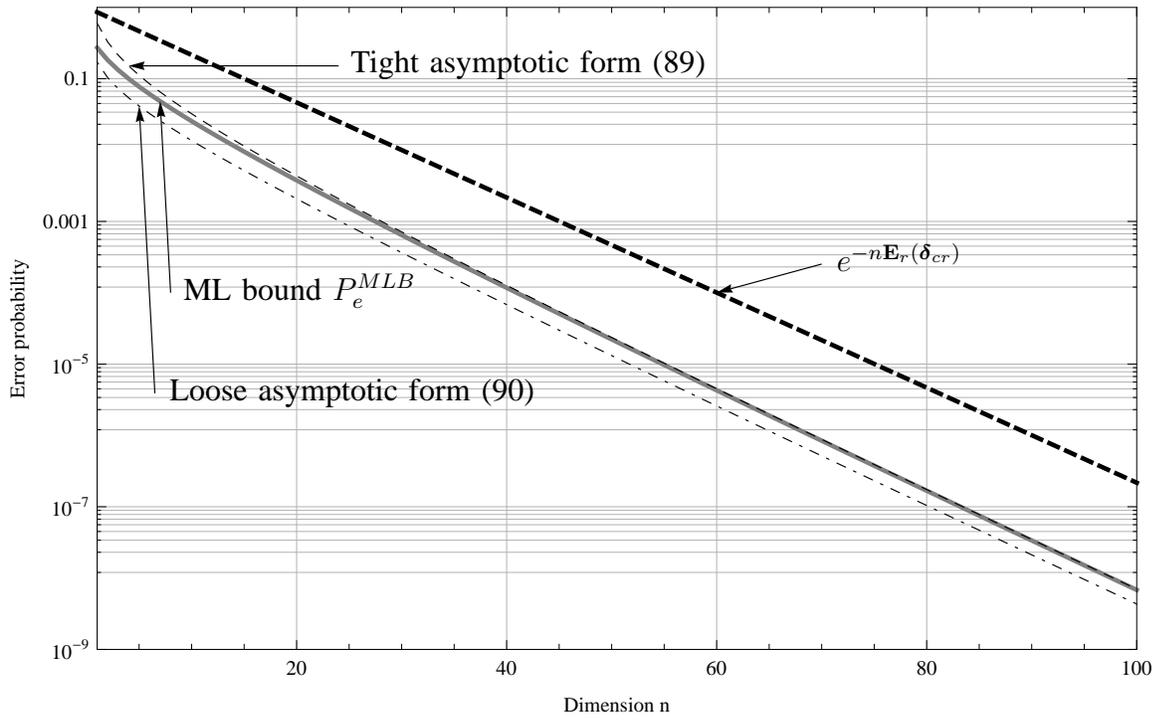}\\
  \caption{Numerical evaluation of the ML bound at $\NLD = \NLD_{cr}$ (3.01db from capacity) and its approximations in Theorem~\ref{thm:MaxLikeAnalysisAtDeltaCr} vs the dimension $n$.
  The asymptotic form \eqref{eqn:MaxLikeAnalysisAtDeltaCr} is tighter than the simpler \eqref{eqn:MaxLikeAnalysisAtDeltaCrLoose}. Both forms approximate the true value of the ML bound better than the error exponent term alone.
  }\label{fig:MaxLikeBoundPlotDeltaCr}
\end{figure}

\subsection{Asymptotic Analysis of the Typicality Bound}\label{ssec:TypicalityPreciseAsymptotics}
The typicality upper bound on $P_e(n,\NLD)$ (Theorem \ref{thm:AchievabilityTypicality}) is typically weaker than the ML-based bound (Theorem \ref{thm:AchievabilityMaxLike}). In fact, it admits a weaker exponential behavior than the random coding exponent $\E_r(\NLD)$. Define the \emph{typicality exponent} $\E_t(\NLD)$ as
\begin{equation}
    \E_t(\NLD) \triangleq \NLD^*-\NLD -\frac{1}{2}\log(1+2(\NLD^*-\NLD)).
\end{equation}

\begin{theorem}\label{thm:TypicalityPreciseAsymptotics}
 For any $\NLD < \NLD^*$, the typicality upper bound is given asymptotically by
\begin{equation}\label{eqn:TypicalityPreciseAsymptotics}
    P_e^{TB}(n,\NLD) = \frac{e^{-n \E_t(\NLD)} }{\sqrt{n\pi}} \cdot \frac{1 + 2(\NLD^* - \NLD)}{2(\NLD^* - \NLD)}\left(1+O\left(\frac{1}{n}\right)\right)
\end{equation}
\begin{proof} Appendix~\ref{app:ProofAsymptoticsTypicality}.\end{proof}
\end{theorem}

The error exponent $\E_t(\NLD)$ is illustrated in Figure~\ref{fig:Exponents}. As seen in the figure, $\E_t(\NLD)$ is lower than $\E_r(\NLD)$ for all $\NLD$.

\begin{figure}
  \psfrag{esp}{$\E_{sp}(\delta)$}
\psfrag{er}{$\E_{r}(\delta)$}
\psfrag{et}{$\E_{t}(\delta)$}
  \begin{center}
  \includegraphics[width=6in]{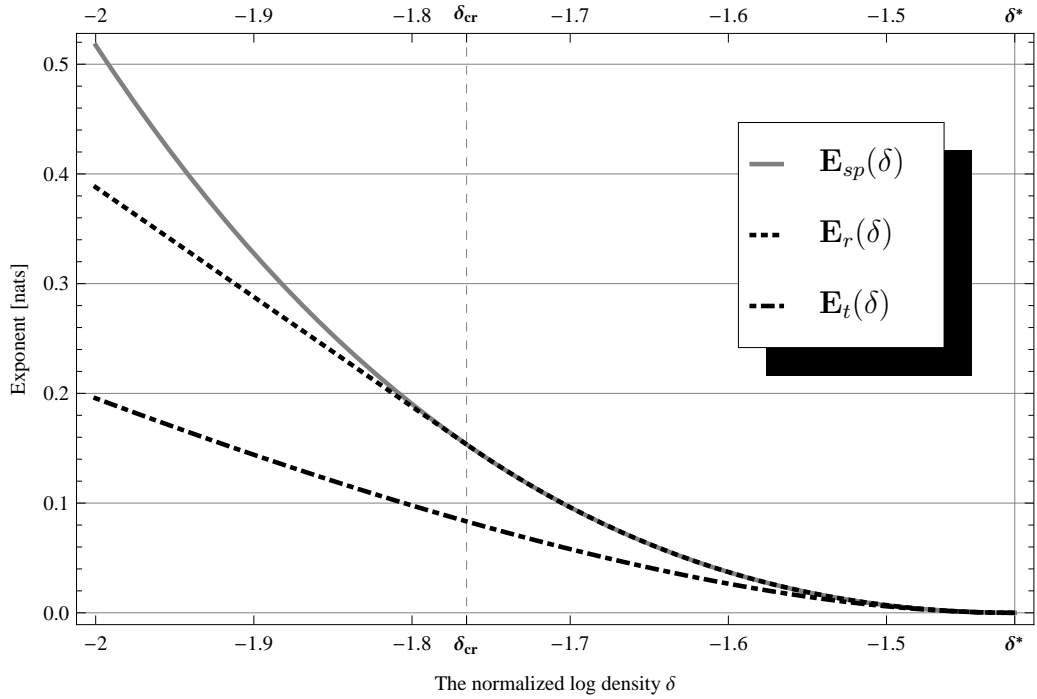}\\
    \caption{Error exponents for the unconstrained AWGN channel. The typicality error exponent $\E_t(\delta)$ (dot-dashed) vs. the random coding exponent $\E_r(\delta)$ (dotted) and the sphere packing $\E_{sp}(\delta)$ (solid). The noise variance $\sigma^2$ is set to 1.}\label{fig:Exponents}
  \end{center}
\end{figure}

\subsection{Asymptotic Analysis of $P_e^{MLB}$ with Poltyrev's $r=\sqrt n \sigma e^{\NLD^*-\NLD}$}\label{ssec:AsymptoticsMLPoltyrev}

In Poltyrev's proof of the random coding exponent \cite{Poltyrev94_CodingWithoutRestrictions}, the suboptimal value for $r$ was used, cf. Section \ref{sec:NewBounds} above. Instead of the optimal $r = \reff = e^{-\NLD}V_n^{1/n}$, he chose $r=\sqrt n \sigma e^{\NLD^*-\NLD}$. In Figures~\ref{fig:AllBoundsDelta15}~and~\ref{fig:AllBoundsDelta2} above we demonstrated how this suboptimal choice of $r$ affects the ML bound at finite $n$. In the figures, it is shown that for $\NLD=-1.5\mathsf{nat}$ (above $\NLD_{cr}$) the loss is more significant than for $\NLD=-2\mathsf{nat}$ (below $\NLD_{cr}$).
Here we utilize the techniques used in the current section in order to provide asymptotic analysis of the ML bound with the suboptimal $r$, and by that explain this phenomenon.

\begin{theorem}\label{thm:AsymptoticsMLPoltyrev}
The ML bound $P_e^{MLB}$, with $r=\sqrt n \sigma e^{\NLD^*-\NLD}$, denoted $\tilde P_e^{MLB}(n,\NLD)$, is given asymptotically as follows:

For $\NLD_{cr} < \NLD < \NLD^*$:
\begin{align}
    \tilde P_e^{MLB}(n,\NLD)
    &= e^{-n \E_r(\NLD)} \left[\frac{1}{n\pi (2-e^{2(\NLD^*-\NLD)})} + \frac{1}{\sqrt{n\pi}(e^{2(\NLD^*-\NLD)}-1)} \right]\left(1+O\left(\tfrac{1}{n}\right)\right)
    \label{eqn:PoltyrevBoundAnalysisAboveDeltaCrTight}\\
    &= e^{-n \E_r(\NLD)}
    \frac{1}{\sqrt{n\pi}(e^{2(\NLD^*-\NLD)}-1)}\left(1+O\left(\tfrac{1}{\sqrt n}\right)\right).
    \label{eqn:PoltyrevBoundAnalysisAboveDeltaCrLoose}
\end{align}

For $\NLD < \NLD_{cr}$:
\begin{equation}\label{eqn:PoltyrevBoundAnalysisBelowDeltaCr}
    \tilde P_e^{MLB}(n,\NLD) = e^{-n\E_r(\NLD)}\frac{1}{\sqrt{2 \pi n}}\left(1+O\left(\tfrac{1}{n}\right)\right).
\end{equation}

For $\NLD = \NLD_{cr}$:
\begin{equation}\label{eqn:PoltyrevBoundAnalysisAtDeltaCr}
    \tilde P_e^{MLB}(n,\NLD_{cr}) = e^{-n\E_r(\NLD_{cr})}\frac{1}{\sqrt{\pi n}}
    \left[1+\frac{1}{\sqrt 8}\right]
    \left(1+O\left(\tfrac{1}{n}\right)\right)
\end{equation}
\end{theorem}
Notes:
\begin{itemize}
  \item For $\NLD > \NLD_{cr}$, $\tilde P_e^{MLB}(n,\NLD)$ is indeed asymptotically worse than $P_e^{MLB}$ with the optimal $r=\reff$ \eqref{eqn:AchievabilityMaxLike}, see \eqref{eqn:MaxLikeLowerBoundAsymptotic}. Specifically, the choice of $r=\sqrt n \sigma e^{\NLD^*-\NLD}$ only balances the exponents of the two expressions of the bound \eqref{eqn:AchievabilityMaxLike}, while leaving the sub-exponential terms unbalanced - see \eqref{eqn:PoltyrevBoundAnalysisAboveDeltaCrTight}. The optimal selection $r=\reff$ balances the sub-exponential terms to the order of $n^{-\frac{1}{2}e^{2(\NLD^*-\NLD)}}$, see Theorem \ref{thm:MaxLikeAnalysis}. This in fact quantifies the asymptotic gap between the bounds, as seen in the Fig.~\ref{fig:AllBoundsDelta15}.
  \item For $\NLD < \NLD_{cr}$, the selection of the suboptimal $r$ has no asymptotic effect, as seen by comparing \eqref{eqn:PoltyrevBoundAnalysisBelowDeltaCr} and \eqref{eqn:MaxLikeAnalysisBelowDeltaCr}. This corroborates the numerical findings presented in Fig.~\ref{fig:AllBoundsDelta2}.
  \item For $\NLD = \NLD_{cr}$ the asymptotic form of the bound is changes by a constant (compare \eqref{eqn:PoltyrevBoundAnalysisAtDeltaCr} and \eqref{eqn:MaxLikeAnalysisAtDeltaCr},\eqref{eqn:MaxLikeAnalysisAtDeltaCrLoose}), and the correction term in the approximation tighter.
\end{itemize}
\begin{proof} Appendix~\ref{app:ProofAsymptoticsMLPoltyrev}.\end{proof}

\fi 

\section{Asymptotics for Fixed Error Probability}\label{sec:NormalApprox}
\if \NormalApprox 1

In the previous section we were interested in the asymptotic behavior of $P_e(n,\NLD)$ when the NLD $\NLD$ is fixed. We now turn to look at a related scenario where the error probability $\eps$ is fixed, and we are interested in the asymptotic behavior of the optimal achievable NLD, denoted $\NLD_\eps(n)$, with $n \ra \infty$. This setting parallels the channel dispersion type results \cite{Strassen62_Asymptotische}\cite{PolyanskiyPVFiniteLength10}\cite[Problem 2.1.24]{CsiszarKorner81}, and is strongly related to the dispersion of the power constrained AWGN channel \cite{PolyanskiyPV09_GaussianDispersion}\cite{PolyanskiyPVFiniteLength10}.

\subsection{The Dispersion of Infinite Constellations}

Let $\eps>0$ denote a fixed error probability value. Clearly, for any $\eps$, $\NLD_\eps(n)$ approaches the optimal NLD $\NLD^*$ as $n \ra \infty$. Here we study the asymptotic behavior of this convergence.
\begin{theorem}\label{thm:NormalApprox}
For a fixed error probability $\eps$, the optimal NLD $\NLD_\eps(n)$ is given, for large enough $n$, by
\begin{equation}\label{eqn:NormalApprox}
    \NLD_\eps(n) = \NLD^* - \sqrt\frac{1}{2n}Q^{-1}(\eps) +\frac{1}{2n}\log n + O\left(\frac{1}{n}\right).
\end{equation}
\end{theorem}

The proof is based on an asymptotic analysis of the finite-dimensional bounds derived in Section~\ref{sec:NewBounds}. Specifically, the converse bound (an upper bound in \eqref{eqn:NormalApprox}) is based on the sphere bound \eqref{eqn:SphereBoundFirst}. The achievability part (a lower bound in \eqref{eqn:NormalApprox}) is based on the ML bound \eqref{eqn:AchievabilityMaxLike}.
The weaker typicality bound is also useful for deriving a result of the type \eqref{eqn:NormalApprox}, but in a slightly weaker form - the typicality bound can only lead to
\begin{equation}\label{eqn:NLDnormalApproxLoose}
    \NLD_\eps(n) \geq \NLD^* - \sqrt\frac{1}{2n}Q^{-1}(\eps) + O\left(\frac{1}{n}\right).
\end{equation}

In Fig.~\ref{fig:NLDBoundsForDispersion} we show the bounds on $\NLD_\eps(n)$ that are derived from the finite dimensional bounds on $P_e(n,\NLD)$ given in Sec.~\ref{sec:NewBounds}, along with the asymptotic form \eqref{eqn:NormalApprox}, derived in this section, which tightly approximates $\NLD_\eps(n)$. In addition, the term \eqref{eqn:NLDnormalApproxLoose} is also depicted, which only loosely approximates $\NLD_\eps(n)$. The chosen error probability for the figure is $\eps = 0.01$.

\begin{figure}
  {\scriptsize
  \psfrag{zzPoltyrevCapacity}{Poltyrev capacity $\NLD^*$}
  \psfrag{zzPoltyrevCritical}{Critical rate $\NLD_{cr}$}
  \psfrag{zzConverse}{Converse}
  \psfrag{zzDirectML}{ Achievability (ML)}
  \psfrag{zzDirectT}{ Achievability (Typicality)}
  \psfrag{zzApprox}{Tight approximation \eqref{eqn:NormalApprox}
  $\NLD^* - \sqrt\frac{1}{2n}Q^{-1}(\eps) + \frac{1}{2n}\log n$}
  \psfrag{zzApprox2}{Loose approximation \eqref{eqn:NLDnormalApproxLoose}
  $\NLD^* - \sqrt\frac{1}{2n}Q^{-1}(\eps)$}
  \noindent\makebox[\textwidth]{
  \includegraphics[width=8in]{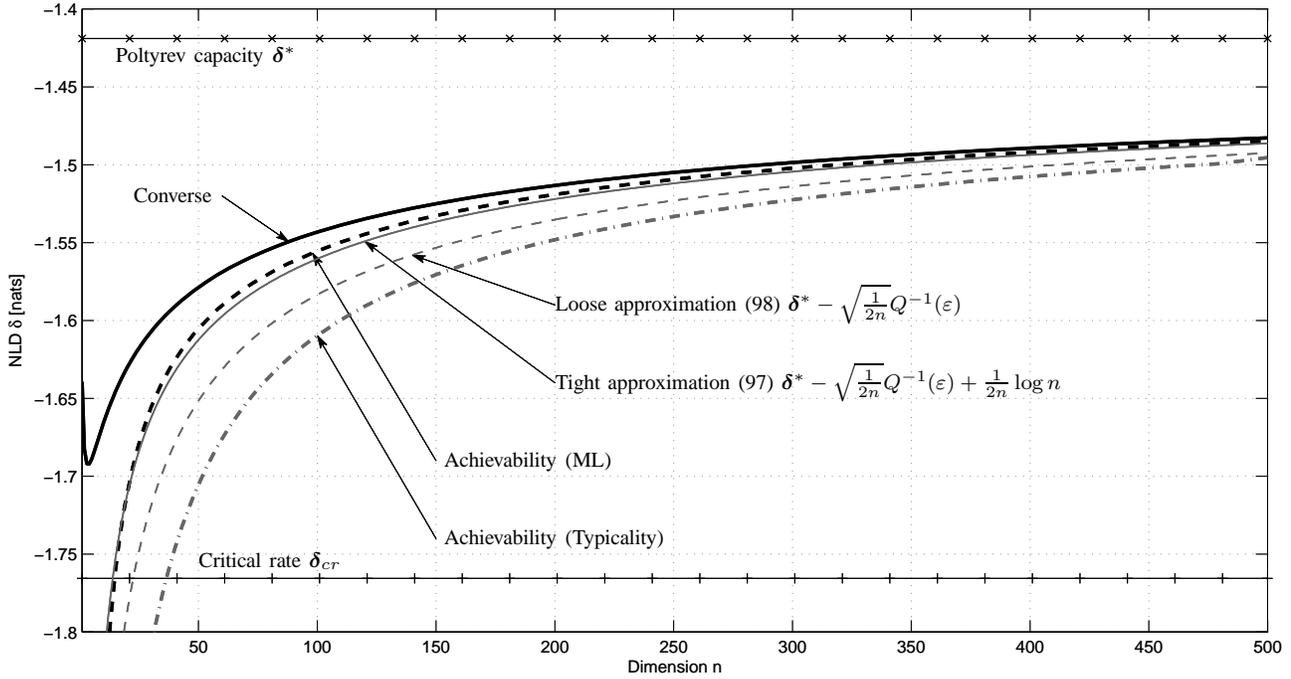}}\\
    }
  \caption{Bounds and approximations of the optimal NLD $\NLD_\eps(n)$ for error probability $\eps = 0.01$. Here the noise variance $\sigma^2$ is set to $1$.}\label{fig:NLDBoundsForDispersion}
\end{figure}

\bigskip

Before proving the theorem, let us discuss the result. By the similarity of Equations \eqref{eqn:DMCdispersion} and \eqref{eqn:NormalApprox} we can isolate the constant $\frac{1}{2}$ and identify it as the dispersion of the unconstrained AWGN setting. This fact can be intuitively explained from several directions.

One interesting property of the channel dispersion theorem \eqref{eqn:DMCdispersion} is the following connection to the error exponent. Under some mild regularity assumptions, the error exponent can be approximated near the capacity by
\begin{equation}\label{eqn:ApproxErrExpByDispersion}
    \E(R) \cong \frac{(C-R)^2}{2V},
\end{equation}
where $V$ is the channel dispersion.
The fact that the error exponent can be approximated by a parabola with second derivative $\frac{1}{V}$ was already known to Shannon (see \cite[Fig. 18]{PolyanskiyPVFiniteLength10}). This property holds for DMC's and for the power constrained AWGN channel and is conjectured to hold in more general cases.
Note, however, that while the parabolic behavior of the exponent hints that the gap to the capacity should behave as $O\left(\frac{1}{\sqrt{n}}\right)$, the dispersion theorem
cannot be derived directly from the error exponent theory. Even if the error probability was given by $e^{-n\E(R)}$ exactly, \eqref{eqn:DMCdispersion} cannot be deduced from \eqref{eqn:ApproxErrExpByDispersion} (which holds only in the Taylor approximation sense).

Analogously to \eqref{eqn:ApproxErrExpByDispersion}, we examine the error exponent for the unconstrained Gaussian setting. For NLD values above the critical NLD $\NLD_{cr}$ (but below $\NLD^*$), the error exponent is given by
\cite{Poltyrev94_CodingWithoutRestrictions}:
\begin{equation}\label{eqn:PoltyrevErrExp}
    \E(\NLD) = \frac{e^{-2\NLD}}{4 \pi e \sigma^2} + \NLD + \frac{1}{2}\log 2\pi\sigma^2.
\end{equation}
By straightforward differentiation we get that the second derivative (w.r.t. $\NLD$) of $\E(\NLD,\sigma^2)$ at $\NLD = \NLD^*$ is given by $2$, so according to \eqref{eqn:ApproxErrExpByDispersion}, it is expected that the dispersion for the unconstrained AWGN channel will be $\frac{1}{2}$. This agrees with our result \eqref{eqn:NormalApprox} and its similarity to \eqref{eqn:DMCdispersion}, and extends the correctness of the conjecture \eqref{eqn:ApproxErrExpByDispersion} to the unconstrained AWGN setting as well. It should be noted, however, that our result provides more than just proving the conjecture: there also exist examples where the error exponent is well defined (with second derivative), but a connection of the type \eqref{eqn:ApproxErrExpByDispersion} can only be achieved asymptotically with $\eps \ra 0$ (see, e.g. \cite{IngberFederPBICM_full}). Our result \eqref{eqn:NormalApprox} holds for any finite $\eps$, and also gives the exact $\frac{1}{n}\log n$ term in the expansion.

Another indication that the dispersion for the unconstrained setting should be $\frac{1}{2}$ comes from the connections to the power constrained AWGN. While the capacity $\frac{1}{2}\log(1+P)$, where $P$ denotes the channel SNR, is clearly unbounded with $P$, the form of the error exponent curve does have a nontrivial limit as $P \ra \infty$. In \cite{ErezZamirAWGN} it was noticed that this limit is the error exponent of the unconstrained AWGN channel (sometimes termed the `Poltyrev exponent'), where the distance to the capacity is replaced by the NLD distance to $\NLD^*$.
By this analogy we examine the dispersion of the power constrained AWGN channel at high SNR. In \cite{PolyanskiyPVFiniteLength10} the dispersion was found, given (in $\mathsf{nat}^2$ per channel use) by
\begin{equation}\label{eqn:PCAWGNdispersion}
    V_{AWGN} = \frac{P(P+2)}{2(P+1)^2}.
\end{equation}
This term already appeared in Shannon's 1959 paper on the AWGN error exponent \cite{Shannon59Gaussian}, where its inverse is exactly the second derivative of the error exponent at the capacity (i.e. \eqref{eqn:ApproxErrExpByDispersion} holds for the AWGN channel). It is therefore no surprise that by taking $P\ra\infty$, we get the desired value of $\frac{1}{2}$, thus completing the analogy between the power constrained AWGN and its unconstrained version. This convergence is quite fast, and is tight for SNR as low as $10$dB (see Fig. \ref{fig:AWGNdispersion}).

\begin{figure}\begin{center}
  \includegraphics[width=5in]{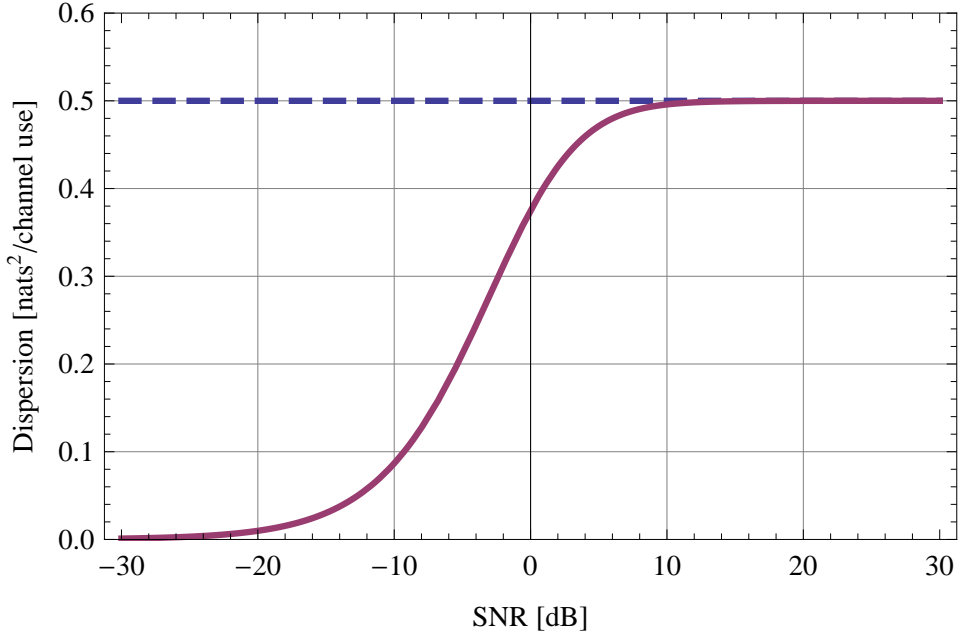}\\
  \caption{The power-constrained AWGN dispersion \eqref{eqn:PCAWGNdispersion} (solid) vs. the unconstrained dispersion (dashed)}\label{fig:AWGNdispersion}
\end{center}\end{figure}

\subsection{A Key Lemma}

In order to prove Theorem~\ref{thm:NormalApprox} we need the following lemma regarding the norm of a Gaussian vector.
\begin{lem}\label{lem:GaussianOutOfSphereBerryEssen}
Let $\bZ = [Z_1,...,Z_n]^T$ be a vector of $n$ zero-mean, independent Gaussian random variables, each with mean $\sigma^2$. Let $r>0$ be a given arbitrary radius. Then the following holds for any dimension $n$:
\begin{equation}
    \left|\Pr\{ \|\bZ\| > r \} - Q\left( \frac{r^2 - n\sigma^2}{\sigma^2 \sqrt{2n}}\right) \right| \leq \frac{6T}{\sqrt n},
\end{equation}
where
\begin{equation}\label{eqn:T}
    T = \E\left[\left|\frac{X^2-1}{\sqrt 2}\right|^3\right]\approx 3.0785,
\end{equation}
for a Standard Gaussian RV $X$.
\begin{proof} The proof relies on the convergence of a sum of independent random variables to a Gaussian random variable, i.e. the central limit theorem. We first note that
\begin{align}
  \Pr\{ \|\bZ\| > r \} &=
     \Pr\left\{ \sum_{i=1}^nZ_i^2 > r^2 \right\}.
\end{align}

Let $Y_i = \frac{Z_i^2-\sigma^2}{\sigma^2\sqrt 2}$. It is easy to verify that $\E[Y_i]=0$ and that $\VAR[Y_i]=1$. Let $S_n \triangleq \frac{1}{\sqrt{n}}\sum_{i=1}^{n}Y_i$. Note that $S_n$ also has zero mean and unit variance. It follows that
\begin{align}
    \Pr\left\{\sum_{i=1}^n Z_i^2 \geq r^2\right\}&=
    \Pr\left\{\sum_{i=1}^n \frac{Z_i^2-\sigma^2}{\sigma^2 \sqrt 2} \geq \frac{r^2-n\sigma^2}{\sigma^2 \sqrt 2}\right\}\nonumber\\
    &=\Pr\left\{\sum_{i=1}^n Y_i \geq \frac{r^2-n\sigma^2}{\sigma^2 \sqrt 2}\right\}\nonumber\\
    &= \Pr\left\{S_n \geq \frac{r^2-n\sigma^2}{\sigma^2 \sqrt{2n}}\right\}.
\end{align}

$S_n$ is a normalized sum of i.i.d. variables, and by the central limit theorem converges to a standard Gaussian random variables. The Berry-Esseen theorem (see Appendix \ref{app:CLT}) quantifies the rate of convergence in the cumulative distribution function sense. In the specific case discussed in the lemma we get

\begin{equation}
    \left| \Pr\left\{S_n \geq \frac{r^2-n\sigma^2}{\sigma^2 \sqrt{2n}}\right\} -
     Q\left[ \frac{r^2-n\sigma^2}{\sigma^2 \sqrt{2n}}\right]\right| \leq \frac{6T}{\sqrt n},
\end{equation}
where $T=\E[|Y_i|^3]$. Note that $T$ is independent of $\sigma^2$, finite, and can be evaluated numerically to about $3.0785$.
\end{proof}
\end{lem}


\subsection{Proof of Theorem \ref{thm:NormalApprox}}
\begin{proof}[Proof of Direct part]

Let $\eps$ denote the required error probability. We shall prove the existence of an IC (more specifically, a lattice) with error probability at most $\eps$ and NLD satisfying \eqref{eqn:NormalApprox}.

It is instructive to first prove a slightly weaker version of \eqref{eqn:NormalApprox} based on the typicality decoder (Theorem~\ref{thm:AchievabilityTypicality}). While easier to derive, this will show the existence of lattices with NLD $\NLD = \NLD^* - \sqrt\frac{1}{2n}Q^{-1}(\eps)+ O\left(\frac{1}{n}\right)$. Proving the stronger result \eqref{eqn:NormalApprox} is more technical and will proven afterwards using the ML achievability bound (Theorem \ref{thm:AchievabilityMaxLike}).

\bigskip

Recall the achievability bound in Theorem \ref{thm:AchievabilityTypicality}: for any $r>0$ there exist lattices with NLD $\NLD$ and error probability $P_e$ that is upper bounded by
\begin{equation}
    P_e \leq \gamma V_n r^n + \Pr\left\{ \|\bZ\| > r \right\}.
\end{equation}

We determine $r$ s.t. $\Pr(\|\bZ\|>r) = \eps\left[1-\frac{1}{\sqrt n} \right]$ and $\gamma$ s.t. $\gamma V_n r^n= \frac{\eps}{\sqrt n}$. This way it is assured that the error probability is not greater than the required $\eps\left[1-\frac{1}{\sqrt n} \right] + \frac{\eps}{\sqrt n} = \eps$. Now define $\alpha_n$ s.t. $r^2 = n\sigma^2(1+\alpha_n)$ (note that $r$ implicitly depends on $n$ as well).

\begin{lem}\label{lem:alpha_n}
$\alpha_n$, defined above, is given by
\begin{equation}
    \alpha_n = \sqrt\frac{2}{n} Q^{-1}(\eps) + O\left(\frac{1}{n}\right).
\end{equation}
    \begin{proof} By construction, $r$ is chosen s.t.
    \begin{equation}
        \Pr(\|\bZ\|^2>r^2) = \eps\left[1-\frac{1}{\sqrt n}\right].
    \end{equation}
    By the definition of $\alpha_n$,
    \begin{equation}\label{eqn:alphadef}
        \Pr(\|\bZ\|^2>n\sigma^2 (1+\alpha_n)) = \eps\left[1-\frac{1}{\sqrt n}\right].
    \end{equation}
    By Lemma \ref{lem:GaussianOutOfSphereBerryEssen},
    \begin{align}
        \Pr(\|\bZ\|^2>n\sigma^2 (1+\alpha_n))
        &= Q\left( \frac{n\sigma^2 (1+\alpha_n) - n\sigma^2}{\sigma^2 \sqrt{2n}}\right) + O\left(\frac{1}{\sqrt n}\right)\nonumber\\
        &= Q\left( \sqrt{\frac{n}{2}} \alpha_n\right) + O\left(\frac{1}{\sqrt n}\right).
    \end{align}
    Combined with \eqref{eqn:alphadef}, we get
    \begin{equation}
        \eps\left[1-\frac{1}{\sqrt n}\right] = Q\left( \sqrt{\frac{n}{2}} \alpha_n\right) + O\left(\frac{1}{\sqrt n}\right),
    \end{equation}
    or
    \begin{equation}
        \eps+ O\left(\frac{1}{\sqrt n}\right) = Q\left( \sqrt{\frac{n}{2}} \alpha_n\right).
    \end{equation}
    Taking $Q^{-1}(\cdot)$ of both sides, we get
    \begin{equation}
        \sqrt{\frac{n}{2}} \alpha_n = Q^{-1}\left(\eps+ O\left(\frac{1}{\sqrt n}\right)\right).
    \end{equation}
    By the Taylor approximation of $Q^{-1}(\eps + x)$ around $x=0$, we get
    \begin{equation}
        \sqrt{\frac{n}{2}} \alpha_n = Q^{-1}\left(\eps \right) + O\left(\frac{1}{\sqrt n}\right),
    \end{equation}
    or
    \begin{equation}
        \alpha_n = \sqrt{\frac{2}{n}} Q^{-1}\left(\eps \right) + O\left(\frac{1}{n}\right),
    \end{equation}
    as required. \end{proof}
\end{lem}

So far, we have shown the existence of a lattice $\Lambda$ with error probability at most $\eps$.
The NLD is given by
\begin{align*}
    \NLD &=\frac{1}{n}\log\gamma\\
    &= \frac{1}{n}\log \frac{\eps}{ V_n r^n \sqrt n} \\
    &= -\frac{1}{n}\log V_n - \log r - \frac{\log n}{2n} +\frac{1}{n}\log \eps\\
    &= -\frac{1}{n}\log V_n - \frac{1}{2}\log [n \sigma^2 (1+\alpha_n)] - \frac{\log n}{2n} +\frac{1}{n} \log\eps.
\end{align*}

$V_n$ can be approximated by (see Appendix \ref{app:Vn}) by
\begin{equation}\label{eqn:VnApprox}
    \frac{1}{n}\log V_n = \frac{1}{2}\log \frac{2 \pi e}{n} - \frac{1}{2n}\log n + O\left(\frac{1}{n}\right).
\end{equation}
We therefore have
\begin{align}
    \NLD &= -\frac{1}{2}\log (2 \pi e \sigma^2) - \frac{1}{2}\log  (1+\alpha_n)  + O\left(\frac{1}{n}\right)\\
    &\overset{(a)}{=} \NLD^* - \frac{1}{2}\log  (1+\alpha_n)  + O\left(\frac{1}{n}\right)\\
    &\overset{(b)}{=} \NLD^* - \frac{1}{2}\alpha_n + O\left(\frac{1}{n}\right)\\
    &\overset{(c)}{=} \NLD^* - \sqrt\frac{1}{2n}Q^{-1}(\eps) + O\left(\frac{1}{n}\right),
\end{align}
where $(a)$ follows from the definition of $\NLD^*$, $(b)$ follows from the Taylor approximation for $\log(1+\alpha_n)$ around $\alpha_n=0$ and from the fact that $\alpha_n = O(1/\sqrt n)$, and $(c)$ follows from Lemma \ref{lem:alpha_n}. This completes the achievability part based on the typicality decoder.

\bigskip

In order to prove the stronger achievability result that fits \eqref{eqn:NormalApprox} we follow the same steps with the ML achievability bound. By Theorem \ref{thm:AchievabilityMaxLike} we get that for any $r>0$ there exist a lattice with density $\gamma$ and error probability upper bounded by
\begin{equation}
    P_e \leq \gamma V_n \int_0^r f_R(\tilde r) \tilde r^n d\tilde r + \Pr\left\{ \|\bZ\| > r \right\}.
\end{equation}

Now determine $r$ s.t. $\Pr(\|\bZ\|>r) = \eps\left[1-\frac{1}{\sqrt n} \right]$ and $\gamma$ s.t. $\gamma V_n \int_0^r f_R(\tilde r)  \tilde r^n  d\tilde r = \frac{\eps}{\sqrt n}$. Again, it is assured that the error probability is upper bounded by $\eps$. Define $\alpha_n$ s.t. $r^2 = n\sigma^2(1+\alpha_n)$.

The resulting density is given by
\begin{equation}
    \gamma = \frac{\eps}{\sqrt n V_n \int_0^r f_R(\tilde r)  \tilde r^n  d\tilde r},
\end{equation}
and the NLD by
\begin{align}
    \NLD &= \frac{1}{n}\log \gamma\nonumber\\
     &= \frac{1}{n}\log \eps - \frac{1}{n}\log \left[\sqrt n V_n \int_0^r f_R(\tilde r)  \tilde r^n  d\tilde r\right]\nonumber\\
     &= \frac{1}{n}\log \eps - \frac{1}{2n}\log n -  \frac{1}{n}\log V_n  -  \frac{1}{n}\log\int_0^r f_R(\tilde r)  \tilde r^n  d\tilde r\nonumber\\
    &= - \frac{1}{2} \log \frac{2 \pi e}{n} -  \frac{1}{n}\log\int_0^{\sqrt{n(1+\alpha_n)}} f_R(\tilde r)  \tilde r^n  d\tilde r + O\left(\frac{1}{n}\right).\label{eqn:preApproxIntegral}
\end{align}
where the last equality follows from the approximation \eqref{eqn:VnApprox} for $V_n$.

We repeat the derivation as in Eq. \eqref{eqn:preMLIntegralBounds} where $r^*$ is replaced by $r=\sqrt{n\sigma^2(1+\alpha_n)}$ and have
\begin{align*}
    \int_0^{\sqrt{n\sigma^2(1+\alpha_n)}} f_R(\tilde r) \tilde r^n d\tilde r
    &= \sigma^n \frac{2^{-n/2}n^n}{\Gamma\left[\frac{n}{2}\right]} \int_0^{1+\alpha_n} e^{-n\tilde r/2} \tilde r^{n-1} d\tilde r\\
    &\leq \sigma^n \frac{2^{-n/2}n^n}{\Gamma\left[\frac{n}{2}\right]}
    \frac{2(1+\alpha_n)^ne^{-n(1+\alpha_n)/2}}{n\left(1-\alpha_n-\tfrac{2}{n}\right)}\\
    &= \sigma^n \frac{2^{-n/2}n^n}{\Gamma\left[\frac{n}{2}\right]}
    \frac{   2e^{n\log(1+\alpha_n)}   e^{-n(1+\alpha_n)/2}  }
    {n\left(1-\alpha_n-\tfrac{2}{n}\right)},
 \end{align*}
where the inequality follows from Lemma~\ref{lem:MaxLikeIntegralBounds}. Therefore the term in \eqref{eqn:preApproxIntegral} can be bounded by
\begin{align*}
    &\frac{1}{n}\log\int_0^{\sqrt{n\sigma^2(1+\alpha_n)}} f_R(\tilde r) \tilde r^n d\tilde r \\
    &\leq
    \frac{1}{2}\log\sigma^2 -\tfrac{1}{2}\log2 +\log n + \log(1+\alpha_n) - \tfrac{1}{2}(1+\alpha_n) +
    \frac{1}{n}\log \frac{1}{\tfrac{n}{2}\Gamma\left[\frac{n}{2}\right]}
    \frac{1}{\left(1-\alpha_n-\tfrac{2}{n}\right)}\\
    &=
    \frac{1}{2}\log\sigma^2 -\tfrac{1}{2}\log2 +\log n + \log(1+\alpha_n) - \tfrac{1}{2}(1+\alpha_n) +
    \frac{1}{n}\log \frac{1}{\tfrac{n}{2}\Gamma\left[\frac{n}{2}\right]}
    + O\left(\tfrac{1}{n}\right)\\
    &=
    \frac{1}{2}\log\sigma^2 -\tfrac{1}{2}\log2 +\log n + \log(1+\alpha_n) - \tfrac{1}{2}(1+\alpha_n) -
    \tfrac{1}{n}\left(\tfrac{1}{2}\log(\pi n) + \tfrac{n}{2}\log\tfrac{n}{2e}\right)
    + O\left(\tfrac{1}{n}\right)\\
    &=
    \frac{1}{2}\log\sigma^2  +\tfrac{1}{2}\log n + \log(1+\alpha_n) - \tfrac{1}{2}\alpha_n
     -\tfrac{1}{2n}\log n
    + O\left(\tfrac{1}{n}\right)\\
    &\overset{(a)}{=}
    \frac{1}{2}\log\sigma^2  +\tfrac{1}{2}\log n +  \tfrac{1}{2}\alpha_n
     -\tfrac{1}{2n}\log n
    + O\left(\tfrac{1}{n}\right)\\
 \end{align*}
$(a)$ follows from the Taylor expansion of $\log(1+\xi)$ at $\xi = 0$ and from the fact that $\alpha_n = O(\frac{1}{\sqrt n})$.
Plugging back to \eqref{eqn:preApproxIntegral} combined with Lemma \ref{lem:alpha_n} completes the proof of the direct part.

\bigskip

\emph{Proof of Converse Part:}

Let $\eps>0$, and let $\left\{\S_n\right\}_{n\in\Naturals}$ be a series of IC's, where for each $n$, $P_e(\S_n) \leq \eps$. Our goal is to upper bound the NLD $\NLD_n$ of $\S_n$ for growing $n$.

By the sphere bound we have
\begin{equation}
    \eps \geq P_e(\S_n) \geq \Pr\{\|\bZ\|> r^*\},
\end{equation}
where $r^* = e^{-\NLD_n}V_n^{-1/n}$. By Lemma~\ref{lem:GaussianOutOfSphereBerryEssen},
\begin{equation}
    \eps \geq \Pr\{\|\bZ\|> r^*\} \geq Q\left( \frac{r^{*2} - n\sigma^2}{\sigma^2 \sqrt{2n}}\right) - \frac{6T}{\sqrt n},
\end{equation}
where $T$ is a constant. It follows by algebraic manipulations that
\begin{equation}
    \NLD_n  \leq -\frac{1}{2}\log\left(1 + \sqrt\frac{2}{n}Q^{-1}\left(\eps +\frac{6T}{\sqrt n}\right)\right) - \frac{1}{n}\log V_n - \frac{1}{2}\log (n\sigma^2).
\end{equation}
By the Taylor approximation of $\log(1+x)$ at $x=0$ and of $Q^{-1}(y)$ at $y=\eps$, and by the approximation \eqref{eqn:VnApprox} for $V_n$,
\begin{equation}
     \NLD_n  \leq -\sqrt\frac{1}{2n}Q^{-1}\left(\eps\right) - \frac{1}{2}\log \frac{2 \pi e}{n} +\frac{1}{2n}\log n - \frac{1}{2}\log (n\sigma^2)+ O\left(\frac{1}{n}\right).
\end{equation}
By the definition of $\NLD^*$ we finally get
\begin{equation}
     \NLD_n  \leq \NLD^*  -\sqrt\frac{1}{2n}Q^{-1}\left(\eps\right) +\frac{1}{2n}\log n + O\left(\frac{1}{n}\right),
\end{equation}
as required.
\end{proof}

\fi 

\section{Comparison with Known Infinite Constellations}\label{sec:Comparison}
\if \Comparison 1

In this section we compare the finite dimensional bounds of Section ~\ref{sec:NewBounds} with the performance of some known IC's.

We start with the low dimensional IC's, which include classic sphere packings: the integer lattice, the hexagonal lattice, the packings $D_4$, $E_8$, $BW_{16}$ and the leech lattice $\Lambda_{24}$ (see Conway and Sloane \cite{ConwaySloane1993}). In low dimensions we present Monte Carlo simulation results based on the ML decoder.
In higher dimensions we consider low density lattice codes (LDLC) \cite{LDLC_IT2008} with dimensions $n=100,300,500$ and $1000$ (designed by Y. Yona). In dimension $n=127$ we present the results for the packing $S_{127}$ \cite{AgrawalVardy2000GMD}.

In Fig.~\ref{fig:NLDBoundsAndLowDimLattices} we show the gap to (Poltyrev's) capacity of the low dimensional IC's, where the error probability is set to $\eps=0.01$. As seen in the figure, these low dimensional IC's outperform the best achievability bound (Theorem~\ref{thm:AchievabilityMaxLike}). At $n=1$, the integer lattice achieves the sphere bound (and is, essentially, the only lattice for $n=1$).

\begin{figure}
  {\scriptsize
  \psfrag{zzPoltyrevCapacity}{Poltyrev capacity $\NLD^*$}
  \psfrag{zzPoltyrevCritical}{Critical rate $\NLD_{cr}$}
  \psfrag{zzConverse}{Converse}
  \psfrag{zzDirectML}{ Achievability (ML)}
  \psfrag{zzApprox}{Approximation \eqref{eqn:NormalApprox}}
\noindent\makebox[\textwidth]{
  \includegraphics[width=8in]{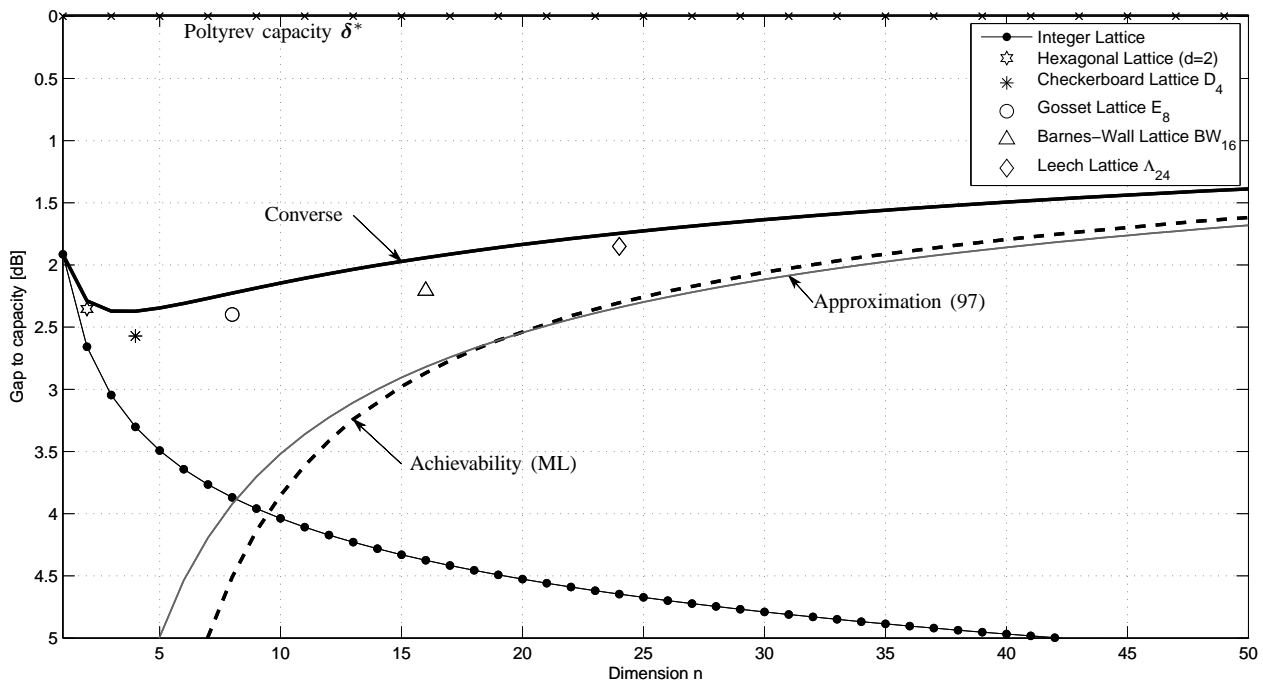}}}\\
  \caption{Low-dimensional IC's for coding over the unconstrained AWGN. The error probability is set to  $\eps=0.01$.}\label{fig:NLDBoundsAndLowDimLattices}
\end{figure}

\bigskip

From the presentation of Fig.~\ref{fig:NLDBoundsAndLowDimLattices} it is difficult to compare IC's with different dimensions. For example, the hexagonal lattice closer to the capacity than the lattice $D_4$, and also the gap to the sphere bound is smaller. Obviously this does not mean that $D_4$ is inferior. To facilitate the comparison between different dimensions we propose the following comparison: Set a fixed value for the error probability for $n=1$ denoted $\eps_1$. Then define, for each $n$, the normalized error probability
$$\eps_n \triangleq  1 - (1-\eps_1)^n.$$
Using this normalization enables the true comparison between IC's with different dimensions. The achieved gap to capacity with a normalized error probability remains the same when a scheme is used say $k$ times, and the block length becomes $k\cdot n$. For example, the integer lattice maintains a constant $\NLD$ for any $n$ with the normalized error probability, as opposed to the case presented in Fig.~\ref{fig:NLDBoundsAndLowDimLattices}, where the performance \emph{decreases}. In Fig.~\ref{fig:NLDBoundsAndLowDimLatticesNormalized} we plot the same data as in Fig.~\ref{fig:NLDBoundsAndLowDimLattices} for normalized error probability with $\eps_1 = 10^{-5}$. We also plot the normalized error probability itself for reference. In Fig.~\ref{fig:NLDBoundsAndMidDimLatticesNormalized} we present the performance of IC's in higher dimensions (again, with normalized error probability and $\eps_1 = 10^{-5}$). The included constellations are the leech lattice again (for reference), LDLC with $n=100,300,500,1000$ and degrees $5,6,7,7$ respectively (cf. \cite{LDLC_IT2008} and \cite{YairISIT2009} for more details on the construction of LDLC and the degree). For LDLC's, the figure shows simulation results based on a suboptimal low complexity parametric iterative decoder \cite{YairISIT2009}. In addition, we present the performance of the packing $S_{127}$\cite{AgrawalVardy2000GMD} (which is a multilevel coset code\cite{ForneySphereBound00}).

Notes:
\begin{itemize}
  \item At higher dimensions, the performance of the presented IC's no longer outperforms the achievability bound.
  \item It is interesting to note that the Leech lattice replicated 4 times (resulting in an IC at $n=96$) outperforms the LDLC with $n=100$ (but remember that the LDLC performance is based on a low complexity suboptimal decoder where the Leech lattice performance is based on the ML decoder).
  \item The approximation \eqref{eqn:NormalApprox} no longer holds formally for the case of normalized error probability. This follows since the correction term in \eqref{eqn:NormalApprox} depends on the error probability. Nevertheless, as appears in Fig.~\ref{fig:NLDBoundsAndMidDimLatticesNormalized}, the approximation appears to still hold.
\end{itemize}

\begin{figure}
  {\scriptsize
  \psfrag{zzPoltyrevCapacity}{Poltyrev capacity $\NLD^*$}
  \psfrag{zzPoltyrevCritical}{Critical rate $\NLD_{cr}$}
  \psfrag{zzConverse}{Converse}
  \psfrag{zzDirectML}{\hspace{-.5in}Achievability (ML)}
  \psfrag{zzApprox}{Approximation \eqref{eqn:NormalApprox}}
\noindent\makebox[\textwidth]{
  \includegraphics[width=8in]{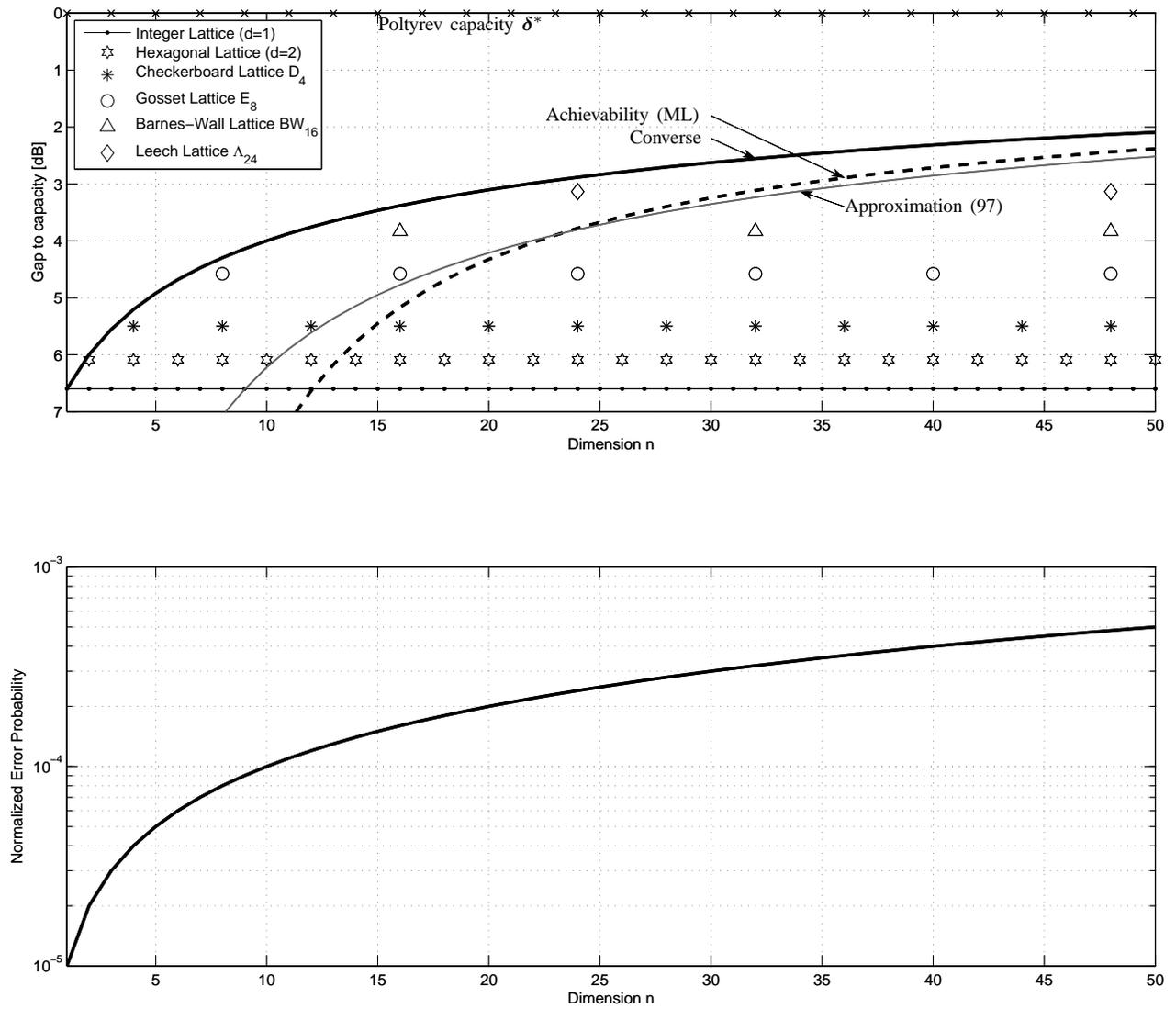}}}\\
  \caption{Top: performance of different constellations (dimensions $1-24$) for normalized error probability, with $\eps_1=10^{-5}$. Bottom: the normalized error probability.}\label{fig:NLDBoundsAndLowDimLatticesNormalized}
\end{figure}

\begin{figure}
  {\scriptsize
  \psfrag{zzPoltyrevCapacity}{Poltyrev capacity $\NLD^*$}
  \psfrag{zzPoltyrevCritical}{Critical rate $\NLD_{cr}$}
  \psfrag{zzConverse}{ Converse}
  \psfrag{zzDirectML}{ Achievability (ML)}
  \psfrag{zzApprox}{Approximation \eqref{eqn:NormalApprox}}
\noindent\makebox[\textwidth]{
  \includegraphics[width=8in]{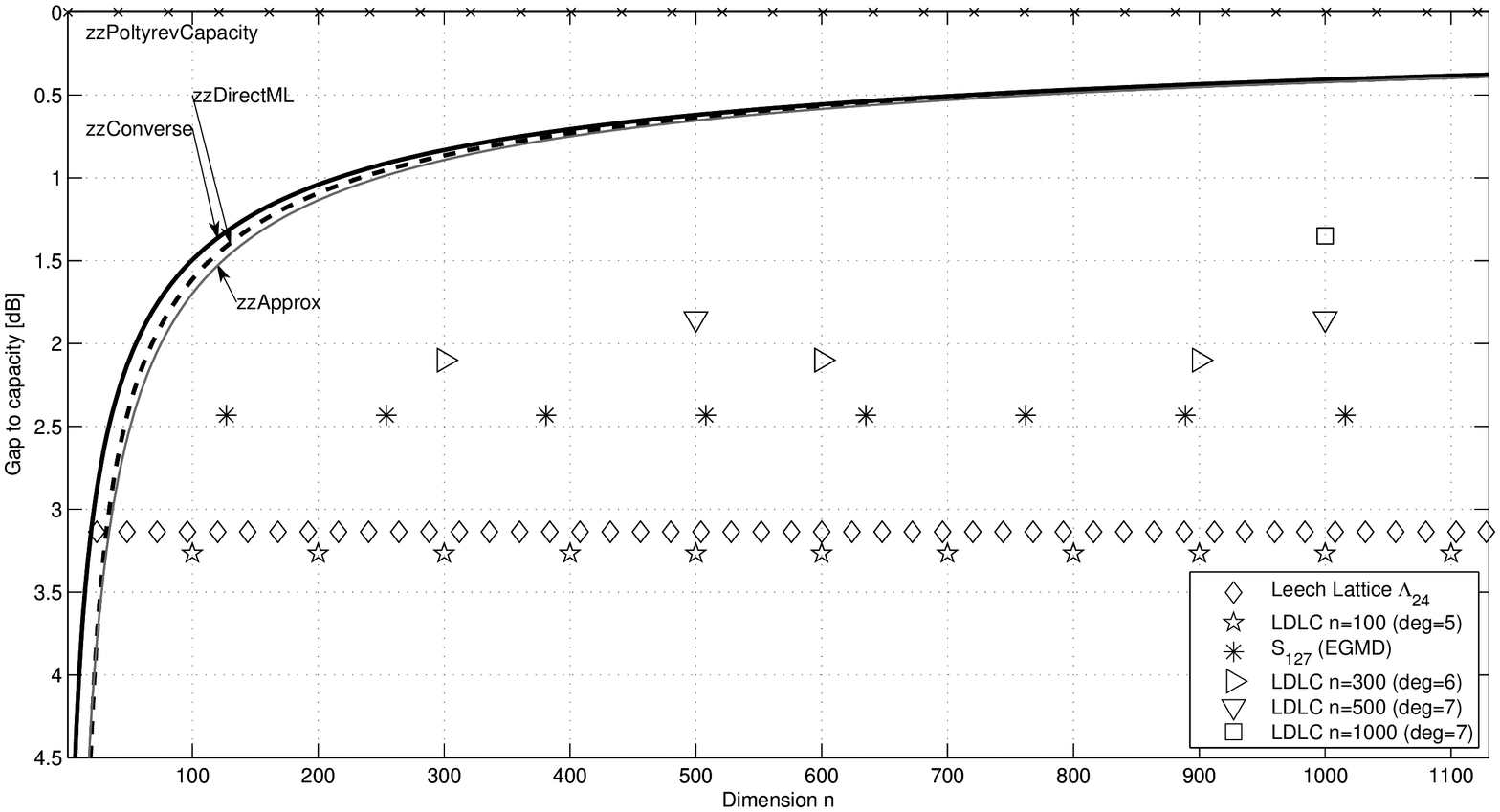}}}\\
  \caption{Performance of different constellations (dimensions $24-1000$) for normalized error probability, with $\eps_1=10^{-5}$.}\label{fig:NLDBoundsAndMidDimLatticesNormalized}
\end{figure}

\fi 

\section{Volume-to-Noise Ratio Analysis}\label{sec:VNR}
\if \VNR 1

The VNR $\mu$, defined in \eqref{eqn:gapToCapacityRatioSigma}, can describe the distance from optimality for a given IC and noise variance, and we say that an IC $\S$ operating at noise level $\sigma^2$ is in fact operating at VNR $\mu$. Equivalently, we can define the VNR as a function of the IC and the error probability: For a given IC $\S$ and error probability $\eps$, let $\mu(\S,\eps)$ be defined as follows:
\begin{equation}\label{eqn:defVNR}
    \mu(\S, \eps) \triangleq \frac{e^{-2\NLD(\S)}}{2\pi e \sigma^2(\eps)},
\end{equation}
where $\sigma^2(\eps)$ is the noise variance s.t. the error probability is exactly $\eps$. Note that $\mu(\S, \eps)$ does not depend on scaling of the IC $\S$, and therefore can be thought of as a quantity that depends only on the `shape' of the IC.

We now define a related fundamental quantity $\mu_n(\eps)$, as the minimal value of $\mu(\S, \eps)$ among all $n$-dimensional IC's. It is known that for any $0 < \eps < 1$, $\mu_n(\eps) \ra 1$ as $n \ra \infty$ \cite{ZamirLatticesEverywhere09}. We now quantify this convergence, based on the analysis of $\NLD_\eps(n)$.
It follows from the definitions of $\mu_n(\eps)$ and $\NLD_\eps(n)$ that the following relation holds for any $\sigma^2$:
\begin{equation}\label{eqn:mu-NLD}
    \mu_n(\eps) = \frac{e^{-2\NLD_\eps(n)}}{2\pi e \sigma^2} = e^{2(\NLD^* - \NLD_\eps(n))}.
\end{equation}
(note that $\NLD_\eps(n)$ implicitly depends on $\sigma^2$ as well). We may therefore use the results in the paper to understand the behavior of $\mu_n(\eps)$. For example, any of the bounds in Theorem~\ref{thm:AchievabilityTypicality}, Theorem~\ref{thm:AchievabilityMaxLike} or the sphere bound (Theorem~\ref{thm:SphereBound}) can be applied in order to bound $\mu_n(\eps)$ for finite $n$ and $\eps$. Furthermore, the asymptotic behavior of $\mu_n(\eps)$ is characterized by the following:

\begin{theorem}
For a fixed error probability $0<\eps<1$, The optimal VNR $\mu_n(\eps)$ is given by
\begin{equation}\label{eqn:VNR}
    \mu_n(\eps) = 1 + \sqrt\frac{2}{n}Q^{-1}(\eps) - \frac{1}{n} \log n+ O\left(\frac{1}{n}\right).
\end{equation}
\begin{proof}
In Theorem~\ref{thm:NormalApprox} we have shown that for given $\eps$ and $\sigma^2$, the optimal NLD $\NLD$ is given by
\begin{equation}
    \NLD_\eps(n) = \NLD^* - \sqrt\frac{1}{2n}Q^{-1}(\eps) + \frac{1}{2n}\log n + O\left(\frac{1}{n}\right),
\end{equation}
where $\NLD^* = \frac{1}{2}\log\frac{1}{2 \pi e \sigma^2}$.

According to \eqref{eqn:mu-NLD} we write
\begin{align}
\mu_n(\eps) &= \exp\left[\sqrt\frac{2}{n} Q^{-1}(\eps) - \frac{1}{n}\log n + O\left(\frac{1}{n}\right)\right]\nonumber\\
&= 1+\sqrt\frac{2}{n} Q^{-1}(\eps) - \frac{1}{n}\log n + O\left(\frac{1}{n}\right)\nonumber\\
\end{align}
where the last step follows from the Taylor expansion of $e^x$. 
\end{proof}
\end{theorem}
\fi 

\section{Summary}\label{sec:Summary}

\if \Summary 1
In this paper we examined the unconstrained AWGN channel setting in the finite dimension regime. We provided two new achievability bounds and extended the converse bound (sphere bound) to finite dimensional IC's. We then analyzed these bounds asymptotically in two settings. In the first setting where the NLD (which is equivalent to the rate in classic channel coding) was fixed, we evaluated the (bounds on the) error probability when the dimension $n$ grows, and provided asymptotic expansions that are significantly tighter than those in the existing error exponent analysis. In the second setting, the error probability $\eps$ is fixed, and we investigated the optimal achievable NLD for growing $n$. We showed that the optimal NLD can be tightly approximated by a closed-form expression, and the gap to the optimal NLD vanishes as the inverse of the square root of the dimension $n$. The result is analogous to the channel dispersion theorem in classical channel coding, and agrees with the interpretation of the unconstrained setting as the high-SNR limit of the power constrained AWGN channel. The approach and tools developed in this paper can be used to extend the results to more general noise models, and also to finite constellations.

\section*{Acknowledgment}
The authors are grateful to Y. Yona for assistance in providing the simulation results for low-density lattice codes in Section~\ref{sec:Comparison}.

\fi 

\if \app 1
\appendices

\section{Proof of the Bounds Equivalence}\label{app:Equivalence}

\begin{proof}[Proof of Theorem~\ref{thm:Equivalence}]
It remains to show that
\begin{align}
    n \int_0^{2r} w^{n-1}\Pr\{\bZ \in D(r,w)\} dw  = \int_0^{r^*} f_R(\rho) \rho^n d\rho.
\end{align}

\begin{lem}
For $\bZ\sim N(0,I\sigma^2)$, and any $r \geq w/2 \geq 0$,
\begin{equation}   
    \Pr\{\bZ \in D(r,w)\} = \int_{w/2}^r f_{Z}(z) \int_0^{\frac{\sqrt{r^2-z^2}}{\sigma}} f_{\chi_{n-1}}(t)dt dz,
\end{equation}
where $D(w,r)$ was defined after Eq. \eqref{eqn:AchievabilityPoltyrev}, $f_Z(z) = \frac{1}{\sqrt{2\pi\sigma^2}}e^{-z^2/(2\sigma^2)}$ is the pdf of a $N(0,\sigma^2)$ random variable, and
$f_{\chi_{n-1}}(t) = \frac{t^{n-2}e^{-t^2/2}}{2^{\frac{n-1}{2}-1}\Gamma\left(\frac{n-1}{2}\right)}$
is the pdf of a $\chi$ random variable with $n-1$ degrees of freedom.
\begin{proof}
By the spherical symmetry of the Gaussian pdf, we may assume w.l.o.g. that the hyperplane at distance $\frac{w}{2}$ is perpendicular to the $Z_1$ axis. We therefore have
\begin{align*}
    \Pr\{\bZ \in D(w,r)\}
    &= \Pr\{Z_1 > \frac{w}{2} , \|\bZ\| \leq r \} \\
    &= \Pr\left\{Z_1 > \frac{w}{2} , \sum_{i=1}^n Z_i^2 \leq r^2 \right\} \\
    &= \int_{w/2}^r f_{Z}(z) \Pr\left\{\sum_{i=1}^n Z_i^2 \leq r^2 | Z_1 = z\right\}dz\\
    &= \int_{w/2}^r f_{Z}(z) \Pr\left\{\sum_{i=2}^n Z_i^2 \leq r^2-z^2\right\}dz\\
    &= \int_{w/2}^r f_{Z}(z) \Pr\left\{\frac{1}{\sigma}\sqrt{\sum_{i=2}^n Z_i^2} \leq \frac{\sqrt{r^2-z^2}}{\sigma}\right\}dz\\
    &= \int_{w/2}^r f_{Z}(z) \int_0^{\frac{\sqrt{r^2-z^2}}{\sigma}} f_{\chi_{n-1}}(t)dt dz,
\end{align*}
where the last equality follows from the fact that a $\chi_{n-1}$ random variable is equivalent to the square root of a sum of $n-1$ independent squared standard Gaussian random variables.
\end{proof}\end{lem}

We use the result of the lemma and get
\begin{align*}
    &n \int_0^{2r} w^{n-1}\Pr\{\bZ \in D(r,w)\} dw  \\
    =& n \int_0^{2r} w^{n-1}\int_{w/2}^r f_{Z}(z) \int_0^{\frac{\sqrt{r^2-z^2}}{\sigma}} f_{\chi_{n-1}}(t)dt dz dw\\
    =& n \int_0^{2r} \int_{w/2}^r w^{n-1} f_{Z}(z) \int_0^{\frac{\sqrt{r^2-z^2}}{\sigma}} f_{\chi_{n-1}}(t)dt dz dw\\
    \overset{(a)}{=}& n \int_{0}^r \int_0^{2z}  w^{n-1} f_{Z}(z) \int_0^{\frac{\sqrt{r^2-z^2}}{\sigma}} f_{\chi_{n-1}}(t)dt dw dz\\
    =& \int_{0}^r \int_0^{\frac{\sqrt{r^2-z^2}}{\sigma}} (2z)^n f_{Z}(z)  f_{\chi_{n-1}}(t)dt dz,
\end{align*}
were $(a)$ follows from changing the order of integration. We set $u = z/\sigma$ and get
\begin{align*}
    &\int_{0}^{r/\sigma} \int_0^{\sqrt{r^2/\sigma^2-u^2}} (2\sigma u)^n \sigma f_{Z}(\sigma u)  f_{\chi_{n-1}}(t)dt du \\
    =&\int_{0}^{r/\sigma} \int_0^{\sqrt{r^2/\sigma^2-u^2}} (2\sigma u)^n \frac{1}{\sqrt{2 \pi}}e^{-u^2/2}  \frac{t^{n-2}e^{-t^2/2}}{2^{\frac{n-1}{2}-1}\Gamma\left(\frac{n-1}{2}\right)} dt du \\
    =&\frac{2^{\frac{n}{2}+1}\sigma^n}{\sqrt{\pi}\Gamma\left(\frac{n-1}{2}\right)}
    \int_{0}^{r/\sigma} \int_0^{\sqrt{r^2/\sigma^2-u^2}}  u^n t^{n-2}e^{-(u^2+t^2)/2} dt du.
\end{align*}
We switch to polar coordinates and set $u = \rho \cos \theta $ and $t = \rho \sin\theta$. The expression becomes
\begin{align}
    &\frac{2^{\frac{n}{2}+1}\sigma^n}{\sqrt{\pi}\Gamma\left(\frac{n-1}{2}\right)}
    \int_{0}^{r/\sigma} \int_0^{\pi/2}  \rho^{2n-1} \cos^n\theta\sin^{n-2}\theta e^{-\rho^2/2} d\theta d\rho\nonumber\\
    =& \frac{2^{\frac{n}{2}+1}\sigma^n}{\sqrt{\pi}\Gamma\left(\frac{n-1}{2}\right)}
    \int_{0}^{r/\sigma}  \rho^{2n-1} e^{-\rho^2/2}  d\rho \int_0^{\pi/2} \cos^n\theta\sin^{n-2}\theta d\theta
    .\label{eqn:pre-sincosIntegral}
\end{align}
It can be shown using e.g. \cite[Eqs. 18.32 and 25.9]{SchaumMath99} that
\begin{equation}
    \int_0^{\pi/2} \cos^n\theta\sin^{n-2}\theta d\theta = \frac{2^{-n}\sqrt{\pi}\Gamma\left(\frac{n-1}{2}\right)}{\Gamma\left(\frac{n}{2}\right)}.
\end{equation}
Equation \eqref{eqn:pre-sincosIntegral} now simplifies to
\begin{align}
   &\frac{2^{1-\frac{n}{2}}\sigma^n}{\Gamma\left(\frac{n}{2}\right)}
    \int_{0}^{r/\sigma}  \rho^{2n-1} e^{-\rho^2/2}  d\rho\nonumber\\
   =& \sigma^n \int_{0}^{r/\sigma}  f_{\chi_n}(\rho)\rho^n d\rho\nonumber\\
   =& \sigma^n \int_{0}^{r/\sigma}  \sigma f_R(\sigma \rho)\rho^n d\rho \nonumber\\
   =&  \int_{0}^{r} f_R(\rho')\rho'^n d\rho' \nonumber\\
   =&  \int_{0}^{r} f_R(\rho)\rho^n d\rho, \nonumber
\end{align}
where $f_{\chi_n}(\cdot)$ is the pdf of a $\chi$ random variable with $n$ degrees of freedom, and $f_R(\cdot)$ is the pdf of $\|\bZ\|$. This completes the proof of the theorem.
\end{proof}

\section{Properties of Regular IC's}\label{app:AvgVol}

\begin{proof}[Proof of Lemma \ref{lem:avgVol}]
Let $\S$ be a regular IC with a given $r_0$. Let $\Vor(a)$ denote the union of all the Voronoi cells of code points in $\cube{a}$:
\begin{equation}
    \Vor(a) \triangleq \bigcup_{s\in \S\cap \cube{a}} W(s).
\end{equation}
Since all Voronoi cells are bounded in spheres of radius $r_0$, we note the following (for $a > 2r_0$):
\begin{itemize}
  \item All the Voronoi cells of the code points in $\cube{a}$ are contained in $\cube{a+2r_0}$, and therefore $\Vor(a) \subseteq \cube{a+2r_0}$.
  \item Any point in $\cube{a - 2r_0}$ must be in a Voronoi cell of some code point. These code points cannot be outside $\cube{a}$ because the Voronoi cells are bounded in spheres of radius $r_0$, so they must lie within $\cube{a}$, and we get that $ \cube{a-2r_0} \subseteq \Vor(a)$.
\end{itemize}
It follows that
\begin{equation*}
    (a-2r_0)^n \leq |\Vor(r)| \leq (a+2r_0)^n,
\end{equation*}
or
\begin{equation}
    (a-2r_0)^n \leq \sum_{s\in \S\cap \cube{a}} v(s) \leq (a+2r_0)^n.
\end{equation}
Dividing by $a^n$ and taking the limit of $a\ra\infty$ gives
\begin{equation}
    \lim_{a\ra\infty} \frac{\sum_{s\in \S\cap \cube{a}} v(s)}{a^n} = 1.
\end{equation}

Since, by assumption, the limit of the density $\gamma(\S)$ exists, we get
\begin{align}
  \gamma(\S) &= \lim_{a\ra\infty} \frac{M(\S,a)}{a^n}\nonumber\\
   &= \lim_{a\ra\infty} \frac{M(\S,a)}{\sum_{s\in \S\cap \cube{a}} v(s)}\nonumber \\
   &= \frac{1}{\lim_{a\ra\infty}
   \EE_a[v(S)]}
   \nonumber\\
   &= \frac{1}{v(\S)}.
\end{align}
As a corollary, we get that for regular IC's the average volume $v(\S)$ exists in the limit (and not only in the $\limsup$ sense).
\end{proof}

\section{Convexity of the equivalent sphere bound}\label{app:SPBconvexity}

\begin{proof}[Proof of Lemma \ref{lem:SPBconvexity}]

Suppose $v$ is the volume of the Voronoi cell. The radius of the equivalent sphere is given by $r=v^{1/n}V_n^{-1/n}$. The equivalent sphere bound is given by
\begin{align}
    \SPB(v) &= \Pr\left\{\sum_{i=1}^n Z_i^2 \geq r^2\right\}\nonumber\\
    &= \Pr\left\{\sum_{i=1}^n (Z_i/\sigma)^2 \geq \frac{v^{2/n}}{V_n^{2/n}\sigma^2}\right\}\nonumber\\
    &\triangleq \Pr\left\{\sum_{i=1}^n (Z_i/\sigma)^2 \geq \left(C_1 \cdot v\right)^{2/n}
    \right\},
\end{align}
where $C_1$ is a constant.

We note that $\sum_{i=1}^n (Z_i/\sigma)^2$ is a sum of $n$ i.i.d. squared Gaussian RV's with zero mean and unit variance, which is exactly a $\chi^2$ distribution with $n$ degrees of freedom. We therefore get:
\begin{align}
    \SPB(v) &= \frac{1}{\Gamma(n/2)2^{n/2}}\int_{\left(C_1\cdot v\right)^{2/n}}^\infty x^{n/2-1}e^{-x/2}dx\nonumber\\
     &= C_2\int_{\left(C_1\cdot v\right)^{2/n}}^\infty x^{n/2-1}e^{-x/2}dx\nonumber\\
     &\triangleq C_2 F(C_1\cdot v),
\end{align}
where $C_2$ is a constant and $F(t) \triangleq \int_{t^{2/n}}^\infty x^{n/2-1}e^{-x/2}dx$. It can be verified by straightforward differentiation that
\begin{equation}
    \frac{\partial^2}{\partial t^2} F(t) =
    \frac{\partial^2}{\partial t^2} \int_{t^{2/n}}^\infty x^{n/2-1}e^{-x/2}dx = \frac{2}{n^2}t^{\frac{2}{n}-1}\exp\left(-\frac{1}{2}t^{2/n}\right),
\end{equation}
which is strictly positive for all $t>0$. Therefore $F(t)$ is convex, and the equivalent sphere bound $\SPB(v) = C_2 F(C_1\cdot v)$ is a convex function of $v$.\end{proof}

\section{Proof of the Regularization Lemma}\label{app:regularization}
\begin{proof}[Proof of Lemma \ref{lem:regularization}]
Our first step will be to find a hypercube $\cube{a_*}$, so that the density of the points in $\S\cap \cube{a_*}$ and the error probability of codewords in $\S\cap \cube{a_*}$ are close enough to $\gamma$ and $\eps$, respectively. We then replicate this cube in order to get a regular IC with the desired properties. The idea is similar to that used in \cite[Appendix C]{Poltyrev94_CodingWithoutRestrictions}, where it was used for expurgation purposes. As discussed in \ref{ssec:SphereBoundIC} above, we wish to avoid expurgation process that weakens the bound for finite dimensional IC's.

By the definition of $P_e(\S)$ and $\gamma(\S)$,
\begin{equation}
    \gamma(\S) = \limsup_{a \ra\infty} \frac{M(\S,a)}{a^n} = \lim_{a \ra\infty} \sup_{b>a}\frac{M(\S,b)}{b^n}
\end{equation}
\begin{equation}
    \eps = P_e(\S) = \limsup_{a \ra\infty} \frac{1}{M(\S,a)} \sum_{s\in \S\cap \cube{a}}P_e(s) = \lim_{a \ra\infty} \sup_{b>a}\frac{1}{M(\S,b)} \sum_{s\in \S\cap \cube{b}}P_e(s).
\end{equation}

Let $\tau_\gamma = \sqrt{1+\xi}$ and $\tau_{\eps}=1+\frac{\xi}{2}$.

By definition of the limit, there must exist $a_0$ large enough s.t. for every $a>a_0$, both hold:
\begin{equation}\label{eqn:sup_gamma}
    \sup_{b>a}\frac{M(\S,b)}{b^n} > \gamma \cdot \frac{1}{\tau_\gamma},
\end{equation}
and
\begin{equation}\label{eqn:sup_lambda}
    \sup_{b>a}\frac{1}{M(\S,b)} \sum_{s\in \S\cap \cube{b}}P_e(s) <\eps \cdot \tau_\eps.
\end{equation}

Define $\Delta$ s.t. $Q(\Delta/\sigma) = \eps \cdot \frac{\xi}{2}$, and define $a_\Delta$ as the solution to
\begin{equation}
    \left(\frac{a_\Delta+2\Delta}{a_\Delta}\right)^n = \sqrt{1+\xi}.
\end{equation}
Let $a_{\max} = \max\{a_0,a_\Delta\}$.
According to \eqref{eqn:sup_gamma}, there must exist $a_*>a_{\max}$ s.t.
\begin{equation}\label{eqn:bound_card_G}
    \frac{M(\S,a_*)}{a_*^n} >\gamma \cdot \frac{1}{\tau_\gamma}.
\end{equation}
By \eqref{eqn:sup_lambda} we get that
\begin{equation}
   \frac{1}{M(\S,a_*)} \sum_{s\in \S\cap \cube{a_*}}P_e(s) \leq \sup_{b>a_{\max}}\frac{1}{M(\S,b)} \sum_{s\in S\cap \cube{b}}P_e(s) <\eps \cdot \tau_\eps.
\end{equation}

Now consider the \emph{finite} constellation $G = \S\cap \cube{a_*}$. For $s\in G$, denote by $P_e^G (s)$ the error probability of $s$ when $G$ is used for transmission with Gaussian noise. Since $G\subset \S$, clearly $P_e^G(s) \leq P_e(s)$ for all $s\in G$. The average error probability for $G$ is bounded by
\begin{equation}\label{eqn:UB_lambdaG}
    P_e(G) \triangleq \frac{1}{|G|}\sum_{s\in G}P_e^G(s) \leq \frac{1}{|G|}\sum_{s\in G}P_e(s) \leq \eps \cdot \tau_\eps.
\end{equation}

We now turn to the second part - constructing an IC from the code $G$.

Define the IC $\S'$ as an infinite replication of $G$ with spacing of $2\Delta$ between every two copies as follows:
\begin{equation}\label{eqn:S'construction}
    \S' \triangleq \left\{ s+I \cdot (a_* + 2\Delta) : s\in G, I \in \Integers_n\right\},
\end{equation}
where $\Integers_n$ denotes the integer lattice of dimension $n$.

Now consider the error probability of a point $s\in \S'$ denoted by $P_e^{\S'}(s)$. This error probability equals the probability of decoding by mistake to another codeword from the same copy of $G$ or to a codeword in another copy. By the union bound, we get that
\begin{equation}
    P_e^{\S'}(s) \leq P_e^{G}(s) + Q(\Delta/\sigma).
\end{equation}
The right term follows from the fact that in order to make a mistake to a codeword in a different copy of $G$, the noise must have an amplitude of at least $\Delta$. The average error probability over $\S'$ is bounded by
\begin{align}
    P_e(\S') &\leq P_e(G)+Q(\Delta/\sigma) \leq \eps \cdot \tau_\eps + Q(\Delta/\sigma) = \eps(1+\xi)
\end{align}
as required, where the last equality follows from the definition of $\tau_\eps$ and $\Delta$.

The density of points in the new IC enclosed within a cube of edge size $a_*+2\Delta$ is given by $|G|(a_*+2\Delta)^{-n}$. Define $\tilde a_k = (a_*+2\Delta)(2k-1) $ for any integer $k$. Note that for any $k>0$, $\cube{\tilde a_k}$ contains exactly $(2k-1)^n$ copies of $G$, and therefore
\begin{equation}
    \frac{M(\S',\tilde a_k)}{\tilde a_k^n} = \frac{|G|(2k-1)^n}{\tilde a_k^n} = \frac{|G|}{(a_*+2\Delta)^n}.
\end{equation}

For any $a>0$, let $k^*$ be the minimal integer $k$ s.t. $\tilde a_k \geq a$. Clearly,
\begin{equation}
    \tilde a_{k^*-1} = \tilde a_{k^*}-(a_*+2\Delta) <  a \leq \tilde a_{k^*}.
\end{equation}
Therefore
\begin{equation}
\frac{M(\S',\tilde a_{k^*-1})}{a^n}<    \frac{M(\S',a)}{a^n} \leq \frac{M(\S',\tilde a_{k^*})}{a^n},
\end{equation}
and
\begin{equation}\label{eqn:gammaS'beforeLimit}
\frac{|G|}{(a_*+2\Delta)^n}\frac{\tilde a_{k^*-1}^n}{a^n}
<\frac{M(\S',a)}{a^n} \leq
\frac{|G|}{(a_*+2\Delta)^n}\frac{\tilde a_{k^*}^n}{a^n}.
\end{equation}
By taking the limit $a \ra \infty$ of \eqref{eqn:gammaS'beforeLimit}, we get that the limit exists and is given by
\begin{equation}\label{eqn:gammaS'limit}
    \gamma(\S') = \lim_{a\ra\infty} \frac{M(\S',a)}{a^n} = \frac{|G|}{(a_*+2\Delta)^n}.
\end{equation}


It follows that
\begin{align}
  \gamma(\S') &= \frac{|G|}{(a_*+2\Delta)^n} \nonumber\\
  &= \frac{|G|}{a_*^n}\frac{a_*^n}{(a_*+2\Delta)^n} \nonumber\\
   &\overset{(a)}{\geq} \gamma(\S) \frac{1}{\tau_\gamma}\left(\frac{a_*}{a_*+2\Delta}\right)^n  \nonumber\\
   &\overset{(b)}{\geq} \gamma(\S) \frac{1}{1+\xi}.
\end{align}
where $(a)$ follows from \eqref{eqn:bound_card_G} and $(b)$ follows from the definitions of $\tau_\gamma,a_\Delta$ and from the fact that $a_\Delta \leq a_*$.

It remains to show that the resulting IC $\S'$ is regular, i.e. that all the Voronoi cells can be bounded in a sphere with some fixed radius $r_0$. The fact that the average density is achieved in the limit (ant not only in the $\limsup$ sense) was already established in \eqref{eqn:gammaS'limit}.

Let $s$ be an arbitrary point in $\S'$. By construction (see \eqref{eqn:S'construction}), the points
\begin{equation*}
    \{s \pm (a_*+2\Delta) e_i|i=1,...,n\}
\end{equation*}
 are also in $\S'$ (where $e_i$ denotes the vector of $1$ in the $i$-th coordinate, and the rest are zeros). We therefore conclude that the Voronoi cell $W(s)$ is contained in the hypercube $s + \cube{a_*+2\Delta}$, and is clearly bounded within a sphere of radius $r_0 \triangleq \sqrt n (a_*+2\Delta)$.
\end{proof}

\section{Proof of Integral Bounding Lemmas}\label{app:IntegralBounds}

\begin{proof}[Proof of Lemma~\ref{lem:SpherePackingIntegralBounds}]
Define
\begin{align*}
F(\rho)\triangleq \log\left[ \rho^{\frac{n}{2}-1}e^{-n\rho/2}\right] = \left(\frac{n}{2}-1\right)\log \rho - \frac{n\rho}{2},
\end{align*}
so that $\rho^{\frac{n}{2}-1}e^{-n\rho/2} = \exp[F(\rho)]$. Let $F_1(\rho)$ and $F_2(\rho)$ be the first and second order Taylor series of $F(\rho)$ around $\rho=x$, respectively, i.e.
\begin{align}
    F_1(\rho) = \alpha + \beta(\rho-x);\quad F_2(\rho) = \alpha + \beta(\rho-x)-\tau^2(\rho-x)^2
\end{align}
where
\begin{align*}
 \alpha &\triangleq \left(\frac{n}{2}-1\right)\log x - \frac{n x}{2};\\
 \beta &\triangleq \frac{\frac{n}{2}-1}{x} - \frac{n}{2};\\
 \tau &\triangleq \sqrt{\frac{\frac{n}{2}-1}{2x^2}}.
\end{align*}
We note that for any $\xi>0$,
\begin{equation}\label{eqn:logBounds}
    \xi  - \frac{\xi^2}{2}\leq \log (1+\xi) \leq \xi.
\end{equation}
It follows that for all $\rho>x$,
\begin{equation}
    F_2(\rho) \leq F(\rho) \leq F_1(\rho),
\end{equation}
and the integral is bounded from above and below by
\begin{equation}
    \int_x^{\infty}e^{F_2(\rho)}d\rho \leq  \int_x^{\infty} \rho^{\frac{n}{2}-1}e^{-n\rho/2}d\rho \leq  \int_x^{\infty}e^{F_1(\rho)}d\rho.
\end{equation}

To prove the upper bound \eqref{eqn:SpherePackingIntegralUB} we continue with
\begin{align*}
    \int_x^{\infty}e^{F(\rho)}d\rho
    &\leq \int_x^{\infty}e^{F_1(\rho)}d\rho\\
    &= e^\alpha \int_x^{\infty}\exp\left[\beta(\rho-x)\right]d\rho\\
    &= \frac{e^\alpha}{-\beta},
\end{align*}
where the last equality follows from the assumption $x>1-\frac{2}{n}$. Plugging the values for $\alpha$ and $\beta$ yields \eqref{eqn:SpherePackingIntegralUB}.

To prove the lower bound \eqref{eqn:SpherePackingIntegralUB} we write
\begin{align*}
    \int_x^{\infty}e^{F(\rho)}d\rho
    &\geq \int_x^{\infty}e^{F_2(\rho)}d\rho\\
    &= \int_x^{\infty}\exp\left[\alpha + \beta(\rho-x) -\tau^2(\rho-x)^2\right]d\rho\\
    &= \int_0^{\infty}\exp\left[\alpha + \beta\rho -\tau^2\rho^2\right]d\rho\\
    &= \exp\left(\alpha +\frac{\beta^2}{4\tau^2}\right)\int_0^{\infty}\exp\left[-\left(\tau\rho - \frac{\beta}{2\tau}\right)^2\right]d\rho\\
    &= \exp\left(\alpha +\frac{\beta^2}{4\tau^2}\right)
    \frac{\sqrt{\pi}}{\tau}
    \int_0^{\infty}
    \frac{\tau}{\sqrt{\pi}}
    \exp\left[-\tau^2\left(\rho - \frac{\beta}{2\tau^2}\right)^2\right]d\rho\\
    &= \exp\left(\alpha +\frac{\beta^2}{4\tau^2}\right)
    \frac{\sqrt{\pi}}{\tau}
    Q\left(\frac{-\beta}{\tau\sqrt{2}}\right).
\end{align*}
Plugging back the values for $\alpha, \beta$ and $\tau$ leads to \eqref{eqn:SpherePackingIntegralLB1}.

We continue with a well known lower bound for the $Q$ function:
\begin{equation}
    Q(z) \geq \frac{1}{\sqrt{2\pi}z}e^{-z^2/2}\left(\frac{1}{1+z^{-2}}\right) \quad \forall z >0.\label{eqn:QfuncLowerBound}
\end{equation}
Recalling the definition of $\Upsilon$, we write
\begin{align*}
    \int_x^{\infty} \rho^{\frac{n}{2}-1}e^{-n\rho/2}d\rho
    &\geq
    2x^{\frac{n}{2}}e^{ - \frac{n x}{2}}
    \exp\left[\Upsilon^2/2\right]
    \sqrt{\frac{\pi}{n-2}}
    Q\left(\Upsilon\right)\\
    &\geq
    \frac{2x^{\frac{n}{2}}e^{ - \frac{n x}{2}}}{n(x-1+\frac{2}{n})}
    \left(\frac{1}{1+\Upsilon^{-2}}\right),
\end{align*}
to arrive at \eqref{eqn:SpherePackingIntegralLB2}. Eq. \eqref{eqn:SpherePackingIntegralLB3} follows immediately since $1-\xi \leq \frac{1}{1+\xi}$ for all $\xi \in \Reals$.
\end{proof}

\bigskip

\begin{proof}[Proof of Lemma~\ref{lem:MaxLikeIntegralBounds}]
We rewrite the integrand as $e^{G(\rho)}$ where $G(\rho) \triangleq -n \rho/2 + (n-1) \log \rho$. Since $G(\rho)$ is
concave, it is upper bounded its first order Taylor approximation at any point. We choose the tangent at $\rho=x$. We denote by $G_1(\rho)$ the first order Taylor approximation at that point, and get
\begin{equation}
    G(\rho) \leq G_1(\rho) \triangleq G(x) + G'(x) (\rho-x),
\end{equation}
where $G'(\rho) = \frac{\partial G(\rho)}{\partial \rho} = -\frac{n}{2} + \frac{n-1}{\rho}$. It follows that
\begin{align}
    G_1(\rho)
     &= (n-1) (\log x-1)  +\left(-\frac{n}{2} + \frac{n-1}{x}\right)\rho. \nonumber
\end{align}

Since $G(\rho) \leq G_1(\rho)$ for all $\rho$, we have
\begin{align}
     \int_0^x e^{-n \rho/2}  \rho^{n-1} d\rho &=  \int_0^x e^{G(\rho)} d\rho\nonumber\\
     &\leq \int_0^x e^{G_1(\rho)} d\rho\nonumber\\
     &= x^{n-1} e^{-(n-1)}\int_0^x \exp\left[\left(-\frac{n}{2} + \frac{n-1}{x}\right)\rho\right] d\rho\nonumber\\
     &= \frac{2x^n}{n\left(2-x-\tfrac{2}{n}\right)}\left(e^{-\tfrac{n}{2}x} - e^{-(n-1)}\right),
\end{align}
which gives \eqref{eqn:MLIntegralUpperBound}.

Some extra effort is required in order to prove the lower bound \eqref{eqn:MLIntegralLowerBoundQfunc}. We first switch variables and get
\begin{align}\label{eqn:switchVars}
    \int_0^{x} e^{-n \rho/2}  \rho^{n-1} d\rho
    &=\int_{1/x}^{\infty} e^{-\frac{n}{2u}} u^{-n-1} du\\
    &=\int_{1/x}^{\infty} \exp\left(-\frac{n}{2u} -(n+1)\log u\right) du.
\end{align}
We lower bound the exponent as follows:
\begin{align*}
    -\frac{n}{2u} -(n+1)\log u
    &= -\frac{n}{2u} +(n+1)(\log x - \log (u x))\\
    &= -\frac{n}{2u} +(n+1)(\log x - \log (1 + u x-1))\\
    &\overset{(a)}{\geq} -\frac{n}{2u} +(n+1)(\log x - (u x-1))\\
    &\overset{(b)}{\geq} -\frac{nx}{2}(u^2x^2-3ux+3) +(n+1)(\log x - (u x-1))\\
    &= -\frac{nx}{2}[x^2(u-1/x)^2 -x(u-1/x)+1] +(n+1)(\log x - x(u -1/x)).
\end{align*}
$(a)$ follows from the fact that $\log(1+\xi) \leq \xi$ for all $\xi \in \Reals$. $(b)$ follows from the fact that $\frac{1}{\xi} \leq \xi^2 -3\xi+3$ for all $\xi>1$ (which follows from the fact that $(\xi-1)^3\geq 0$).

So far the integral $\int_0^{x} e^{-n \rho/2}  \rho^{n-1} d\rho$ is lower bounded by \begin{equation}
    \int_{1/x}^\infty \exp(\alpha + \beta (u-1/x) - \tau^2 (u-1/x)^2)du,
\end{equation}
where
\begin{align*}
    \alpha &\triangleq (n+1)\log x-\frac{nx}{2};\\
    \beta &\triangleq  \frac{nx^2}{2}-(n+1)x;\\
    \tau &\triangleq \sqrt{\frac{nx^3}{2}}.
\end{align*}

Following the same steps as in the proof of \eqref{eqn:SphereBoundLowerBoundQfunc} in Lemma \ref{lem:SpherePackingIntegralBounds} gives
\begin{equation}
    \int_{1/x}^\infty \exp(\alpha + \beta (u-1/x) - \tau^2 (u-1/x)^2)
    \geq
    \exp\left(\alpha +\frac{\beta^2}{4\tau^2}\right)
    \frac{\sqrt{\pi}}{\tau}
    Q\left(\frac{-\beta}{\tau\sqrt{2}}\right).
\end{equation}
Plugging the values for $\alpha, \beta$ and $\tau$ yields \eqref{eqn:MLIntegralLowerBoundQfunc}. \eqref{eqn:MLIntegralLowerBoundAnalytic} follows by applying the lower bound \eqref{eqn:QfuncLowerBound} on the $Q$ function.
\end{proof}

\section{Approximating $V_n$}\label{app:Vn}
Here we derive the required approximations for $V_n$, used in Sections~\ref{sec:Properties} and \ref{sec:NormalApprox}.

We first derive \eqref{eqn:VnApproxMult}.

The volume of a hypersphere of unit radius $V_n$ is given by $\frac{\pi^{n/2}}{\Gamma(n/2+1)}$ (see, e.g. \cite[p. 9]{ConwaySloane1993}).

We use the Stirling approximation for the Gamma function for $z\in \Reals$ (see, e.g. \cite[Sec. 5.11]{DLMF}).
\begin{align}
    \Gamma(z+1) = z\Gamma(z) &= z\sqrt \frac{2 \pi}{z} \left(\frac{z}{e}\right)^z \left(1+O\left(\frac{1}{z}\right)\right)\nonumber\\
      &= \sqrt {2 \pi z} \left(\frac{z}{e}\right)^z \left(1+O\left(\frac{1}{z}\right)\right).
\end{align}
$V_n$ becomes
\begin{align}
    V_n &= \frac{\pi^{n/2}}{\Gamma(n/2+1)} = \left(\frac{2\pi e}{n}\right)^{n/2}\frac{1}{\sqrt{n\pi}}\left(1+O\left(\frac{1}{n}\right)\right).
\end{align}

Eq. \eqref{eqn:VnApprox} follows by taking $\frac{1}{n}\log(\cdot)$ and the Taylor expansion of $\log(1+x)$ around $x=0$:

\begin{align}
    \frac{1}{n}\log V_n &= \frac{1}{2}\log \pi - \frac{1}{2n}\log{n} - \frac{1}{2} \log \frac{n}{2e} +  O\left(\frac{1}{n}\right)\nonumber\\
    &=\frac{1}{2}\log \frac{2\pi e}{n} - \frac{1}{2n}\log{n} +  O\left(\frac{1}{n}\right).
\end{align}

\section{Evaluating the ML Bound at $\NLD_{cr}$}\label{app:delta_cr}
\begin{proof}[Proof of Theorem \ref{thm:MaxLikeAnalysisAtDeltaCr}]
We start as in the proof of Theorem \ref{thm:AchievabilityMaxLike} to have
\begin{equation}
    e^{n\NLD_{cr}} V_n \int_0^{r^*} f_R(r) r^n dr
    = \frac{n}{2} e^{n\NLD} V_n^2\sigma^n (2\pi)^{-\frac{n}{2}}n^n \int_0^{\rho^*} e^{-n\rho/2} \rho^{n-1} d\rho.
\end{equation}

We evaluate the integral in two parts:
\begin{equation}\label{eqn:TwoPartIntegral}
    \int_0^{\rho^*} e^{-n\rho/2} \rho^{n-1} d\rho = \int_0^{2} e^{-n\rho/2} \rho^{n-1} d\rho + \int_2^{\rho^*} e^{-n\rho/2} \rho^{n-1} d\rho.
\end{equation}
The term $\int_0^{2} e^{-n\rho/2} \rho^{n-1} d\rho$ can be evaluated by the Laplace method, as in the proof of Lemma \ref{lem:MaxLikeIntegralLaplace}. The difference is that the exponent is minimized with zero first derivative at the boundary point $\rho = 2$, which causes the integral to be evaluated to half the value of the integral in Lemma \ref{lem:MaxLikeIntegralLaplace}, i.e.
\begin{equation}\label{eqn:IntegralAtDeltaCrPart1}
    \int_0^{2} e^{-n\rho/2} \rho^{n-1} d\rho =
    \sqrt{\frac{\pi}{2n}}e^{-n}2^n \left(1+O\left(\tfrac{1}{n}\right)\right).
\end{equation}
The second term in \eqref{eqn:TwoPartIntegral} requires some extra effort.
We first upper bound it as follows:
\begin{align*}
    \int_2^{\rho^*} e^{-n\rho/2} \rho^{n-1} d\rho
    &= \int_2^{\rho^*} \frac{1}{\rho}e^{-n\rho/2} \rho^{n} d\rho\\
    &\overset{(a)}{\leq} \int_2^{\rho^*} \frac{1}{2}e^{-n\rho/2} \rho^{n} d\rho\\
    &\overset{(b)}{\leq} \int_2^{\rho^*} \frac{1}{2}e^{-n} 2^{n} d\rho\\
    &=\frac{1}{2}e^{-n} 2^{n}(\rho^*-2),
\end{align*}
where $(a)$ follows since in the integration interval, $\rho > 2$. $(b)$ follows since $e^{-n\rho/2} \rho^{n}$ is maximized at $\rho=2$. With \eqref{eqn:rhoStarApprox2} we have
\begin{align*}
    \int_2^{\rho^*} e^{-n\rho/2} \rho^{n-1} d\rho
    & \leq \frac{1}{2}e^{-n} 2^{n}(\rho^*-2) \\
    & = \frac{1}{2}e^{-n} 2^{n}\left(\tfrac{2}{n}\log(n \pi) + O\left(\tfrac{\log^2n}{n^2}\right)\right) \\
    & = e^{-n} 2^{n}\frac{\log(n \pi)}{n}\left(1 + O\left(\tfrac{\log n}{n}\right)\right).
\end{align*}

The integral can also be lower bounded as follows:
\begin{align*}
    \int_2^{\rho^*} e^{-n\rho/2} \rho^{n-1} d\rho
    & \overset{(a)}{\geq} \frac{1}{\rho^*}\int_2^{\rho^*} e^{-n\rho/2} \rho^{n} d\rho\\
    & \overset{(b)}{\geq} \frac{1}{\rho^*}\int_2^{\rho^*} e^{n\log\tfrac{2}{e} -\tfrac{n}{8}(\rho-2)^2} d\rho\\
    &= \frac{1}{\rho^*} 2^ne^{-n}\int_2^{\rho^*} e^{-\tfrac{n}{8}(\rho-2)^2} d\rho\\
    &= \frac{1}{\rho^*} 2^ne^{-n}\int_0^{\rho^*-2} e^{-\tfrac{n}{8}\rho^2} d\rho\\
    &= \frac{1}{\rho^*} 2^ne^{-n} \sqrt{\frac{8\pi}{n}}\left(Q(0)-Q\left(\tfrac{\rho^*-2}{\sqrt{4/n}}\right)\right)\\
    &\overset{(c)}{=} \frac{1}{\rho^*} 2^ne^{-n} \sqrt{\frac{8\pi}{n}}\left(\tfrac{1}{2}-\left(\tfrac{1}{2} - \tfrac{1}{\sqrt{2\pi}}\tfrac{\rho^*-2}{\sqrt{4/n}}+O\left(\tfrac{\log^2n}{n}\right)\right)\right)\\
    &= \frac{1}{\rho^*} 2^ne^{-n}
    \left(\rho^*-2 +O\left(\tfrac{\log^2n}{n\sqrt n}\right)\right)\\
    &= 2^ne^{-n}
    \frac{\log(n\pi)}{n}\left(1 +O\left(\tfrac{\log n}{\sqrt n}\right)\right).
\end{align*}
$(a)$ follows since $\rho\leq \rho^*$. $(b)$ follows from the fact that $n\rho/2 +n\log\rho \geq n\log\tfrac{2}{e} -\tfrac{n}{8}(\rho-2)^2$ for all $\rho>2$ (which follows from \eqref{eqn:logBounds}). $(c)$ follows from the Taylor expansion $Q(\xi) = \tfrac{1}{2} - \tfrac{\xi}{\sqrt{2\pi}} + O(\xi^2)$ and since $\rho^*-2 = O(\tfrac{\log n}{n})$.

In total we get
\begin{equation*}
    \int_2^{\rho^*} e^{-n\rho/2} \rho^{n-1} d\rho = 2^ne^{-n}
    \frac{\log(n\pi)}{n}\left(1 +O\left(\tfrac{\log n}{\sqrt n}\right)\right).
\end{equation*}
Combined with \eqref{eqn:IntegralAtDeltaCrPart1} we have
\begin{align*}
    \int_0^{\rho^*} e^{-n\rho/2} \rho^{n-1} d\rho
    &= \int_0^{2} e^{-n\rho/2} \rho^{n-1} d\rho + \int_2^{\rho^*} e^{-n\rho/2} \rho^{n-1} d\rho \\
    &= 2^ne^{-n}\left[\sqrt{\frac{\pi}{2n}}+\frac{\log(n\pi)}{n}\right]
    \left(1+O\left(\tfrac{\log^2n}{n}\right)\right).
\end{align*}
The approximation \eqref{eqn:VnApproxMult} for $V_n$ finally yields
\begin{equation}
    e^{n\NLD_{cr}} V_n \int_0^{r^*} f_R(r) r^n dr =  e^{-n \E_{r}(\NLD_{cr})}\frac{1}{2\pi}
    \left[\sqrt{\frac{\pi}{2n}}+\frac{\log(n \pi)}{n}\right]
    \left(1+O\left(\tfrac{\log^2n}{n}\right)\right),
\end{equation}
and the proof is completed by adding the asymptotic form \eqref{eqn:SphereBoundPreciseAsymptotics} of the sphere bound at $\NLD = \NLD_{cr}$.
\end{proof}

\section{ }\label{app:ProofAsymptoticsTypicality}
\begin{proof}[Proof of Theorem \ref{thm:TypicalityPreciseAsymptotics}]
The typicality bound is given by
\begin{equation}
    P_e(n,\NLD) \leq  \gamma V_n r^n + \Pr\left\{ \|\bZ\| > r \right\},
\end{equation}
where $r = \sigma\sqrt{n (1 + 2\NLD^* -2\NLD)}$. The rightmost term can be written as (see the proof of Theorem \ref{thm:SphereBoundAnalysis}):
\begin{equation}
    \Pr\left\{ \|\bZ\| > r \right\} = \frac{(n/2)^{n/2}}{\Gamma\left[\frac{n}{2}\right]} \int_{1+2(\NLD^*-\NLD)}^\infty \rho^{\frac{n}{2}-1} e^{-\frac{n}{2}\rho}d\rho.
\end{equation}
The above integral can be bounded according to Lemma \ref{lem:SpherePackingIntegralBounds} by
\begin{equation}
    \frac{2x^{\frac{n}{2}}e^{ - \frac{n x}{2}}}{n(x-1+\frac{2}{n})}\frac{1}{1+\Upsilon^{-2}}
    \leq \int_x^\infty \rho^{\frac{n}{2}-1} e^{-\frac{n}{2}\rho}d\rho
    \leq \frac{2x^{\frac{n}{2}}e^{ - \frac{n x}{2}}}{n(x-1+\frac{2}{n})}
\end{equation}
where $x=1 + 2(\NLD^* -\NLD)$ and $\Upsilon \triangleq \frac{n(x-1+\frac{2}{n})}{\sqrt{2(n-2)}} = \Theta(\sqrt n)$.
Eq. \eqref{eqn:TypicalityPreciseAsymptotics} then follows using the approximation \eqref{eqn:VnApproxMult} for $V_n$.
\end{proof}

\section{ }\label{app:ProofAsymptoticsMLPoltyrev}
\begin{proof}[Proof of Theorem~\ref{thm:AsymptoticsMLPoltyrev}]
We first prove \eqref{eqn:PoltyrevBoundAnalysisAboveDeltaCrTight} (and \eqref{eqn:PoltyrevBoundAnalysisAboveDeltaCrLoose} follows immediately).
Let $\tilde\rho = e^{2(\NLD^*-\NLD)}$.
The ML bounds for $r=\sqrt n \sigma e^{\NLD^*-\NLD}$ can be written as (see \eqref{eqn:SPBintegral} and \eqref{eqn:preMLIntegralBounds}):
\begin{align*}
    P_e^{MLB} &= \frac{n}{2}e^{n(\NLD^*+\NLD)}V_n^2e^{n/2}\sigma^{2n}n^n\int_0^{\tilde\rho} e^{-n\rho/2}\rho^{n-1}d\rho + \frac{2^{-n/2}n^{n/2}}{\Gamma(n/2)}\int_{\tilde\rho}^\infty e^{-n\rho/2}\rho^{n/2-1}d\rho.
\end{align*}

Using Lemma~\ref{lem:SpherePackingIntegralBounds} we get that
\begin{equation}\label{eqn:SphereBoundIntegralPoltyrev}
    \int_{\tilde\rho}^\infty e^{-n\rho/2}\rho^{n/2-1}d\rho =
    \frac{2\tilde\rho^{n/2}e^{-n\tilde\rho/2}}{n(\tilde\rho-1)}
    \left(1+O\left(\frac{1}{n}\right)\right),
\end{equation}
and using Lemma~\ref{lem:MaxLikeIntegralBounds} we get that
\begin{equation}
    \int_0^{\tilde\rho} e^{-n\rho/2}\rho^{n-1}d\rho =
    \frac{2\tilde\rho^{n/2}e^{-n\tilde\rho/2}}{n(2-\tilde\rho)}
    \left(1+O\left(\frac{1}{n}\right)\right).
\end{equation}
\eqref{eqn:PoltyrevBoundAnalysisAboveDeltaCrTight} then follows by simple algebra.

\bigskip

To show \eqref{eqn:PoltyrevBoundAnalysisBelowDeltaCr}, repeat the proof of Theorem~\ref{thm:MaxLikeAnalysisBelowDeltaCr} with $\tilde\rho$ instead of $\rho^*$.

\bigskip

To show \eqref{eqn:PoltyrevBoundAnalysisAtDeltaCr}, repeat the proof of Theorem~\ref{thm:MaxLikeAnalysisAtDeltaCr} with $\tilde\rho$ instead of $\rho^*$. Here $\tilde\rho = 2$, therefore there is no need to split the integral into two parts as in \eqref{eqn:TwoPartIntegral}. Therefore the term $e^{n\NLD} V_n \int_0^r f_R(\tilde r) \tilde r^n d\tilde r$ contributes the $\frac{1}{\sqrt 8}$ part of the expression. The contribution of the sphere bound term (the term $\Pr\{\|\bZ\|>r\}$) is approximated as in \eqref{eqn:SphereBoundIntegralPoltyrev}. Note that the result here is different than the $\reff$ case, where the contributions of the two terms in the bound are of different order (see Eq. \eqref{eqn:MaxLikeAnalysisAtDeltaCr}).
\end{proof}

\section{Central Limit Theorem and the Berry-Esseen Theorem}\label{app:CLT}

By the central limit theorem (CLT), A normalized sum of $n$ independent random variables converges (in distribution) to a Gaussian random variable. The Berry-Esseen theorem shows the speed of the convergence (see \cite[Ch. XVI.5]{Feller_1971}). We write here the version for i.i.d. random variables, which is sufficient for this paper.

\begin{theorem}[Berry-Esseen for i.i.d. RV's \cite{Feller_1971}]

Let $\{Y_i\}_{i=1}^n$ be i.i.d. random variables with zero mean and unit variance. Let $T \triangleq \E[|Y_i|^3]$ and assume it is finite. Let $S_n \triangleq \frac{1}{\sqrt{n}}\sum_{i=1}^{n}Y_i$ be the normalized sum. Note that $S_n$ also has zero mean and unit variance.

Then for all $\alpha \in \Reals$ and for all $n\in \Naturals$,
\begin{equation}
    |\Pr\{S_n \geq \alpha\} - Q(\alpha)| \leq \frac{6T}{\sqrt{n}}.
\end{equation}
\end{theorem}

\fi 
\bibliographystyle{IEEEtran}
\bibliography{Master}

\begin{thebibliography}{10}
\providecommand{\url}[1]{#1}
\csname url@samestyle\endcsname
\providecommand{\newblock}{\relax}
\providecommand{\bibinfo}[2]{#2}
\providecommand{\BIBentrySTDinterwordspacing}{\spaceskip=0pt\relax}
\providecommand{\BIBentryALTinterwordstretchfactor}{4}
\providecommand{\BIBentryALTinterwordspacing}{\spaceskip=\fontdimen2\font plus
\BIBentryALTinterwordstretchfactor\fontdimen3\font minus
  \fontdimen4\font\relax}
\providecommand{\BIBforeignlanguage}[2]{{%
\expandafter\ifx\csname l@#1\endcsname\relax
\typeout{** WARNING: IEEEtran.bst: No hyphenation pattern has been}%
\typeout{** loaded for the language `#1'. Using the pattern for}%
\typeout{** the default language instead.}%
\else
\language=\csname l@#1\endcsname
\fi
#2}}
\providecommand{\BIBdecl}{\relax}
\BIBdecl

\bibitem{ForneyUngerboeck98}
G.~D. {Forney Jr.} and G.~Ungerboeck, ``Modulation and coding for linear
  {Gaussian} channels,'' \emph{IEEE Trans. on Information Theory}, vol.~44,
  no.~6, pp. 2384--2415, 1998.

\bibitem{ErezZamirAWGN}
U.~Erez and R.~Zamir, ``Achieving 1/2 log(1+{SNR}) over the additive white
  {Gaussian} noise channel with lattice encoding and decoding,'' \emph{IEEE
  Trans. on Information Theory}, vol.~50, pp. 2293--2314, Oct. 2004.

\bibitem{Poltyrev94_CodingWithoutRestrictions}
G.~Poltyrev, ``On coding without restrictions for the {AWGN} channel,''
  \emph{IEEE Trans. on Information Theory}, vol.~40, no.~2, pp. 409--417, 1994.

\bibitem{GallagerInfoTheoryBook}
R.~G. Gallager, \emph{Information Theory and Reliable Communication}.\hskip 1em
  plus 0.5em minus 0.4em\relax New York, NY, USA: John Wiley \& Sons, Inc.,
  1968.

\bibitem{Shannon59Gaussian}
C.~E. Shannon, ``Probability of error for optimal codes in a gaussian
  channel,'' \emph{The Bell System technical journal}, vol.~38, pp. 611--656,
  1959.

\bibitem{PolyanskiyPVFiniteLength10}
Y.~Polyanskiy, H.~Poor, and S.~Verd\'u, ``Channel coding rate in the finite
  blocklength regime,'' \emph{IEEE Trans. on Information Theory}, vol.~56,
  no.~5, pp. 2307 --2359, May 2010.

\bibitem{Strassen62_Asymptotische}
V.~Strassen, ``Asymptotische absch\"{a}tzungen in shannon's
  informationstheorie,'' \emph{Trans. Third Prague Conf. Information Theory,
  1962, Czechoslovak Academy of Sciences}, pp. 689--723.

\bibitem{PolyanskiyPV09_GaussianDispersion}
Y.~Polyanskiy, V.~Poor, and S.~Verd\'u, ``Dispersion of {Gaussian} channels,''
  in \emph{Proc. IEEE International Symposium on Information Theory}, 2009, pp.
  2204--2208.

\bibitem{GallagerLDPC}
R.~G. Gallager, \emph{Low-Density Parity-Check Codes}.\hskip 1em plus 0.5em
  minus 0.4em\relax Cambridge, MA, USA: The M.I.T. Press, 1963.

\bibitem{DivsalarBounds}
D.~Divsalar, ``A simple tight bound on error probability of block codes with
  application to turbo codes,'' {JPL}, TMO Progr. Rep., pp. 42–139, Nov. 1999.

\bibitem{CoverThomas_InfoTheoryBook}
T.~M. Cover and J.~A. Thomas, \emph{Elements of Information Theory}.\hskip 1em
  plus 0.5em minus 0.4em\relax John Wiley \& sons, 1991.

\bibitem{ConwaySloane1993}
J.~H. Conway and N.~J.~A. Sloane, \emph{Sphere packings, lattices and groups},
  ser. Grundlehren der math. Wissenschaften.\hskip 1em plus 0.5em minus
  0.4em\relax Springer, 1993, vol. 290.

\bibitem{HlawkaGeometric1991}
E.~Hlawka, J.~Shoi{\ss}engeier, and R.~Taschner, \emph{Geometric and Analytic
  Numer Theory}.\hskip 1em plus 0.5em minus 0.4em\relax Springer-Verlang, 1991.

\bibitem{Lekkerkerker87}
P.~M. Gruber and C.~G. Lekkerkerker, \emph{Geometry of Numbers}.\hskip 1em plus
  0.5em minus 0.4em\relax Amsterdam: North-Holland, 1987.

\bibitem{TarokhVardyZeger99_universal}
V.~Tarokh, A.~Vardy, and K.~Zeger, ``Universal bound on the performance of
  lattice codes,'' \emph{Information Theory, IEEE Transactions on}, vol.~45,
  no.~2, pp. 670 --681, mar. 1999.

\bibitem{ForneySphereBound00}
G.~D.~F. Jr., M.~D. Trott, and S.-Y. Chung, ``Sphere-bound-achieving coset
  codes and multilevel coset codes,'' \emph{IEEE Transactions on Information
  Theory}, vol.~46, no.~3, pp. 820--850, 2000.

\bibitem{ZamirLatticesEverywhere09}
R.~Zamir, ``Lattices are everywhere,'' in \emph{4th Annual Workshop on
  Information Theory and its Applications}, UCSD, (La Jolla, CA), 2009.

\bibitem{ErezLitsynZamirLatticesGoodForEverything}
U.~Erez, S.~Litsyn, and R.~Zamir, ``Lattices which are good for (almost)
  everything,'' \emph{Information Theory, IEEE Transactions on}, vol.~51,
  no.~10, pp. 3401 -- 3416, oct. 2005.

\bibitem{WozencraftJacobs}
J.~M. Wozencraft and I.~M. Jacobs, \emph{Principles of Communication
  Engineering}.\hskip 1em plus 0.5em minus 0.4em\relax Prospect Heights, IL,
  USA: Waveland Press Inc., 1990, originally Published 1965 by Wiley.

\bibitem{BarndorffNielsenCox_Asymptotic}
O.~E. Barndorff-Nielsen and D.~R. Cox, \emph{Asymptotic Techniques for Use in
  Statistics}.\hskip 1em plus 0.5em minus 0.4em\relax New York: Chapman and
  Hall, 1989.

\bibitem{CsiszarKorner81}
I.~Csisz\'ar and J.~Korner, \emph{Information Theory - Coding Theorems for
  Discrete Memoryless Systems}.\hskip 1em plus 0.5em minus 0.4em\relax New
  York: Academic Press, 1981.

\bibitem{IngberFederPBICM_full}
A.~Ingber and M.~Feder, ``Parallel bit-interleaved coded modulation,''
  \emph{Available on arxiv.org}.

\bibitem{LDLC_IT2008}
N.~Sommer, M.~Feder, and O.~Shalvi, ``Low-density lattice codes,'' \emph{IEEE
  Trans. on Information Theory}, vol.~54, no.~4, pp. 1561--1585, 2008.

\bibitem{AgrawalVardy2000GMD}
D.~Agrawal and A.~Vardy, ``Generalized minimum distance decoding in euclidean
  space: Performance analysis,'' \emph{IEEE Trans. on Information Theory},
  vol.~46, pp. 60--83, 2000.

\bibitem{YairISIT2009}
Y.~Yona and M.~Feder, ``Efficient parametric decoder of low density lattice
  codes,'' in \emph{Proc. IEEE International Symposium on Information Theory},
  2009, pp. 744--748.

\bibitem{SchaumMath99}
M.~R. Spiegel, \emph{Mathematical Handbook of Formulas and Tables}.\hskip 1em
  plus 0.5em minus 0.4em\relax New York: McGraw-Hill, 1999.

\bibitem{DLMF}
N.~I. of~Standards and Technology, ``Digital library of mathematical
  functions,'' \url{http://dlmf.nist.gov}, May 2010.

\bibitem{Feller_1971}
W.~Feller, \emph{An Introduction to Probability Theory and Its Applications,
  Volume 2 (2nd Edition)}.\hskip 1em plus 0.5em minus 0.4em\relax John Wiley \&
  Sons, 1971.

\end{thebibliography}

\end{document}